\documentclass[11pt]{article}

\usepackage[T1]{fontenc}
\usepackage[utf8]{inputenc}
\usepackage{mathtools}
\usepackage[sb]{libertinus}  
\usepackage[scaled=0.95]{inconsolata}
\usepackage{libertinust1math}  
\usepackage[
    cal=pxtx,
    calscaled=0.95,
    scr=boondoxo,
    bb=boondox]{mathalfa}  

\usepackage[dvipsnames,table]{xcolor}
\usepackage{authblk}

\usepackage{bm}
\usepackage{soul}
\usepackage{graphicx}
\usepackage{
  amssymb,
  amsthm,
  amscd,
  dsfont,
  ytableau,
  enumitem
}
\usepackage[a4paper, left=1in, right=1in,top=1in, bottom=1in]{geometry}
\usepackage{titlesec}
\usepackage[numbers,sort&compress]{natbib}
\usepackage{float}
\usepackage{pdflscape}
\usepackage{tikz}
\usetikzlibrary{arrows.meta,decorations.pathreplacing,calc,positioning}
\usetikzlibrary{decorations.markings}
\tikzset{
  midarrow/.style={
    thick,
    postaction={decorate},
    decoration={markings, mark=at position 0.55 with {\arrow{Stealth}}}
  }
}

\DeclareSymbolFont{yhlargesymbols}{OMX}{yhex}{m}{n}
\DeclareMathAccent{\widehat}{\mathord}{yhlargesymbols}{"62}

\DeclareRobustCommand{\coloneqq}{\mathrel{\raisebox{0.04ex}{$:$}}\mathrel{\mkern-1.2mu}=}

\theoremstyle{plain}
\newtheorem{theorem}{Theorem}
\newtheorem{corollary}[theorem]{Corollary}
\newtheorem{definition}{Definition}
\newtheorem{proposition}{Proposition}
\newtheorem{lemma}{Lemma}

\theoremstyle{remark}
\newtheorem*{remark}{Remark}

\newcommand*{\indicator}{\ensuremath{\mathbf{1}}}

\DeclareMathOperator{\polylog}{polylog}
\DeclareMathOperator{\poly}{poly}
\let\Im\relax
\DeclareMathOperator{\Im}{Im}
\DeclareMathOperator{\Var}{Var}

\DeclarePairedDelimiter\bra{\langle}{\rvert}
\DeclarePairedDelimiter\ket{\lvert}{\rangle}
\DeclarePairedDelimiterX\braket[2]{\langle}{\rangle}{#1 \delimsize\vert #2}
\let\abs\relax  
\DeclarePairedDelimiter{\abs}{\vert}{\rvert}

\DeclarePairedDelimiter{\ceil}{\lceil}{\rceil}
\DeclarePairedDelimiter{\floor}{\lfloor}{\rfloor}
\DeclarePairedDelimiterXPP\tracenorm[1]{}{\lVert}{\rVert}{_1}{#1}
\DeclarePairedDelimiterXPP\diamondnorm[1]{}{\lVert}{\rVert}{_\diamond}{#1}
\DeclarePairedDelimiterXPP\lnorm[1]{}{\lVert}{\rVert}{_2}{#1}
\DeclarePairedDelimiterX\bilinear[2]{\langle}{\rangle}{#1, #2}%
\usetikzlibrary{quantikz2}
\providecommand\given{}
\newcommand\SetSymbol[1][]{%
\nonscript\: :
\allowbreak
\nonscript\:
\mathopen{}}
\DeclarePairedDelimiterX\Set[1]\{\}{%
\renewcommand\given{\SetSymbol[\delimsize]}
#1
}

\newcommand{\eg}{e.g.\@}

\usepackage{macros}

\usepackage{microtype}
\usepackage[
    colorlinks=true,
    urlcolor=blue,
    citecolor=blue,
    linkcolor=blue,
    anchorcolor=blue,
    hypertexnames=false]{hyperref}

\title{Decoded Quantum Interferometry Beyond Hamming: Rank-Metric and Translation Association Schemes}
\author{Alexandre~Krajenbrink}
\author{Colin~Krawchuk}
\author{Ansis~Rosmanis}
\author{Matthias~Rosenkranz}
\affil{Quantinuum, Partnership House, Carlisle Place, London SW1P 1BX, United Kingdom}
\date{June 4, 2026}

\begin{document}

\maketitle
\begin{abstract}
    Decoded Quantum Interferometry (DQI) uses coherent decoding and a quantum
    Fourier transform to find high-quality solutions of structured optimisation
    problems. Existing analyses are closely tied to Hamming space, which
    underlies the optimisation objective, Dicke state preparation and the
    decoding step of the algorithm. Here we extend the core DQI mechanism beyond
    Hamming space to finite geometries with translation symmetry, where points
    are grouped into shells by their distance from a basepoint. Mathematically,
    these geometries are translation association schemes. In this setting the
    algorithm can be analysed by tracking one amplitude per shell, and biasing
    the prepared state towards high-quality solutions becomes a finite
    tridiagonal eigenvalue problem. As a non-Hamming example, we develop an
    efficient DQI protocol for finding an $m\times n$ finite-field matrix with
    smallest rank difference to a target matrix. Initial states are uniform
    superpositions over fixed-rank matrices, and Gabidulin codes provide
    candidates for efficient low-rank decoding up to a cutoff $\ell$. For this
    objective, this finds solutions with an effective-rank proxy near $\min(m,
        n)-\ell$, and the corresponding expected score can be converted into a
    constant-probability bound on the residual rank of a sample. For Gabidulin
    nearest-codeword instances, a covering-radius obstruction shows that this
    bound does not imply an additive guarantee for the true optimum, and we do
    not claim a quantum advantage for the rank-metric construction. The results
    instead identify the geometric and coding ingredients for DQI beyond Hamming
    space.
\end{abstract}

\section{Introduction}
\label{sec:introduction}

Decoded Quantum Interferometry (DQI) is a quantum algorithmic framework for
approximately solving certain structured optimisation
problems~\cite{jordan25:DQI-original}. This is achieved by sampling from a
quantum state prepared via coherent decoding of syndromes and applying the
quantum Fourier transform for constructive interference on candidate solutions
with high objective value. DQI achieves in polynomial time a larger expected
number of satisfied constraints than any known polynomial-time classical
algorithm on a class of instances known as the Optimal Polynomial Intersection
(OPI) problem. This makes DQI a candidate for classically verifiable exponential
quantum advantage in structured optimisation
problems~\cite{khattar25:optimized_DQI}.

DQI is inspired by Regev's reduction for finding short vectors (in Euclidean
distance) on a dual lattice by solving the closest vector problem on the primal
lattice~\cite{regev05:reduction, regevLatticesLearningErrors2009}. Adapted to a
coding theory perspective, this reduction finds small dual codewords by solving
the decoding problem on a linear code in the
Hamming~\cite{debris-alazardQuantumReductionFinding2024} or rank
metric~\cite{blanvillainQuantumDecodingProblem2026}. Regev's procedure decodes
low-weight errors to prepare a superposition over noisy codewords and then
applies the quantum Fourier transform, concentrating amplitude on small dual
codewords. For DQI this is naturally expressed through max-LINSAT. Given
$B\in\F_q^{m\times k}$ and local target sets $F_i\subset\F_q$, the task is to
find $x\in\F_q^k$ such that $Bx$ lies close, in Hamming distance, to the product
set $F_1\times\cdots\times F_m$. The singleton case $F_i=\Set{y_i}$ is the
Hamming nearest-codeword objective, while general product targets count violated
local constraints. In this work, we develop a translation-scheme framework for
carrying the core DQI mechanism beyond Hamming distance to other finite
geometries.

The relevant structure is a $P$-polynomial translation association scheme on a
finite field with a translation-invariant metric. Such a scheme groups points by
their distance from a basepoint into \emph{shells}, generalising sets of fixed
Hamming weight. Equal superpositions over shells are the corresponding shell
Dicke states. In the radial cases relevant here, the Fourier transform maps each
shell Dicke state to a radial superposition, with amplitudes determined by
eigenvalues that depend only on the input and output shell indices. The
$P$-polynomial property then says that a single distance-one step sends each
shell state only to the same or neighbouring shell states. On the shell-state
basis, this gives the finite tridiagonal Jacobi matrix that controls the DQI
interference pattern. With efficient unique decoding on low shells, the usual
DQI decoding and Fourier steps produce amplitudes governed by these radial
eigenvalues, biasing samples toward low residual distance or high proxy score.

The main non-Hamming example studied in this paper is the rank metric on
$\F_q^{m\times n}$. Given an $\F_q$-linear encoding map
$\Phi\colon\Omega\to\F_q^{m\times n}$ with $\Omega\cong \F_q^t$, and a target
matrix $y$, the optimisation problem is
\begin{equation}\label{eq:rank-metric-optimisation-problem}
    \min_{x\in\Omega}\rank(\Phi(x)-y).
\end{equation}
Equivalently, for code $\Code=\Im(\Phi)$, this is rank minimisation over the
affine space $y-\Code$, or nearest-codeword decoding in the rank metric. This
objective is natural when constraint violation is structured by a
low-dimensional row or column space rather than by independent coordinate
errors. For example, it models crisscross array errors~\cite{roth91}, admits an
interpretation as maximising the intersection dimension of graphs of linear
maps, and is closely related to random network
coding~\cite{KoetterKschischang2008,bartz22}. Gabidulin codes are explicit
maximum rank distance (MRD) codes, the rank-metric analogue of Reed--Solomon
codes, with efficient unique decoders up to the usual unique-decoding
radius~\cite{gabidulin1985theory}. They are natural candidates for the coherent
dual-decoding step.

Our contributions are as follows. First, we formulate DQI for $P$-polynomial
translation schemes using shell Dicke states, finite-field Fourier phases, dual
codes, and radial eigenvalues on low-weight shells. Second, we specialise this
framework to the rank-metric scheme on $\F_q^{m\times n}$, derive the rank-shell
sizes and coefficients of the Jacobi operator, and identify the radial
eigenvalues with the $q$-Krawtchouk polynomials. Third, we prove an efficient
state preparation for uniform superpositions over states representing fixed-rank matrices and show that, after decoding and
the quantum Fourier transform, the final DQI state has amplitudes proportional
to $P_\ell\!\left(\rank(\Phi(x)-y)\right)$, where $P_\ell$ is a degree-$\ell$
polynomial determined by the chosen weights $w=(w_0, w_1, \dots, w_\ell)$ of
each shell. Fourth, we relate the rank objective in
Eq.~\eqref{eq:rank-metric-optimisation-problem} to the first radial eigenvalue of the
scheme's distance matrices, show that the expectation of a proxy score reduces
to a quadratic form $w^\dagger\tilde A^{(\ell)}w$ in the truncated Jacobi
matrix, and derive a corresponding variance formula. Finally, we exhibit a
ladder factorisation of the Jacobi operator and use it to bound the optimised
effective-rank proxy. In the large-system regime this proxy is tightly
controlled around $\min(m,n)-\ell$, up to a small square-case correction, and
the optimised expected proxy score gives a constant-probability bound on the
residual rank of a sampled candidate. For Gabidulin nearest-codeword instances, the decoding
radius available to DQI can be smaller than the typical distance from a target
$y$ to the code, limiting what this construction can support as a quantum
advantage claim. We therefore view the rank-metric case as a structural
extension of DQI beyond Hamming space.

This work complements several recent developments around DQI. The algorithmic
applications of the related Regev's reduction were first considered by
Ref.~\cite{chen22:lattice_probs_via_filtering} and subsequently studied in
Refs.~\cite{yamakawa24:RandomOracleNP,chailloux24:decodingProblem}. The original
DQI algorithm was introduced in Ref.~\cite{jordan25:DQI-original} with
improvements in Refs.~\cite{khattar25:optimized_DQI,rosmanis26:linearDQI}.
Variants using broader code families, soft decoders, and quantum decoders have
been studied in
Refs.~\cite{piveteau25:state_discrimination,gu2025algebraicgeometrycodesdecoded,buMultivariateDecodedQuantum2026,
    chailloux25:SoftDecoders,chailloux25_opi_x_softdecoders,
    shuttyOptimizationUsingLocallyQuantum2026}. Hamiltonian DQI extends the idea to
ground-state preparation and Gibbs sampling
\cite{schmidhuber25:HamiltonianDQI,bu26:hamiltonian_DQI,wocjan26:HDQI}. Recent
works have clarified that DQI advantages require substantial structure and can
disappear for unstructured instances
\cite{anschuetz25:DQI-needs-structure,parekh25:DQI_for_maxcut,
    kramerTightInapproximabilityMaxLINSAT2026,sunWorstCaseOptimalPolynomial2026}.
Practical aspects of DQI have also been studied, such as
implementations~\cite{patamawisutQuantumCircuitDesign2025}, the influence of
noise~\cite{bu26:DQI-under-noise, wang25:kernelized_DQI}, and its application to
specific problem
instances~\cite{sabater25:industrialLP-via-DQI,bollmann26:benchmark_DQI}. The
present paper takes a complementary perspective, isolating the distance-shell
spectral structure behind DQI and developing the rank-metric case as a concrete
non-Hamming example.

The paper is organised as follows.
Section~\ref{sec:translation-schemes-radial-structure} introduces translation
association schemes, shell Dicke states, and the Jacobi operator induced by the
$P$-polynomial property. Section~\ref{sec:fourier-analysis-codes} develops the
Fourier and coding tools used by the general translation-scheme DQI algorithm.
Section~\ref{sec:rank-metric-translation-scheme-codes} specialises the geometry
and coding theory to the rank metric and Gabidulin or MRD codes.
Section~\ref{sec:dqi-rank-metric} gives the rank-metric DQI protocol, optimises
its rank-shell weights $w$ through an eigenvalue problem, and relates the
resulting proxy score to the residual rank of sampled candidates. The appendices
collect details on $q$-Krawtchouk polynomials, rank-metric MacWilliams
identities, sampled-rank bounds from the expected proxy score, the Gabidulin
obstruction, and deferred proofs. The efficient state preparation for uniform
superpositions over states representing fixed-rank matrices in
Appendix~\ref{subsec:proof-rank-shell-dicke-preparation} may be of independent
interest.

\section{Translation schemes and radial structure}
\label{sec:translation-schemes-radial-structure}
This section introduces metric translation association schemes used in our DQI
analysis. The purpose is to generalise the sets of fixed Hamming weight used in
DQI to distance shells in a general finite metric space. When the metric has
enough translation symmetry, shell superpositions form a natural radial
subspace. The additional $P$-polynomial condition then makes the
nearest-neighbour distance operator act tridiagonally on this subspace, so the
DQI analysis reduces to a finite Jacobi matrix.

Let us first define our notation. For integers $a\leq b$, we use the notation
$[a, b] \coloneqq \Set{a, a+1, \dots, b}$ and if $a=0$ we simply write
$[b]\coloneqq [0,b]$. $X$ denotes a finite set with cardinality $\abs{X}$. For
any subset $S \subseteq X$ and $x\in X$, the indicator function is denoted by
\begin{equation}
    \indicator_S(x) \coloneqq \begin{cases}
        1 & \quad \text{for } x\in S, \\
        0 & \quad \text{otherwise}.
    \end{cases}
\end{equation}
$\F$ denotes an arbitrary field and, for a prime power $q$, $\F_q$ denotes a
finite field with $q$ elements.

\subsection{Metric translation schemes}
\label{subsec:metric-translation-schemes}

\begin{definition}[Association scheme]\label{def:association-scheme}
    A $D$-class \emph{(symmetric) association scheme} is given by a finite set $X$
    and a collection $\Set{A_i}_{i=0}^D$ of matrices $A_i\in\Set{0,
            1}^{\abs{X}\times\abs{X}}$ that satisfy
    \begin{enumerate}
        \item $A_0 = I$,
        \item $\sum_{i=0}^{D} A_i = J$,
        \item $A_i^\top = A_i$ for all $i\in [0,D]$,
        \item there exist integers $p_{ij}^k$ such that $A_i A_j = \sum_{k=0}^{D}
                  p_{ij}^k A_k$ for all $i,j\in [0,D]$,
    \end{enumerate}
    where $I$ is the $\abs{X}\times\abs{X}$  identity matrix and $J$ is the all-one
    matrix. The $p_{ij}^k$ are called \emph{intersection numbers}.
\end{definition}

Let $(X, d)$ be a finite metric space of diameter $D$. Let us define the \emph{distance
    matrices}
\begin{equation}\label{eq:distance-matrices}
    (A_i)_{xy} = \indicator_{\Set{i}}(d(x,y)),
    \qquad i\in [0,D],\; \forall x, y\in X.
\end{equation}
When these matrices satisfy Def.~\ref{def:association-scheme}, we call the
resulting scheme a \emph{metric association scheme}. We call $A_1$ the
\emph{adjacency matrix}.

Now assume that $X$ carries the additional structure of an additive abelian
group. A metric association scheme on $X$ is called a \emph{metric translation
    scheme} if the distance is translation-invariant:
\begin{equation}
    d(x,y) = d(x+z,y+z),
    \qquad
    \forall x,y,z \in X .
\end{equation}
When no ambiguity arises, we simply call it a \emph{translation scheme}. As a
consequence, we can choose a \emph{basepoint} $0 \in X$ so that
$d(x,y)=d(x-y,0)$.

For the finite-field translation schemes used below, we write the ground set as
$\X$. Here $\X$ denotes an ambient finite-dimensional $\F_q$-vector space,
viewed as an additive group and equipped with a translation-invariant metric. We write
\begin{equation}
    \operatorname{Sym}(\X,d)
    \coloneqq
    \Set{g\colon \X\to\X \text{ bijective} \given
        d(g(x),g(y))=d(x,y),\; \forall x,y\in\X}
\end{equation}
for the symmetry group of the metric space. The subgroup relevant for the
Fourier analysis later is the linear basepoint stabiliser
\begin{equation}\label{eq:linear-basepoint-stabiliser}
    G_0(\X,d)
    \coloneqq
    \Set{g\in \operatorname{Sym}(\X,d) \given
        g \text{ is } \F_q\text{-linear and } g(0)=0}.
\end{equation}

\subsection{Distance shells and radial states}
\label{subsec:distance-shells-radial-states}

We define the \emph{distance-$i$ shell} around the basepoint $0$ as the set
\begin{equation}\label{eq:distance-i-shell}
    \Gamma_i\coloneqq \Set{x\in \X \given d(x,0)=i},
    \qquad i \in [0,D] \, .
\end{equation}

To connect this shell structure to quantum algorithms, we associate $\X$ with
the computational basis $\Set{\ket{x} \given x\in\X}$ of the Hilbert space
$\mathbb{C}^{\X}$. Then we can write the action of any distance matrix of a
translation scheme on basis vectors of the above Hilbert space as
\begin{equation}\label{eq:action-distance-matrix}
    A_i \ket{x}
    = \sum_{e\in \Gamma_i} \ket{x-e}.
\end{equation}
For non-empty $\Gamma_i$, the uniform superposition of states in the
distance-$i$ shell is the shell Dicke state
\begin{equation}\label{eq:shell-state-definition}
    \ket{R_i}
    \coloneqq
    \frac{1}{\sqrt{\abs{\Gamma_i}}} \sum_{x\in \Gamma_i}\ket{x}=\frac{1}{\sqrt{\abs{\Gamma_i}}}\,A_i\ket{0}.
\end{equation}

The shell Dicke state $\ket{R_i}$ forgets the position inside the shell and
keeps only the distance label $i$. This agrees with the symmetry viewpoint:
since each element of $G_0(\X,d)$ is distance-preserving and fixes the
basepoint, $G_0(\X,d)$ maps every shell $\Gamma_i$ to itself. In the
distance-transitive examples considered below, $G_0(\X,d)$ acts transitively on
each $\Gamma_i$, so $\ket{R_i}$ is the uniform superposition over one
$G_0$-orbit. We define the radial subspace as the $(D+1)$-dimensional span
\begin{equation}\label{eq:radial-subspace}
    \Rad\coloneqq \Span\Set{\ket{R_i} \given i \in[0,D]}.
\end{equation}
Thus a state in $\Rad$ is described by one amplitude per distance shell, rather
than one amplitude per element of $\X$.

\subsection{\texorpdfstring{$P$-polynomial schemes and the Jacobi matrix}{P-polynomial schemes and the Jacobi matrix}}
\label{subsec:p-polynomial-schemes-jacobi-matrix}

The DQI analysis uses the radial subspace $\Rad$. We therefore need the action
of the adjacency matrix $A_1$ on shell states, which in the Hamming case are the
usual Dicke states. The $P$-polynomial property makes this action
nearest-neighbour in the shell label. After normalising shell states, $A_1$
couples $\ket{R_i}$ only to $\ket{R_{i-1}}$, $\ket{R_i}$, and $\ket{R_{i+1}}$.
Thus $A_1$ restricts to a tridiagonal Jacobi matrix on $\Rad$, controlling the
DQI mechanism analysed below.

\begin{definition}[$P$-polynomial association scheme]\label{def:p-polynomial-association-scheme}
    A metric association scheme $(X, \Set{A_i}_{i=0}^D)$ is
    \emph{$P$-polynomial} with respect to $A_1$ if multiplication by $A_1$
    satisfies the three-term recurrence relation
    \begin{equation}
        \label{eq:p-polynomial-three-term-recurrence}
        A_1 A_i = b_{i-1} A_{i-1} + a_i A_i + c_{i+1} A_{i+1},
        \qquad i \in [0, D],
    \end{equation}
    where terms with indices outside $[0,D]$ are omitted. The
    \emph{recursion coefficients} are given by
    \begin{equation}\label{eq:recursion-coefficients}
        b_i = p_{1, i+1}^{i} \quad (0\le i<D), \qquad
        a_i = p_{1, i}^i \quad (0\le i\le D), \qquad
        c_i = p_{1, i-1}^{i} \quad (1\le i\le D)
    \end{equation}
    with the endpoint conventions $c_0=b_D=0$ and $b_{-1}=c_{D+1}=0$. Note that $b_i,c_{i+1}>0$ for $0\le i<D$.
\end{definition}

Equation~\eqref{eq:p-polynomial-three-term-recurrence} shows that the adjacency
matrix $A_1$ induces a walk that mixes points on neighbouring shells.
Figure~\ref{fig:radial-shell-walk} illustrates the shells around the basepoint
$0$ and the walk induced by $A_1$.
Equation~\eqref{eq:p-polynomial-three-term-recurrence} also implies that each
distance matrix $A_i$ is a polynomial in the matrix $A_1$.

For a translation scheme, the intersection numbers can be read as shell
intersection counts around the basepoint. Namely, for fixed $z\in\Gamma_k$,
\begin{equation}
    p_{ij}^k
    =
    \abs*{\Set*{x\in\Gamma_i \given z-x\in\Gamma_j}}.
\end{equation}
These numbers do not depend on the particular choice of $z$ in the shell
$\Gamma_k$. In particular, for fixed $z\in\Gamma_i$, the recursion coefficients
count the nearest neighbours of $z$. Equivalently,
\begin{align}
    b_i & = \abs*{\Set*{u \in \Gamma_1 \given z-u \in \Gamma_{i+1}}},
    \qquad 0\le i<D,                                                  \\
    a_i & = \abs*{\Set*{u \in \Gamma_1 \given z-u \in \Gamma_i}},
    \qquad \quad \! 0\le i\le D,                                      \\
    c_i & = \abs*{\Set*{u \in \Gamma_1 \given z-u \in \Gamma_{i-1}}},
    \qquad 1\le i\le D.
\end{align}
The endpoint values are again $b_D=c_0=0$. Furthermore, $p_{ij}^0 =
    \abs{\Gamma_i} \delta_{ij}$, where $\delta_{ij}$ is the Kronecker delta.

The following two identities explain why the shell walk has the form of a
symmetric tridiagonal Jacobi operator on normalised shell states. The first
identity symmetrises the off-diagonal coefficients, while the second is used
later to factorise this operator into ladder operators, giving the radial walk
an oscillator-like interpretation.

\begin{lemma}[Recursion coefficient identities]
    \label{lem:recursion-coefficient-identities}
    For any $P$-polynomial translation scheme, the recursion coefficients
    satisfy
    \begin{equation}
        \label{eq:ratio-off-diagonal}
        \frac{\abs{\Gamma_{i+1}}}{\abs{\Gamma_i}}
        = \frac{b_i}{c_{i+1}},
        \qquad 0\le i<D,
    \end{equation}
    and
    \begin{equation}
        \label{eq:recursion-coefficient-sum}
        a_i + b_i + c_i = \abs{\Gamma_1},
        \qquad 0\le i\le D,
    \end{equation}
    with the endpoint convention $c_0 = b_D = 0$.
\end{lemma}
\noindent The proof is in
Appendix~\ref{subsec:proof-recursion-coefficient-identities}.

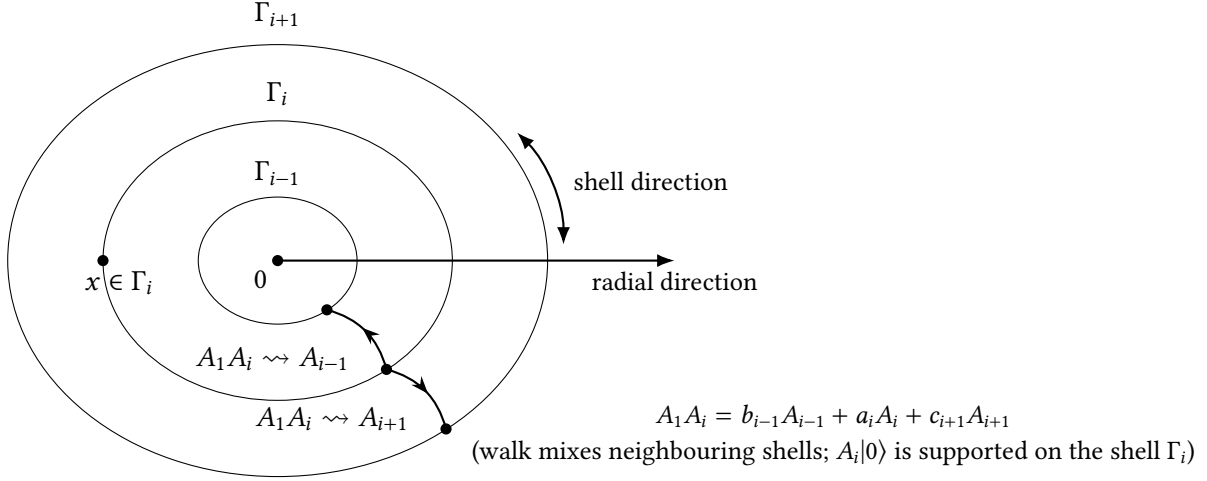
\begin{figure}[t!]
    \centering
    \begin{tikzpicture}[scale=1.05,>=Latex, arrow_line/.style={->, >=Stealth, thick, bend left}]
        \fill (0,0) circle (2pt);
        \node[below left] at (0,0) {$0$};

        \def\alpha{0.8}

        \draw (0,0) ellipse (1.0 and {1.0*\alpha});
        \draw (0,0) ellipse (2.2 and {2.2*\alpha});
        \draw (0,0) ellipse (3.4 and {3.4*\alpha});

        \node at (0,1.1) {$\Gamma_{i-1}$};
        \node at (0,2.1) {$\Gamma_{i}$};
        \node at (0,3.1) {$\Gamma_{i+1}$};

        \fill (-2.2,0) circle (2pt);
        \node[below] at (-2,0) {$x\in\Gamma_i$};

        \draw[->,thick] (0,0) -- (5,0) node[below] {\small radial direction};

        \draw[<->,thick] (3.6,0.2) arc (0:45:2);
        \node at (4.7,1) {{\small shell direction}};

        \coordinate (P1) at (0.62,-0.62);
        \coordinate (P2) at (1.37,-1.37);  
        \coordinate (P3) at (2.12,-2.12);

        \fill (P1) circle (2pt);
        \fill (P2) circle (2pt);
        \fill (P3) circle (2pt);

        \draw[midarrow, bend left]
        (P2) to node[pos=0.5, below left=3pt] {$A_1A_i \leadsto A_{i+1}$} (P3);

        \draw[midarrow, bend right]
        (P2) to node[pos=0.5, below left=3pt] {$A_1 A_i \leadsto A_{i-1}$} (P1);

        \node[align=center] at (7,-2.2) {\small
            $A_1A_i = b_{i-1}A_{i-1} + a_iA_i + c_{i+1}A_{i+1}$\\
            {\small (walk mixes neighbouring shells; $A_i\ket{0}$ is supported on the shell $\Gamma_i$)}
        };
    \end{tikzpicture}
    \caption{Shells around a basepoint $0$ and the radial and shell directions. The state $A_i \ket{0}$ is supported on the shell $\Gamma_i$, while $A_1$ induces a `radial walk' between a shell and its two neighbouring shells.}
    \label{fig:radial-shell-walk}
\end{figure}

Combining Eqs.~\eqref{eq:shell-state-definition}, \eqref{eq:p-polynomial-three-term-recurrence} and
\eqref{eq:ratio-off-diagonal}, the action of $A_1$ on a shell state can be
expressed, with endpoint terms omitted, as
\begin{equation}
    \label{eq:radial-jacobi-matrix}
    \begin{split}
        A_1 \ket{R_i} & = b_{i-1} \sqrt{\frac{\abs{\Gamma_{i-1}}}{\abs{\Gamma_{i}}}} \ket{R_{i-1}} + a_i \ket{R_i} + c_{i+1} \sqrt{\frac{\abs{\Gamma_{i+1}}}{\abs{\Gamma_{i}}}} \ket{R_{i+1}} \\
                      & = \sqrt{b_{i-1}c_i}  \ket{R_{i-1}} + a_i \ket{R_i} + \sqrt{c_{i+1} b_i}\ket{R_{i+1}}.
    \end{split}
\end{equation}
Thus the full operator $A_1$ restricts on the radial subspace to a
$(D+1)$-dimensional symmetric tridiagonal \emph{Jacobi matrix} denoted by
$\tilde{A}$. This is the effective radial operator that later appears in the DQI
performance analysis, where only shell amplitudes are tracked. Its entries are
fully determined by scheme-dependent recurrence coefficients. The identities in
Lemma~\ref{lem:recursion-coefficient-identities} also imply a ladder
factorisation $\tilde A=\abs{\Gamma_1}I-B^\dagger B$. This gives the radial
Jacobi operator an oscillator-like interpretation: the ladder operators move
between neighbouring distance shells with coefficients determined by the
geometry. For example, in the Hamming scheme, the Jacobi operator converges to
the standard quantum harmonic oscillator in the large system ($D\to\infty$),
low-excitation ($i\ll D$) regime~\cite{marwahaComplexityDecodedQuantum2025}. We
use this factorisation in
Sec.~\ref{subsec:rank-shell-weights-ladder-factorisation}.

Finally, we note an orthogonality relation for the distance matrices. It is not
needed for the Jacobi reduction above, but will become the orthogonality of the
Fourier eigenvalues in
Sec.~\ref{subsec:fourier-diagonalisation-translation-schemes}. For $i,j\in
    [0,D]$ and $C\subseteq X$, let $\Tr_C(T)\coloneqq\sum_{x\in C}\bra{x}T\ket{x}$
denote the partial trace over the basis states indexed by $C$. Then
\begin{equation}
    \label{eq:distance-matrix-orthogonality}
    \Tr_C(A_iA_j)
    =
    \abs{C}p_{ij}^0
    =
    \abs{C}\abs{\Gamma_i}\delta_{ij}.
\end{equation}
Indeed, $A_iA_j=\sum_k p_{ij}^kA_k$, and only $A_0$ has nonzero diagonal. For
$C=X$, this is the Hilbert-Schmidt orthogonality of the distance matrices.

\section{Fourier analysis and codes}
\label{sec:fourier-analysis-codes}
This section connects the radial structure in
Sec.~\ref{sec:translation-schemes-radial-structure} to finite-field Fourier
analysis and coding theory. In the Fourier basis, the distance matrices $A_i$
become diagonal, and their eigenvalues are obtained by summing Fourier phases
over distance shells. When these eigenvalues depend only on the shell index, the
Fourier transform preserves the radial description and DQI only needs $D+1$
shell amplitudes. Averaging Fourier phases over a code enforces membership in
the dual code, which is the algebraic constraint used by the syndrome decoding
step of DQI.

\subsection{Characters and the finite-field Fourier transform}
\label{subsec:characters-finite-field-fourier-transform}
We wish to define a unitary Fourier transform $\mathcal{F}_{\X}$ on the space of
quantum states spanned by the basis states $\Set{\ket{x} \given x\in \X}$.
Operationally, $\mathcal{F}_{\X}$ is the quantum Fourier transform used in DQI.
We fix a non-degenerate symmetric $\F_q$-bilinear form
\begin{equation}\label{eq:finite-field-bilinear-form}
    \bilinear{\cdot}{\cdot}\colon \X\times \X\to \F_q .
\end{equation}
In coordinates one may take $\X\cong\F_q^N$ with the standard dot product
$\bilinear{x}{y}=\sum_{i=1}^N x_i y_i$. The Fourier transform used below is the
group Fourier transform of the additive group $\X$. Viewing the finite field
$\F_q$ as an abelian group, an \emph{additive character} is a group homomorphism
\begin{equation}
    \chi\colon \F_q \longrightarrow  \C\setminus\Set{0}
\end{equation}
satisfying
\begin{equation}
    \chi(a+b)=\chi(a)\chi(b)\, ,\quad \forall a,b\in\F_q,
\end{equation}
and we require that $\chi$ does not map every element in $\mathbb{F}_q$ to $1$
(it is nontrivial). In other words, addition in $\F_q$ becomes multiplication of
complex numbers. Moreover, for any $a\in\F_q$,
\begin{equation}\label{eq:additive-character-orthogonality}
    \sum_{t\in\F_q}\chi(at)=
    \begin{cases}
        q & \quad \text{for } a=0,     \\
        0 & \quad \text{for } a\neq 0.
    \end{cases}
\end{equation}
If $a\neq 0$, the map $t\mapsto at$ permutes $\F_q$, so the sum equals
$\sum_{t}\chi(t)=0$ for any nontrivial character. If $q$ is prime, a standard
example is $\chi(t)=e^{2\pi i t/q}$. If $q=p^m$ is a prime power, choose any
nontrivial additive character, for instance
\begin{equation}
    \chi(t)=\exp \left(\frac{2 \pi i}{p} \operatorname{Tr}_{\F_{p^m} / \F_p}(t)\right)
\end{equation}
where
\begin{equation}
    \Tr_{\F_{p^m} / \F_{p}}(t)=t+t^p+t^{p^2}+\cdots+t^{p^{m-1}}
\end{equation}
is the field trace. The group of additive characters of $\X$ is
$\widehat{\X}=\mathrm{Hom}((\X,+), \C\setminus\Set{0})$. Fixing any nontrivial
additive character $\chi$ on $\F_q$, we can use it to obtain lifted characters
in $\widehat{\X}$. Using the bilinear form,
Eq.~\eqref{eq:finite-field-bilinear-form}, every element $y\in\X$ defines a map
\begin{equation}\label{eq:ambient-additive-character}
    \adchar_y(x)\coloneqq\chi(\bilinear{x}{y}),
    \qquad x,y\in\X.
\end{equation}
Because $\bilinear{\cdot}{\cdot}$ is bilinear in $\F_q$,
\begin{equation}
    \adchar_y(x+x^\prime)=\chi(\bilinear{x+x^\prime}{y})
    =\chi(\bilinear{x}{y}+\bilinear{x^\prime}{y})
    =\adchar_y(x) \adchar_y(x^\prime)
\end{equation}
showing that $\adchar_y$ is an additive character of $\X$. Because the bilinear
form is non-degenerate, the mapping $y \mapsto \adchar_y$ is an isomorphism,
providing the identification
\begin{equation}
    \widehat{\X}\cong \X,
    \qquad \adchar_y\longleftrightarrow y.
\end{equation}

Now recall the Fourier transform of a function
\begin{equation}
    f\colon \X \to
    \mathbb{C}, \qquad \hat{f}(y) =
    \frac{1}{\sqrt{\abs*{\X}}}\sum_{x\in\X} \adchar_y(x) f(x).
\end{equation}
The pairing $(x,y) \mapsto \adchar_y(x)$
defines the unitary Fourier transform $\mathcal{F}_{\X}$ on
$\C^{\X}$ by
\begin{equation}\label{eq:fourier-basis-state}
    \mathcal{F}_{\X} \ket{x}
    \coloneqq
    \frac{1}{\sqrt{\abs*{\X}}}\sum_{y\in \X}\adchar_y(x)\ket{y}.
\end{equation}
When $\X\cong\F_q^N$ and $q=p$ is prime, Eq.~\eqref{eq:fourier-basis-state} is
implemented by applying the $p$-point quantum Fourier transform $F_p$ to each
coordinate after identifying $\F_p=[p-1]$, so $\mathcal{F}_{\X}=F_p^{\otimes
            N}$. When $q=p^s$ is a prime power, we first decompose field elements into
prime-subfield components, $\F_q^N\cong \F_p^{sN}$. Then the Fourier transform
factors into the $p$-point quantum Fourier transform acting on each of the $sN$
components as in the prime case. When the set of the scheme is clear from
context, we also write $\mathcal{F}$ instead of $\mathcal{F}_\X$.

\subsection{Fourier diagonalisation of translation schemes}
\label{subsec:fourier-diagonalisation-translation-schemes}

Each distance matrix $A_i$ acts by summing over translations by elements of
$\Gamma_i$. Since the finite-field Fourier transform diagonalises translations,
one expects the states $\mathcal{F}_{\X}\ket{x}$ to diagonalise all $A_i$. The
next lemma makes this precise.

\begin{lemma}[Eigenvectors and eigenvalues of distance matrices]
    \label{lem:distance-matrix-fourier-eigenvectors}
    Let $(\X, \Set{A_i}_{i=0}^D)$ be a translation scheme. For each $i \in
        [0,D]$ and each $x\in \X$, the Fourier basis state $\mathcal{F}_\X
        \ket{x}$ is an eigenvector of $A_i$
    \begin{equation}
        A_i \mathcal{F}_\X \ket{x} = \lambda_i(x)\mathcal{F}_\X \ket{x}
    \end{equation}
    with Fourier eigenvalue
    \begin{equation}\label{eq:distance-matrix-character-sum}
        \lambda_i(x) \coloneqq\;\sum_{e\in\Gamma_i}\adchar_x(e),
    \end{equation}
    where $\Gamma_i$ is the distance-$i$ shell of the scheme,
    Eq.~\eqref{eq:distance-i-shell}. The formula writes the Fourier eigenvalue as
    a character sum over the shell $\Gamma_i$.
\end{lemma}
\noindent The proof is in Appendix~\ref{subsec:proof-distance-matrix-fourier-eigenvectors}.

Applying this to a shell state gives
\begin{equation}
    \mathcal{F}_{\X}\ket{R_i}
    =
    \frac{1}{\sqrt{\abs{\X}\abs{\Gamma_i}}}
    \sum_{x\in\X}\lambda_i(x)\ket{x}.
\end{equation}
For the DQI analysis, we want the amplitudes to be determined only by the
distance $d(x,0)$, rather than individual $x$. Equivalently, all basis states
inside the same shell $\Gamma_j$ should have the same amplitude. This is the
radiality condition introduced next.

\begin{definition}[Radial Fourier eigenvalues]
    \label{def:radial-fourier-eigenvalues}
    Let $(\X,\Set{A_i}_{i=0}^D)$ be a translation scheme with shells
    $\Set{\Gamma_j}_{j=0}^D$, and let $\lambda_i(x)$ be the Fourier eigenvalue
    from Eq.~\eqref{eq:distance-matrix-character-sum}. We say that the scheme has
    \emph{radial Fourier eigenvalues} if, for every $i,j\in[0,D]$, the value
    $\lambda_i(x)$ is constant over all $x\in\Gamma_j$. In this case we write
    \begin{equation}
        \lambda_i(j)
        \coloneqq
        \lambda_i(x),
        \qquad x\in\Gamma_j,
    \end{equation}
    or equivalently $\lambda_i(d(x,0))$.
\end{definition}

Under this condition, $\mathcal{F}_{\X}\ket{R_i}$ lies in the radial subspace
$\Rad=\Span\Set{\ket{R_j}\given j\in[0,D]}$. We call this state the
\emph{dual radial state}
\begin{equation}
    \ket{\lambda_i}
    \coloneqq
    \mathcal{F}_{\X}\ket{R_i}
    =
    \frac{1}{\sqrt{\abs{\X}\abs{\Gamma_i}}}
    \sum_{x\in\X}\lambda_i(d(x,0))\ket{x}.
\end{equation}
Equivalently,
\begin{equation}
    \label{eq:fourier-transform-radial-state}
    \ket{\lambda_i}
    =
    \frac{1}{\sqrt{\abs{\X}\abs{\Gamma_i}}}
    \sum_{j=0}^D \sqrt{\abs{\Gamma_j}}\lambda_i(j)\ket{R_j}.
\end{equation}
In this case the Fourier step can be analysed using only one amplitude per
shell, hence at most $D+1$ amplitudes. The next lemma gives sufficient
conditions under which this radiality property holds.

\begin{lemma}[Sufficient conditions for radial Fourier eigenvalues]
    \label{lem:radial-fourier-eigenvalue-conditions}
    Let $(\X,\Set{A_i}_{i=0}^D)$ be a symmetric translation scheme with shells
    $\Set{\Gamma_i}_{i=0}^D$, and let $G_0$ be the linear basepoint stabiliser
    from Eq.~\eqref{eq:linear-basepoint-stabiliser}. Suppose that $G_0$ acts transitively on each
    shell $\Gamma_j$, and that for every $g\in G_0$ the adjoint $g^\top$,
    defined by $\bilinear{g(x)}{y}=\bilinear{x}{g^\top(y)}$, also preserves each
    shell. Then the Fourier eigenvalues in
    Eq.~\eqref{eq:distance-matrix-character-sum} are radial.
\end{lemma}
\noindent The proof is in
Appendix~\ref{subsec:proof-radial-fourier-eigenvalue-conditions}.

For a $P$-polynomial scheme, Fourier diagonalisation turns the matrix recurrence
Eq.~\eqref{eq:p-polynomial-three-term-recurrence} into a scalar three-term
recurrence for the radial Fourier eigenvalues. For each shell label $j\in[0,D]$,
\begin{equation}
    \label{eq:character-sum-three-term-recurrence}
    \lambda_1(j)\lambda_i(j)
    =
    b_{i-1} \lambda_{i-1}(j) + a_i \lambda_i(j)
    + c_{i+1} \lambda_{i+1}(j),
    \qquad i\in [0,D],
\end{equation}
where terms with indices outside $[0,D]$ are omitted. Thus the functions
$j\mapsto\lambda_i(j)$ are the radial eigenvalue polynomials associated with the
shell walk (orthogonal polynomials). In the classical metric schemes, these
families are the finite orthogonal polynomials appearing in the ($q$-)Askey
scheme, a hierarchy of hypergeometric orthogonal polynomials and their limit
relations~\cite{koekoekHypergeometricOrthogonalPolynomials2010}. For the Hamming
scheme they are the Krawtchouk polynomials discussed in
Ref.~\cite{marwahaComplexityDecodedQuantum2025}, and for the rank-metric scheme
introduced in Sec.~\ref{sec:rank-metric-translation-scheme-codes} they are the
$q$-Krawtchouk polynomials discussed in
Appendix~\ref{sec:rank-metric-q-krawtchouk-polynomials}.

Because the matrices are simultaneously diagonalised, the orthogonality of the
distance matrices in Eq.~\eqref{eq:distance-matrix-orthogonality} implies the
orthogonality of the radial Fourier eigenvalues:
\begin{equation}
    \label{eq:character-sum-orthogonality}
    \sum_{x\in \X} \lambda_i(d(x, 0))\lambda_j^*(d(x, 0))
    = \sum_{k=0}^D \abs{\Gamma_k} \lambda_i(k) \lambda_j^*(k)
    = \abs[\Big]{\X} \abs[\Big]{\Gamma_i} \delta_{ij}.
\end{equation}
This general algebraic relation guarantees the orthogonality of the dual radial
states, and it serves as a crucial identity when analysing the performance of
DQI in Sec.~\ref{sec:dqi-rank-metric}.

\subsection{Additive codes and dual-distance orthogonality}
\label{subsec:additive-codes-dual-distance-orthogonality}

For DQI we need the following results from coding theory: a code average should
enforce a dual constraint, and low-shell errors should be distinguishable from
their syndrome. Let $\X$ be equipped with the character pairing from
Sec.~\ref{subsec:characters-finite-field-fourier-transform}. An \emph{additive
    code} is a subgroup $\Code\le(\X,+)$. Its annihilator is
\begin{equation}
    \Code^\perp = \Set*{y \in \X\, \given \adchar_y(x)=1, \; \forall x\in \Code }.
\end{equation}
When $\X\cong\F_q^N$ and $\Code$ is $\F_q$-linear, this is the usual orthogonal
dual, $\Set*{y\in\F_q^N \given \bilinear{x}{y}=0,\; \forall x\in\Code}$. For a
translation-invariant metric $d$, define
\begin{equation}
    d(\Code)
    \coloneqq \min_{c\in\Code\setminus\Set{0}} d(c,0),
\end{equation}
with $d(\Set{0})=\infty$, and write $d^\perp=d(\Code^\perp)$.

The basic identity is the standard character average over a subgroup. For every
$z\in\X$,
\begin{equation}
    \sum_{x\in\Code}\adchar_z(x)
    =
    \abs{\Code}\indicator_{\Code^\perp}(z).
\end{equation}
At the level of quantum states, the preceding character average is exactly the
amplitude obtained by applying $\mathcal{F}_{\X}$ to the uniform superposition
over the code. Hence $\mathcal{F}_{\X}$ maps
$\abs{\Code}^{-1/2}\sum_{x\in\Code}\ket{x}$ to the uniform superposition over
$\Code^\perp$. This is the algebraic part of the syndrome step in DQI.

The coding-theoretic setup for DQI consists of a finite message space
$\Omega$ and an injective additive encoding
\begin{equation}
    \Phi\colon\Omega\to\Code\le \X,
    \qquad
    \Code=\operatorname{Im}(\Phi).
\end{equation}
Thus $\Phi$ identifies $\Omega$ with $\Code$. Given a target $y\in \X$, define
the residual
\begin{equation}\label{eq:residual-map}
    \Delta_y(x)\coloneqq \Phi(x)-y,
    \qquad x\in\Omega.
\end{equation}
We can think of $\Omega$ as a vector space $\Omega\cong \F_q^t$. The adjoint map
$\Phi^\top\colon\X\to\Omega$ is defined by
\begin{equation}
    \adchar_x(\Phi^\top(e))=\adchar_{\Phi(x)}(e),
    \qquad e\in\X,\ x\in\Omega.
\end{equation}
This is the abstract form of the syndrome map: it maps an error $e$ to the phase
it induces on codewords. In coordinates this is the transpose syndrome map. In
the rank-metric specialisation below, $\Phi^\top(e)=e\mathcal G^\top$.

For a shell cutoff $\ell$, write
\begin{equation}
    \Gamma_{\le \ell}\coloneqq\bigcup_{k=0}^{\ell}\Gamma_k.
\end{equation}
The DQI protocol assumes an efficient decoder that recovers the unique error
$e\in\Gamma_{\le \ell}$ from the syndrome $\Phi^\top(e)$. Uniqueness of this
recovery is equivalent to $\Phi^\top$ being injective on $\Gamma_{\le\ell}$.
Since $\ker(\Phi^\top)=\Code^\perp$, the uniqueness condition is guaranteed
whenever $d(\Code^\perp)>2\ell$.

The next lemma applies the subgroup character average above to the character-sum
formula for $\lambda_i$. After expanding the two character sums, it averages
over pairs of shell errors $e\in\Gamma_i$ and $e^\prime\in\Gamma_j$. The average
over $c\in\Code$ keeps exactly the pairs whose difference lies in $\Code^\perp$.
If $i+j<d^\perp$, the triangle inequality forces every such low-shell pair to be
diagonal. This argument does not require radial Fourier eigenvalues. It only
uses their expansion as sums over shells.

\begin{lemma}[Shell orthogonality below the dual code distance]
    \label{lem:shell-orthogonality-dual-distance}
    Let $\X$ be equipped with a translation-invariant metric $d$ of diameter $D$.
    Let $\Gamma_i$ be the distance-$i$ shell as in
    Eq.~\eqref{eq:distance-i-shell}, and let $\lambda_i$ denote the character sum
    in Eq.~\eqref{eq:distance-matrix-character-sum}. Let $\Code\le \X$ be an
    additive code with annihilator minimum distance $d^\perp$. For $i,j\in[0,D]$
    and $y\in\X$, define
    \begin{equation}\label{eq:code-averaged-character-sum}
        S_{i,j}(y)
        \coloneqq\sum_{c\in \Code} \lambda_i(c-y)\lambda_j^*(c-y).
    \end{equation}
    If $i+j<d^\perp$, then any $e\in\Gamma_i$ and $e^\prime\in\Gamma_j$ with
    $e-e^\prime\in\Code^\perp$ satisfy $e=e^\prime$. Consequently,
    \begin{equation}\label{eq:code-average-shell-orthogonality}
        S_{i,j}(y)=\abs{\Code}\,\abs{\Gamma_i}\,\delta_{i,j}.
    \end{equation}
\end{lemma}
\noindent Operationally, the code average removes all low-shell off-diagonal
terms: below the dual distance, the only dual codeword that can appear as a
difference of two shell elements is zero. The proof is in
Appendix~\ref{subsec:proof-shell-orthogonality-dual-distance}.

\section{Rank-metric translation scheme and codes}
\label{sec:rank-metric-translation-scheme-codes}
This section supplies the two rank-metric inputs used by DQI in
Sec.~\ref{sec:dqi-rank-metric}. The rank metric on $\X=\F_q^{m\times n}$ gives
the distance shells, radial Fourier eigenvalues, Jacobi matrix, and objective
function. Rank-metric codes give the syndrome map whose dual decoder is used to
reversibly uncompute low-rank errors. In this setting, Hamming strings are
replaced by matrices, Hamming weights by ranks of matrix differences, and
maximum distance separable (MDS) codes such as Reed--Solomon codes by MRD codes
such as Gabidulin codes.

Fix a prime power $q$ and integers $m,n\ge 1$. Let the ambient space be
$\X=\F_q^{m\times n}$ endowed with the rank distance
\begin{equation}\label{eq:rank-distance}
    d_{\mathrm{rk}}(x,y) \coloneqq \rank(x-y),
    \qquad x,y\in \F_q^{m\times n},
\end{equation}
whose diameter is $D=\min(m,n)$. Fourier analysis on this additive group uses
the standard non-degenerate $\F_q$-bilinear form
\begin{equation}\label{eq:trace-pairing}
    \bilinear{x}{y} \coloneqq \Tr(x y^\top)\in \F_q,
    \qquad x,y\in \F_q^{m\times n},
\end{equation}
where $\Tr$ denotes matrix trace over $\F_q$. We call
Eq.~\eqref{eq:trace-pairing} \emph{trace pairing}. The trace pairing lets us
write every additive character of $\F_q^{m\times n}$ in the form
$\adchar_y(x)=\chi(\Tr(xy^\top))$ for a unique matrix $y\in\F_q^{m\times n}$.
Thus the same pairing controls both the Fourier phases and the dual-code
syndrome used by DQI. The distance matrices are
\begin{equation}\label{eq:rank-distance-matrices}
    (A_r)_{xy} = \indicator_{\Set{r}}(\rank(x-y)), \qquad x, y\in \F_q^{m\times n} \, .
\end{equation}

\subsection{Rank-metric geometry}
\label{subsec:rank-metric-geometry}

The rank-$r$ shell around $0$ is
\begin{equation}
    \Gamma_r\coloneqq \Set{x\in \F_q^{m\times n} \given \rank(x)=r},
    \qquad  r\in [0,D].
\end{equation}
The maps $x\mapsto gxh$, with $g\in\mathrm{GL}_m(\F_q)$ and
$h\in\mathrm{GL}_n(\F_q)$, preserve rank and act transitively on each shell. The
adjoint maps with respect to the trace pairing also preserve rank. Hence
Lemma~\ref{lem:radial-fourier-eigenvalue-conditions} applies: the Fourier
eigenvalues of the rank-distance matrices depend only on the rank shell. Thus
the shell superposition prepared by DQI is described by one coefficient per rank
shell, and after the Fourier step the amplitude of a candidate $x$ depends only
on $\rank(\Delta_y(x))$.

The next lemma summarises the properties of the rank-metric scheme. The
$P$-polynomial property makes the radial dynamics tridiagonal, the shell sizes
normalise the rank-shell Dicke states, and the coefficients $a_r,b_r,c_r$ become
the entries of the Jacobi matrix used in the performance analysis. The following
notation is convenient for the rank shell sizes: the $q$-Pochhammer symbol is
defined by
\begin{equation}
    \label{eq:finite-q-pochhammer-symbol}
    (a;q)_k \coloneqq \prod_{j=0}^{k-1} (1-aq^j),
    \qquad k\in\Z_{\ge 0},
\end{equation}
and the Gaussian binomial coefficient by
\begin{equation}
    \label{eq:gaussian-binomial-coefficient}
    \binom{n}{k}_q
    \coloneqq \prod_{i=0}^{k-1}\frac{q^{n-i}-1}{q^{k-i}-1} =
    \begin{cases}
        \dfrac{(q;q)_n}{(q;q)_k\,(q;q)_{n-k}} & \quad \text{for }0\le k\le n, \\
        0                                     & \quad \text{otherwise}.
    \end{cases}
\end{equation}

\begin{lemma}[Rank-metric shell sizes and intersection numbers]\label{lem:rank-metric-shell-sizes-intersection-numbers}
    The rank-metric scheme $(\F_q^{m\times n}, \Set{A_r}_{r=0}^D)$, with
    $D=\min(m,n)$, has the following properties:

    \begin{enumerate}
        \item \label{item:rank-metric-p-polynomial} The scheme is $P$-polynomial.
              Its radial Fourier eigenvalues $\lambda_r(\cdot)$ form the
              $q$-Krawtchouk family, see
              Appendix~\ref{sec:rank-metric-q-krawtchouk-polynomials}. Equivalently,
              \begin{equation}
                  A_1A_r=b_{r-1}A_{r-1}+a_rA_r+c_{r+1}A_{r+1},
                  \qquad r\in[0,D],
              \end{equation}
              with $b_{-1}=c_{D+1}=0$ at the endpoints (cf. Def.~\ref{def:p-polynomial-association-scheme}).

        \item \label{item:rank-metric-shell-size} The size of the rank-$r$ shell, for
              $r\in[0,D]$, is
              \begin{equation}
                  \label{eq:rank-shell-size}
                  \abs{\Gamma_r}= (-1)^r\binom{m}{r}_q
                  \binom{n}{r}_q
                  (q;q)_r
                  q^{\frac{r(r-1)}{2}}.
              \end{equation}

        \item \label{item:rank-metric-intersection-parameters} The recursion coefficients are
              \begin{equation}\label{eq:rank-b-coefficient}
                  b_r = \frac{(q^m-q^r)(q^n-q^r)}{q-1},
                  \qquad 0\le r<D,
              \end{equation}
              and
              \begin{equation}\label{eq:rank-c-coefficient}
                  c_r = \frac{(q^r-1)q^{r-1}}{q-1},
                  \qquad 1\le r\le D.
              \end{equation}
              With the boundary values $b_D=0$ and $c_0=0$, the remaining coefficient is
              \begin{equation}\label{eq:rank-a-coefficient}
                  a_r = \abs{\Gamma_1} - b_r - c_r,
                  \qquad 0\le r\le D.
              \end{equation}
    \end{enumerate}
\end{lemma}
\noindent The proof is in Appendix~\ref{subsec:proof-rank-metric-shell-sizes-intersection-numbers}.

\subsection{Rank-metric codes}
\label{subsec:rank-metric-codes}

For DQI, the relevant code operation is dual syndrome decoding. After the
Fourier phases have been encoded, the algorithm computes $\Phi^\top(e)$ and must
recover the low-rank error $e$ reversibly from this syndrome. We now specialise
the additive-code setup from
Sec.~\ref{subsec:additive-codes-dual-distance-orthogonality} to $\F_q$-linear
rank-metric codes and then recall a family where this decoding is efficient.

Fix $N=mn$ and a vector-space identification $\F_q^N\cong\F_q^{m\times n}$. A
linear rank-metric code is an $\F_q$-linear subspace $\Code\le\F_q^{m\times n}$
with $t=\dim_{\F_q}\Code$ and $\abs{\Code}=q^t$. Its annihilator is the usual
orthogonal dual with respect to $\bilinear{\cdot}{\cdot}$.

Choose a message space $\Omega\cong\F_q^t$ and an injective linear encoding
$\Phi\colon\Omega\to \Code\leq \F_q^{m\times n}$ with $\Code=\Im(\Phi)$. After
vectorisation, $\Phi$ is represented by a generator matrix $\mathcal
    G\in\F_q^{t\times N}$:
\begin{equation}
    \Code=\Set{x\mathcal G\given x\in\F_q^t}.
\end{equation}
With the row-vector convention, the dual code is
\begin{equation}
    \Code^\perp=\Set{z\in\F_q^N\given z\mathcal G^\top=0}.
\end{equation}

There are two syndrome maps in play. If $\mathcal H\in\F_q^{(N-t)\times N}$ is a
full-row-rank parity-check matrix for $\Code$, then the usual syndrome for
decoding $\Code$ is $s_\Code(y)=y\mathcal H^\top$. DQI instead needs to decode
the dual code $\Code^\perp$. Since $\mathcal G$ generates $\Code$, it is a
parity-check matrix for $\Code^\perp$, and the relevant syndrome of an error $e$
is
\begin{equation}
    \Phi^\top(e)=e\mathcal G^\top.
\end{equation}
This is exactly the adjoint encoding map from
Sec.~\ref{subsec:additive-codes-dual-distance-orthogonality}.

The minimum rank distance of $\Code$ is
\begin{equation}
    d_{\mathrm{rk}}(\Code)
    \coloneqq \min_{c\in\Code\setminus\Set{0}}\rank(c).
\end{equation}
The rank-metric DQI protocol assumes a polynomial-time unique decoder for
$\Code^\perp$: given $\Phi^\top(e)$ and $\rank(e)\le \ell$, it recovers $e$,
typically for $\ell\le\floor{(d_\mathrm{rk}^\perp-1)/2}$, where
$d_\mathrm{rk}^\perp=d_\mathrm{rk}(\Code^\perp)$. Generic rank-metric codes do
not provide such a decoder, so we use structured MRD families.

\paragraph{Maximum rank distance (MRD) codes.}

Let $d_\mathrm{rk}=d_{\mathrm{rk}}(\Code)$. The rank-metric Singleton bound is
\begin{equation}
    t \leq \max(m,n)(D-d_\mathrm{rk}+1).
\end{equation}
It follows by deleting $d_\mathrm{rk}-1$ columns if $n\le m$, or
$d_\mathrm{rk}-1$ rows if $m\le n$; see Appendix~\ref{subsec:deferred-rank-metric-code-facts}. A
code is MRD if it meets this bound with equality. In that case,
\begin{equation}
    d_\mathrm{rk}=D-\frac{t}{\max(m,n)}+1.
\end{equation}
Thus MRD codes maximise the rank distance, and hence the unique-decoding radius,
for a fixed ambient matrix size and code dimension. They are the rank-metric
analogue of MDS codes such as Reed--Solomon codes in the Hamming scheme.

\paragraph{Gabidulin codes.}

Gabidulin codes give explicit MRD codes whose rank-metric syndrome decoding can
be used for the reversible dual-decoding step in DQI. Assume $n\le m$ and
identify $\F_q^{m\times n}$ with $(\F_{q^m})^n$ using a fixed $\F_q$-basis of
$\F_{q^m}$. For an integer $k\le n$, a Gabidulin code has message space
$(\F_{q^m})^k$, so its $\F_q$-dimension in the DQI notation is $t=mk$. Choose
$\F_q$-linearly independent evaluation points
$\alpha_1,\ldots,\alpha_n\in\F_{q^m}$. A message
$\mathbf{a}=(a_0,\ldots,a_{k-1})\in(\F_{q^m})^k$ defines the linearised
polynomial
\begin{equation}
    f_{\mathbf{a}}(z)=\sum_{i=0}^{k-1}a_i z^{q^i},
\end{equation}
which is $\F_q$-linear as a map $\F_{q^m}\to\F_{q^m}$. The Gabidulin encoding is
\begin{equation}
    \Phi(\mathbf{a})=
    (f_{\mathbf{a}}(\alpha_1),\ldots,f_{\mathbf{a}}(\alpha_n)).
\end{equation}
Equivalently, over $\F_{q^m}$ this is multiplication by the Moore--Vandermonde
generator matrix
\begin{equation}
    \mathcal{G}_{\mathrm{Gab}} \coloneqq
    \begin{pmatrix}
        \alpha_1           & \alpha_2           & \cdots & \alpha_n           \\
        \alpha_1^q         & \alpha_2^q         & \cdots & \alpha_n^q         \\
        \vdots             &                    &        & \vdots             \\
        \alpha_1^{q^{k-1}} & \alpha_2^{q^{k-1}} & \cdots & \alpha_n^{q^{k-1}}
    \end{pmatrix}\in\F_{q^m}^{k\times n}.
\end{equation}
The proof of the distance statement is deferred to
Appendix~\ref{subsec:deferred-rank-metric-code-facts}.

\begin{lemma}[Gabidulin code properties~\cite{gabidulin1985theory}]\label{lem:gabidulin-code-properties}
    Let $\Code$ be an $[n,k]$ Gabidulin code over $\F_{q^m}$ with $n\le m$.
    Viewed as an $\F_q$-linear subspace of $\F_q^{m\times n}$,
    \begin{enumerate}
        \item $\dim_{\F_q}\Code=mk$;
        \item $d_{\mathrm{rk}}(\Code)=n-k+1$, so $\Code$ is MRD;
        \item $\Code^\perp$ is an $[n,n-k]$ Gabidulin code; and
        \item $\Code$ and $\Code^\perp$ admit polynomial-time decoding up to
              $\floor{(d_\mathrm{rk}-1)/2}$ rank errors, where $d_\mathrm{rk}$ is the relevant minimum
              rank distance.
    \end{enumerate}
\end{lemma}

Such decoders can be implemented using the extended Euclidean algorithm for
linearised polynomials or Berlekamp--Massey variants over the $q$-polynomial
ring~\cite{gabidulin1985theory,Richter2004}. Together,
Lemmas~\ref{lem:rank-metric-shell-sizes-intersection-numbers} and \ref{lem:gabidulin-code-properties}
provide the two ingredients required in Sec.~\ref{sec:dqi-rank-metric}: radial
shell data for the Jacobi analysis and an efficient reversible decoder for the
dual syndrome step when $\Code^\perp$ is chosen from this family.

\section{Rank-metric decoded quantum interferometry}
\label{sec:dqi-rank-metric}

This section presents a DQI algorithm based on association schemes and
rank-metric codes introduced in previous sections. Given a target matrix $y$ and
an encoding map $\Phi$, the goal is to find $x$ that minimises $\rank(\Phi(x) -
    y)$. The algorithm prepares a weighted superposition over ``error'' matrices
with maximal rank $\ell$, encodes $y$ into Fourier phases, uses a reversible
decoder of a dual code to uncompute the error register, and applies the quantum
Fourier transform. In the resulting quantum state, the amplitude of a candidate
$x$ depends only on the residual rank $\rank(\Phi(x)-y)$. The algorithmic
freedom appears in the amplitudes $w_0,\ldots,w_\ell$ on matrices of each rank
$0,\ldots,\ell$ in the initial superposition. They determine the bias of sampled
solutions toward low-rank residuals. We quantify this bias with a normalised
proxy score and show that the best shell weights are obtained by diagonalising
an $(\ell+1)\times(\ell+1)$ tridiagonal matrix acting on the weights vector
$(w_0,\ldots,w_\ell)$. Asymptotically, the optimised weights are supported, up
to negligible mass, on a constant-width window of rank shells immediately below
the cutoff $\ell$.

\subsection{Translation-scheme DQI protocol}
\label{subsec:translation-scheme-dqi-protocol}
We state the DQI protocol for the rank-metric translation scheme with diameter
$D$. Fix a rank-metric code $\Code\le \F_q^{m\times n}$ and a target word $y\in
    \F_q^{m\times n}$. Let $\Phi\colon \Omega \to \Code \leq \F_q^{m\times n}$ be an
injective evaluation map with $\Code=\Im(\Phi)$ and $\abs{\Omega}=\abs{\Code}$
(cf. Sec.~\ref{subsec:rank-metric-codes}). The rank-metric nearest-codeword
problem is
\begin{equation}\label{eq:rank-metric-nearest-codeword}
    \mathrm{OPT}(y)\coloneqq \min_{c\in\Code}\rank(c-y) = \min_{x \in \Omega}\rank(\Delta_y(x))
\end{equation}
with the residual $\Delta_y(x)$ defined in Eq.~\eqref{eq:residual-map}.

\paragraph{Motivation for rank-metric objective.}
The objective in Eq.~\eqref{eq:rank-metric-nearest-codeword} is natural when
errors are correlated by a low-dimensional row or column structure. In
crisscross models, $\rank(y-c)$ measures the dimension of the row or column
subspace needed to explain the residual $y-c$~\cite{roth91}. Similarly, if
$a,b\in\mathrm{Mat}_{m\times n}(\F_q)$ are viewed as maps $\F_q^n\to\F_q^m$,
then
\begin{equation}
    \dim\bigl(\operatorname{graph}(a)\cap\operatorname{graph}(b)\bigr)
    =
    n-\rank(a-b).
\end{equation}
Thus minimising residual rank is equivalent to maximising agreement on a large
subspace. These examples motivate the objective, but we do not claim that DQI
improves over existing classical methods.

\paragraph{Decoder and shell cutoff.}
Fix a shell cutoff $\ell$. The coherent decoding step requires a reversible
unique decoder for $\Code^\perp$. From $\Phi^\top(e)$ it must recover every
error $e$ with $\rank(e)\le\ell$. For Gabidulin dual codes this holds, for
example, when
\begin{equation}
    \ell\le
    \min\!\left(D,\left\lfloor\frac{d_\mathrm{rk}^\perp-1}{2}\right\rfloor\right),
\end{equation}
where $\ell\le D$ ensures that the shells $\Gamma_0,\ldots,\Gamma_\ell$ exist.

The protocol aims to prepare candidate amplitudes of the form
\begin{equation}
    P_\ell(r)
    =
    \sum_{k=0}^{\ell}
    \frac{w_k}{\sqrt{\abs{\Gamma_k}}}\lambda_k(r),
    \qquad
    \sum_{k=0}^{\ell}\abs{w_k}^2=1.
\end{equation}
The weights $w_k$ are chosen later to bias measurement toward low-rank
residuals. The remaining implementation ingredient is an efficient rank-shell
state preparation given in the following lemma.

\begin{lemma}[Rank-shell Dicke state preparation]
    \label{lem:rank-shell-dicke-preparation}
    For every $0\le r\le D=\min(m,n)$, there is a unitary quantum circuit which,
    using reversible $\F_q$-arithmetic, controlled single-register rotations,
    and clean uniform $\F_q$-register preparations, prepares the rank-$r$
    shell Dicke state
    \begin{equation}
        \ket{R_r}
        =
        \frac{1}{\sqrt{\abs{\Gamma_r}}}
        \sum_{e\in\Gamma_r}\ket{e}.
    \end{equation}
    A controlled implementation of $\ket{r}\ket{0}\mapsto \ket{r}\ket{R_r}$ for
    all $0\le r\le \ell$ has complexity $\mathscr{O}(mn\ell)$ reversible
    $\F_q$-field operations, plus $\mathscr{O}(mn\ell)$ controlled
    state-preparation primitives. All required rotation angles are classically
    computable from $(q,m,n,\ell)$. With standard finite-field arithmetic and
    rotation synthesis, an $\varepsilon$-approximation has size
    $\tilde{\mathscr{O}}(mn\ell\,\polylog(q/\varepsilon))$ elementary gates.
\end{lemma}
\noindent The proof is in
Appendix~\ref{subsec:proof-rank-shell-dicke-preparation}.

The protocol prepares the states $\ket{\psi_1}\to\ket{\psi_6}$ as follows.

\vspace{\topsep}
\noindent\textbf{Step 1.} Initialise the weight register of $\ceil{\log_2(\ell+1)}$ qubits. This can be done with $\tilde{\mathscr{O}}(\ell)$ gates.
\begin{align}
    \ket{\psi_1}
     & = \sum_{r=0}^{\ell} w_r \ket{r}
    \notag                                      \\
    \shortintertext{\textbf{Step 2.} Conditioned on the weight register, prepare the
        corresponding rank-shell state using
        Lemma~\ref{lem:rank-shell-dicke-preparation}, then uncompute the weight register
        by coherently computing $\rank(e)$. This costs
        $\tilde{\mathscr{O}}(mn\ell\polylog(q))$ gates.}
    \ket{\psi_2'}
     & = \sum_{k=0}^{\ell} w_k \ket{k}\ket{R_k}
    = \sum_{k=0}^{\ell} \frac{w_k}{\sqrt{\abs{\Gamma_k}}} \ket{k}
    \sum_{e\in\Gamma_k}\ket{e},\notag           \\
    \ket{\psi_2}
     & = \sum_{k=0}^{\ell}
    \frac{w_k}{\sqrt{\abs{\Gamma_k}}}
    \sum_{e\in\Gamma_k}\ket{e}.
    \notag                                      \\
    \shortintertext{\textbf{Step 3.} Encode the target matrix into phases using the trace-pairing character
        $\adchar_y(-e)=\chi(-\Tr(ey^\top))$. This uses $\mathscr{O}(mn\polylog(q))$ phase gates.}
    \ket{\psi_3}
     & = \sum_{k=0}^{\ell}
    \frac{w_k}{\sqrt{\abs{\Gamma_k}}}
    \sum_{e\in\Gamma_k}\adchar_y(-e)\ket{e}.
    \notag                                      \\
    \shortintertext{\textbf{Step 4.} Compute the dual syndrome
        $\Phi^\top(e)=e\mathcal G^\top$ into a syndrome register. For a rank-metric code
        with $\dim_{\F_q}(\Code)=t$, matrix multiplication costs
        $\tilde{\mathscr{O}}(tmn\polylog(q))$ gates. For Gabidulin codes, $t=mk$, giving
        $\tilde{\mathscr{O}}(m^2kn\polylog(q))$ gates.}
    \ket{\psi_4}
     & = \sum_{k=0}^{\ell}
    \frac{w_k}{\sqrt{\abs{\Gamma_k}}}
    \sum_{e\in\Gamma_k}\adchar_y(-e)
    \ket{e}\ket{\Phi^\top(e)}.
    \notag                                      \\
    \shortintertext{\textbf{Step 5.} Use the reversible dual decoder to recover $e$
        from $\Phi^\top(e)$ and uncompute the error register. In general this costs
        $T_\mathrm{dec}(m,n,q,\ell)$. For Gabidulin codes, dual decoding is polynomial
        time, \eg{} $\tilde{\mathscr{O}}(n^2\poly(m,\log q))$ gates for a
        quadratic-time decoder.}
    \ket{\psi_5}
     & = \sum_{k=0}^{\ell}
    \frac{w_k}{\sqrt{\abs{\Gamma_k}}}
    \sum_{e\in\Gamma_k}\adchar_y(-e)
    \ket{\Phi^\top(e)}.
    \notag                                      \\
    \shortintertext{\textbf{Step 6.} Apply the Fourier transform over $\Omega$ with cardinality $\abs{\Omega}=q^t$,
        costing roughly $\tilde{\mathscr{O}}(t\log q)$ gates. The third line uses the defining adjoint relation for $\Phi^\top$.}
    \ket{\psi_6}
     & =
    \sum_{k=0}^{\ell}\frac{w_k}{\sqrt{\abs{\Gamma_k}}}
    \sum_{e\in\Gamma_k}\adchar_y(-e)
    \,\mathcal F_{\Omega}\ket{\Phi^\top(e)}
    \notag                                      \\
     & =
    \sum_{k=0}^{\ell}
    \frac{w_k}{\sqrt{\abs{\Omega}\abs{\Gamma_k}}}
    \sum_{x\in\Omega}\sum_{e\in\Gamma_k}
    \adchar_x(\Phi^\top(e))
    \adchar_y(-e)\ket{x}
    \notag                                      \\
     & =
    \sum_{k=0}^{\ell}
    \frac{w_k}{\sqrt{\abs{\Omega}\abs{\Gamma_k}}}
    \sum_{x\in\Omega}\sum_{e\in\Gamma_k}
    \adchar_{\Phi(x)}(e)
    \adchar_y(-e)\ket{x}
    \notag                                      \\
     & =
    \sum_{k=0}^{\ell}
    \frac{w_k}{\sqrt{\abs{\Omega}\abs{\Gamma_k}}}
    \sum_{x\in\Omega}\sum_{e\in\Gamma_k}
    \adchar_{\Delta_y(x)}(e)\ket{x}
    \notag                                      \\
     & =
    \sum_{k=0}^{\ell}
    \frac{w_k}{\sqrt{\abs{\Omega}\abs{\Gamma_k}}}
    \sum_{x\in\Omega}\lambda_k(\Delta_y(x))\ket{x}.
    \label{eq:dqi-final-state}
\end{align}
\vspace{\topsep}

The rank-metric translation scheme satisfies
Lemma~\ref{lem:radial-fourier-eigenvalue-conditions} so its Fourier eigenvalues
$\lambda_k$ are radial. Hence $\lambda_k(\Delta_y(x))$ depends only on
$\rank(\Delta_y(x))$, and the amplitude of a candidate $x\in\Omega$ is a
function only of its residual rank. Therefore, the final state can be written as
\begin{equation}\label{eq:dqi-final-radial-state}
    \begin{split}
        \ket{\psi_6}
         & =
        \sum_{k=0}^{\ell}
        \frac{w_k}{\sqrt{\abs{\Omega}\abs{\Gamma_k}}}
        \sum_{x\in\Omega}
        \lambda_k\!\left(\rank(\Delta_y(x))\right) \ket{x} \\
         & =
        \frac{1}{\sqrt{\abs{\Omega}}}
        \sum_{x\in\Omega}
        P_\ell\left(\rank(\Delta_y(x))\right)\ket{x}.
    \end{split}
\end{equation}
By the three-term recurrence \eqref{eq:character-sum-three-term-recurrence},
which follows from the $P$-polynomial structure, each $\lambda_k(r)$ is a
polynomial of degree $k$ in $\lambda_1(r)$. Hence $P_\ell$ is a polynomial of
degree at most $\ell$ in $\lambda_1(r)$. Measuring the final state in the
computational basis samples $x$ with probability proportional to
$\abs{P_\ell\left(\rank(\Delta_y(x))\right)}^2$.

Although we have written the protocol for the rank-metric scheme, the same
mechanism applies to general translation schemes with shell states whose Fourier
transforms depend only on distance, efficient shell-state preparation, and an
efficient reversible decoder for the dual syndrome map. Radiality is what makes
the final amplitude depend only on the distance of the residual, rather than on
the residual itself. In the rank-metric case, radiality follows from
Lemma~\ref{lem:radial-fourier-eigenvalue-conditions}, shell-state preparation is
provided by Lemma~\ref{lem:rank-shell-dicke-preparation}, and the decoder is
provided by the code family chosen in Sec.~\ref{subsec:rank-metric-codes}.

\subsection{Proxy scores and DQI performance}
\label{subsec:proxy-scores-dqi-performance}

This subsection quantifies how much the final DQI state favours low-rank
residuals. We use the monotone score $q^{-\rank(z)}$, normalised to have mean
zero and unit variance over all matrices. After centring and scaling, this score
is exactly the first rank-shell Fourier mode, so its expectation in the DQI
state depends only on the shell weights. The distance condition that makes dual
decoding unique also makes cross-terms between different prepared shells vanish,
Lemma~\ref{lem:shell-orthogonality-dual-distance}. This reduces the expectation
of the proxy-score observable to a quadratic form in a tridiagonal submatrix of
the Jacobi operator $\tilde{A}$ acting on the weight vector
$(w_0,\ldots,w_\ell)$.

We first express the monotone rank score $q^{-\rank(z)}$ through the first
radial Fourier eigenvalue.

\begin{lemma}[Rank score as a first Fourier mode]
    \label{lem:rank-score-first-fourier-mode}
    For every $z\in\F_q^{m\times n}$,
    \begin{equation}
        \label{eq:q-minus-rank-character-score}
        q^{-\rank(z)}
        =
        q^{-m}+q^{-n}-q^{-(m+n)}
        +
        \frac{q-1}{q^{m+n}}\lambda_1(\rank(z)).
    \end{equation}
\end{lemma}
\noindent Since $q^{-\rank(z)}$ is strictly decreasing in $\rank(z)$, minimising
residual rank is equivalent to maximising $q^{-\rank(\Delta_y(x))}$, or
equivalently $\lambda_1(\rank(\Delta_y(x)))$ by
Lemma~\ref{lem:rank-score-first-fourier-mode}. The proof is in
Appendix~\ref{subsec:proof-rank-score-first-fourier-mode}.

The next lemma gives the uniform mean and variance of this rank score. These
moments are used below to turn $q^{-\rank(z)}$ into a normalised proxy score
whose scale is comparable across $m,n,q$.

\begin{lemma}[Uniform expectation and variance of the rank score]
    \label{lem:rank-score-uniform-moments}
    The average and variance of $q^{-\rank(z)}$ over the uniform measure on
    $\F_q^{m\times n}$ are given by
    \begin{equation}
        \begin{split}
            \mathbb{E}_{z\sim\text{Unif}} \left[q^{-\rank(z)}\right]
             & =
            q^{-m}+q^{-n}-q^{-(m+n)}\, , \\
            \Var_{z\sim\text{Unif}} \left[q^{-\rank(z)} \right]
             & =
            (q-1)\,q^{-(m+n)}\,(1-q^{-m})(1-q^{-n}).
        \end{split}
    \end{equation}
\end{lemma}
\noindent The proof is in
Appendix~\ref{subsec:proof-rank-score-uniform-moments}.

Define the normalised proxy score
\begin{equation}
    O_q(z)
    \coloneqq
    \frac{q^{-\rank(z)} - \mathbb{E}_{z\sim \text{Unif}}[q^{-\rank(z)}]}{\sqrt{\Var_{z\sim \text{Unif}}[q^{-\rank(z)}]}}.
\end{equation}
The corresponding proxy-score objective is
\begin{equation}\label{eq:rank-proxy-objective}
    \mathrm{OPT}_q(y)
    =
    \max_{x \in \Omega} \left[O_q(\Delta_y(x)) \right]
    =
    \frac{1}{\sqrt{\abs{\Gamma_1}}} \max_{x \in \Omega} \left[ \lambda_1(\rank(\Delta_y(x))) \right],
\end{equation}
where $\abs{\Gamma_1} = \frac{(q^m-1)(q^n-1)}{q-1}$ for the rank metric case.
This proxy-score objective is represented by the diagonal proxy-score observable
\begin{equation}
    \label{eq:rank-proxy-observable}
    \mathcal{O}_q
    =
    \frac{1}{\sqrt{\abs{\Gamma_1}}}
    \sum_{x\in\Omega}
    \lambda_1\!\left(\rank(\Delta_y(x)) \right) \ket{x}\bra{x}.
\end{equation}

This proxy-score observable depends on $y$ only through residual ranks. The
final state from Eq.~\eqref{eq:dqi-final-radial-state} contains only $\lambda_k$
with $k\le\ell$. Therefore products of amplitudes involve precisely the
low-shell sums controlled by Lemma~\ref{lem:shell-orthogonality-dual-distance}.
Under the dual-distance assumptions in the next lemma, these sums collapse to
the upper $(\ell+1)\times(\ell+1)$ block $\tilde A^{(\ell)}$ of the Jacobi
matrix $\tilde A$ from Eq.~\eqref{eq:radial-jacobi-matrix}.

\begin{lemma}[Expectation and variance of the rank-metric proxy-score observable]
    \label{lem:rank-metric-proxy-score-observable-moments}
    Assume $\ell\le D$. If $2\ell+1<d^\perp$, then the expectation of
    $\mathcal{O}_q$ in the final DQI state is
    \begin{equation}
        \mathbb{E}_{\psi_6}[\mathcal{O}_q]
        \coloneqq
        \frac{\bra{\psi_6}\mathcal{O}_q\ket{\psi_6}}{\braket{\psi_6}{\psi_6}}
        =
        \frac{1}{\sqrt{\abs{\Gamma_1}}}
        w^\dagger \tilde{A}^{(\ell)} w
        \label{eq:dqi-observable-expectation-jacobi}
    \end{equation}
    If $2\ell+2<d^\perp$, then the variance of $\mathcal{O}_q$ in the final DQI
    state is
    \begin{equation}\label{eq:dqi-observable-variance-jacobi}
        \begin{split}
            \operatorname{Var}_{\psi_6}[\mathcal{O}_q]
             & \coloneqq
            \frac{\bra{\psi_6}\mathcal{O}_q^2\ket{\psi_6}}{\braket{\psi_6}{\psi_6}}
            -
            \mathbb{E}_{\psi_6}^2 [\mathcal{O}_q] \\
             & =
            \frac{1}{\abs{\Gamma_1}}
            \left(
            w^\dagger \left(\tilde{A}^{(\ell)}\right)^2 w
            - \left(w^\dagger \tilde{A}^{(\ell)} w\right)^2
            + \abs{w_\ell}^2 b_{\ell}c_{\ell+1}
            \right),
        \end{split}
    \end{equation}
    where $\abs{\Gamma_1} = \frac{(q^m-1)(q^n-1)}{q-1}$.
\end{lemma}
\noindent The proof is in Appendix~\ref{subsec:proof-rank-metric-proxy-score-observable-moments}.

The derivation of Eq.~\eqref{eq:dqi-observable-expectation-jacobi} differs
slightly from Ref.~\cite{jordan25:DQI-original}. Here we compute the Fourier
transform on the final state \eqref{eq:dqi-final-state} explicitly, whereas
Ref.~\cite{jordan25:DQI-original} applies it to the proxy-score observable using
generalised $X$ and $Z$ operators. The same method applies to diagonal
observables given by other linear combinations of the radial Fourier
eigenvalues, producing a corresponding matrix $\hat A^{(\ell)}$.

The normalisation of the proxy score $O_q$ uses the uniform measure on
$\F_q^{m\times n}$ only to fix a scale. The DQI state samples the affine coset
$\Code-y$, but under the low-shell orthogonality assumptions the first two coset
moments agree with the full matrix-space moments.

\paragraph{Relation to Hamming DQI.}
The original DQI in Ref.~\cite{jordan25:DQI-original} can be viewed as DQI with
the Hamming scheme. The Hamming metric $d_H$ appears inside the rank metric by
restricting residuals to diagonal matrices: for all $u,v\in\F_q^n$,
\begin{equation}
    d_{\mathrm{rk}}(\operatorname{diag}(u),\operatorname{diag}(v))
    =
    d_H(u,v).
\end{equation}
If $A_r^{\mathrm{rk}}$ denotes the rank-metric distance matrix on
$\F_q^{n\times n}$ and $A_r^H$ denotes the Hamming distance matrix on
$\F_q^n$, then
\begin{equation}
    (A_r^{\mathrm{rk}})_{\operatorname{diag}(u),\operatorname{diag}(v)}
    =
    (A_r^H)_{u,v},
    \qquad u,v\in\F_q^n,\ r\in[0,n].
\end{equation}
In this restricted metric sense, the singleton-target Hamming nearest-codeword
version of DQI is recovered by restricting the ambient matrix space, residuals,
and shell states to diagonal matrices. This includes the max-XORSAT case, where
the target sets are singletons. It does not recover the full max-LINSAT setting
of Ref.~\cite{jordan25:DQI-original}, where the target is a product set
$F_1\times\cdots\times F_m$ rather than a single received word. Extending the
rank-metric objective to product target sets is a separate issue.

This does not mean that rank-metric codes become Hamming-metric codes under this
restriction. The point is only that the common mechanism is the association
scheme structure: prepare radial shell superpositions, encode the problem in
additive-character phases, use reversible dual decoding on low shells, and apply
the quantum Fourier transform, which produces amplitudes depending only on the
residual distance.

\subsection{Rank-shell weights and the ladder factorisation}
\label{subsec:rank-shell-weights-ladder-factorisation}

We now return to the choice of rank-shell weights. By
Lemma~\ref{lem:rank-metric-proxy-score-observable-moments}, if $2\ell+1<d^\perp$, then
\begin{equation}
    \mathbb{E}_{\psi_6}[\mathcal{O}_q]
    =
    \frac{1}{\sqrt{\abs{\Gamma_1}}}
    w^\dagger \tilde A^{(\ell)}w .
\end{equation}
Thus optimising the initial error superposition amounts to maximising this
quadratic form over normalised $w=(w_0,\ldots,w_\ell)$. For the rank scheme,
$\tilde A^{(\ell)}$ is tridiagonal and admits a bidiagonal factorisation, which
equivalently turns the problem into minimising $\|B^{(\ell)}w\|^2$. On the full
radial basis using coefficients from
Lemma~\ref{lem:recursion-coefficient-identities}, define the ladder operators
\begin{equation}
    B\ket{R_i}
    =
    \sqrt{b_i}\ket{R_i}
    -
    \sqrt{c_i}\ket{R_{i-1}},
    \qquad
    B^\dagger\ket{R_i}
    =
    \sqrt{b_i}\ket{R_i}
    -
    \sqrt{c_{i+1}}\ket{R_{i+1}},
\end{equation}
with terms outside $[0,D]$ omitted. Thus $B$ and $B^\dagger$ only couple
neighbouring rank shells.

The truncated operator $B^{(\ell)}$ is the upper-left $(\ell+1)\times(\ell+1)$
block acting on $\Span\Set{\ket{R_i}\given i\in[0,\ell]}$. Using
Eq.~\eqref{eq:ratio-off-diagonal} and $a_k=\abs{\Gamma_1}-b_k-c_k$ from
Eq.~\eqref{eq:rank-a-coefficient}, the truncated Jacobi matrix has entries
\begin{equation}
    \begin{split}
        (\tilde A^{(\ell)})_{kk}
         & =
        \abs{\Gamma_1}-b_k-c_k,
        \qquad 0\le k\le \ell, \\
        (\tilde A^{(\ell)})_{k,k+1}
        =
        (\tilde A^{(\ell)})_{k+1,k}
         & =
        \sqrt{b_kc_{k+1}},
        \qquad 0\le k<\ell,
    \end{split}
\end{equation}
and all other entries are zero. This matrix admits the following ladder
factorisation. The upper-bidiagonal operator $B^{(\ell)}$ has diagonal entries
$(\sqrt{b_k})_{k=0}^{\ell}$ and upper diagonal entries
$(-\sqrt{c_k})_{k=1}^{\ell}$, and $(B^{(\ell)})^\dagger$ is its adjoint:
\begin{equation}
    B^{(\ell)} \;=\;
    \begin{pmatrix}
        \sqrt{b_0} & -\sqrt{c_1} &             &        &                \\
                   & \sqrt{b_1}  & -\sqrt{c_2} &        &                \\
                   &             & \sqrt{b_2}  & \ddots &                \\
                   &             &             & \ddots & -\sqrt{c_\ell} \\
                   &             &             &        & \sqrt{b_\ell}
    \end{pmatrix}, \quad (B^{(\ell)})^\dagger \;=\;
    \begin{pmatrix}
        \sqrt{b_0}  &             &            &                &               \\
        -\sqrt{c_1} & \sqrt{b_1}  &            &                &               \\
                    & -\sqrt{c_2} & \sqrt{b_2} &                &               \\
                    &             & \ddots     & \ddots
                    &                                                           \\
                    &             &            & -\sqrt{c_\ell} & \sqrt{b_\ell}
    \end{pmatrix}.
\end{equation}
Recall the boundary values $b_{-1}=c_{D+1}=b_D=c_0=0$. A direct computation
gives
\begin{equation}\label{eq:jacobi-ladder-factorisation}
    \tilde A^{(\ell)} = \abs{\Gamma_1}\,I - (B^{(\ell)})^\dagger B^{(\ell)}.
\end{equation}
Thus maximising the expected proxy-score observable is equivalent to minimising
$\|B^{(\ell)}w\|^2$ over normalised rank-shell weight vectors $w$.

Combining Lemma~\ref{lem:rank-metric-proxy-score-observable-moments} with the factorisation above
and $\|w\|^2=1$, we obtain the equivalent representation
\begin{equation}
    \label{eq:dqi-observable-expectation-ladder}
    \mathbb{E}_{\psi_6}[\mathcal{O}_q]
    =
    \frac{\abs{\Gamma_1}-\|B^{(\ell)} w\|^2}{\sqrt{\abs{\Gamma_1}}}.
\end{equation}
Thus maximising the expected proxy-score observable is equivalent to minimising
$\|B^{(\ell)}w\|^2$ over normalised radial weight vectors $w$.

Since $\mathcal{O}_q$ represents the normalised proxy score associated with
$q^{-\rank(\Delta_y(x))}$, we package the achieved expected unnormalised score
as the \emph{effective-rank proxy}
\begin{equation}
    R_\ell(w)\coloneqq
    -\log_q\!\left(
    1-\frac{q-1}{q^{m+n}}\,
    \|B^{(\ell)} w\|^2
    \right).
\end{equation}
Thus $q^{-R_\ell(w)}$ is the expected unnormalised score
$q^{-\rank(\Delta_y(\tilde{x}))}$ for a measured DQI output $\tilde{x}$, while
$R_\ell(w)$ itself should not be read as the expected residual rank.

The next subsection bounds the relevant smallest eigenvalue, interprets the
optimised effective-rank proxy, and converts it into a statement about the
residual rank of a sampled candidate.

\subsection{Performance bounds with optimised rank-shell weights}
\label{subsec:performance-bounds-optimised-rank-shell-weights}

To optimise the rank-shell weights, choose $w^{(\ell)}$ as a normalised
nonnegative eigenvector associated with the largest eigenvalue of $\tilde
    A^{(\ell)}$. By Eq.~\eqref{eq:jacobi-ladder-factorisation}, this is equivalent
to choosing $w^{(\ell)}$ as the eigenvector of $(B^{(\ell)})^\dagger B^{(\ell)}$
associated with its smallest eigenvalue. Let $\mu^{(\ell)}_{\min}$ and
$\mu^{(\ell)}_{\max}$ denote the smallest and largest eigenvalues of
$(B^{(\ell)})^\dagger B^{(\ell)}$ after multiplication by the normalisation
factor $(q-1)/q^{m+n}$. The optimised radial weights therefore give
\begin{equation}
    R_\ell(w^{(\ell)})= - \log_q(1-\mu^{(\ell)}_{\min}).
\end{equation}
Since $-\log_q(1-\mu)$ is increasing in $\mu$, every normalised $w$ satisfies
\begin{equation}
    \label{eq:rank-proxy-ladder-bounds}
    -\log_q(1-\mu^{(\ell)}_{\min})
    \leq R_\ell(w) \leq
    -\log_q(1-\mu^{(\ell)}_{\max}).
\end{equation}
It remains to bound $\mu^{(\ell)}_{\min}$. The next lemma gives the two
estimates that control the optimised effective-rank proxy.

\begin{lemma}[Spectral bounds for the truncated rank Jacobi operator]
    \label{lem:truncated-rank-jacobi-spectral-bounds}
    Let $D=\min(m,n)$. For every prime power $q$, all $m,n\ge 1$, and
    $0\le\ell<D$,
    \begin{equation}
        \label{eq:truncated-jacobi-lower-eigenvalue-bound}
        \left(
        \sqrt{(1-q^{\ell-m})(1-q^{\ell-n})}
        -
        \sqrt{(q^\ell-1)q^{\,\ell-1-m-n}}
        \right)^2
        \leq
        \mu_{\min}^{(\ell)}
        \leq
        1-q^{\ell-D}.
    \end{equation}
\end{lemma}

The upper bound comes from viewing $\tilde A^{(\ell)}$ as a principal submatrix
of the full rank-scheme Jacobi operator. The lower bound uses the bidiagonal
ladder factorisation from Sec.~\ref{subsec:rank-shell-weights-ladder-factorisation}. The
proof is in Appendix~\ref{subsec:proof-truncated-rank-jacobi-spectral-bounds}.

\paragraph{Large-system interpretation.}
By Lemma~\ref{lem:truncated-rank-jacobi-spectral-bounds} and
Eq.~\eqref{eq:rank-proxy-ladder-bounds}, the optimised effective-rank proxy is
controlled by Eq.~\eqref{eq:truncated-jacobi-lower-eigenvalue-bound}. We now
study this bound under the fixed aspect-ratio scaling
\begin{equation}
    n,m \to \infty,
    \qquad
    \frac{\ell}{D}=\theta <1,
    \qquad
    \frac{n}{m}=\rho \in (0,\infty).
\end{equation}
assuming that the field size $q$ is fixed. Here the case $\rho=1$ in the display
below refers to the exact square sequence $m=n$; near-square scalings with
$n/m\to1$ but $m\ne n$ have constants depending on the rectangular offset, as in
the boundary-profile analysis below. In this case the left side of the
inequality \eqref{eq:truncated-jacobi-lower-eigenvalue-bound} converges at
leading order to
\begin{equation}\label{eq:large-system-truncated-jacobi-eigenvalue-bound}
    \left(
    \sqrt{(1-q^{\ell-m})(1-q^{\ell-n})}
    -
    \sqrt{(q^\ell-1)q^{\,\ell-1-m-n}}
    \right)^2 \to
    \begin{cases}
        1-q^{\ell-D}, \quad              & \rho \neq 1 \\
        1-2(1+q^{-1/2})q^{\ell-D}, \quad & \rho=1
    \end{cases}
\end{equation}
Thus the optimised effective-rank proxy is confined near $D-\ell$, with only a
constant-width correction in the square-matrix case:
\begin{equation}
    D-\ell - \log_q(2+2q^{-1/2})\delta_{\rho,1}
    \leq R_\ell(w^{(\ell)}) \leq D-\ell.
\end{equation}
Note that the correction $\log_q(2+2q^{-1/2})$ lies between $0$ and $2$, so the
asymptotic interval has width at most $2$.

In Hamming DQI, the semicircle law estimates how far the optimised radial
weights can bias samples toward satisfying more coordinate constraints; in the
singleton-target case, this is bias toward small Hamming residuals. In
rank-metric DQI, the analogous large-system estimate is
Eq.~\eqref{eq:large-system-truncated-jacobi-eigenvalue-bound}. The bias toward
low rank is controlled by the $q$-geometric rank-shell recurrence coefficients
rather than by a semicircle law.

\paragraph{From expected proxy score to sampled rank.}
The effective-rank proxy estimate also gives a sampling statement. Let
$\tilde{x}\in\Omega$ be the measured DQI output under the optimised weights, and
set $R=\rank(\Delta_y(\tilde{x}))$. Under the slightly stronger condition
$2\ell+2<d^\perp$, Appendix~\ref{sec:expected-proxy-score-rank-tail-bound} shows
that, for every $a>0$,
\begin{equation}
    \Pr\!\left[R\le D-\ell+a\right]
    \ge
    \frac{(1-q^{-a})^2}{1+q}.
\end{equation}
Thus, for fixed $q$ and $a$, the effective-rank proxy corresponds to a
constant-probability bound on the actual sampled residual rank. This is an
absolute rank guarantee, not an approximation theorem relative to
$\mathrm{OPT}(y)$. For Gabidulin nearest-codeword instances, the
Appendix~\ref{sec:expected-proxy-score-rank-tail-bound} also shows why this should not be
read as an OPT-relative approximation guarantee: the dual unique-decoding radius
limits the cutoff $\ell$, while the primal covering radius can make
$\mathrm{OPT}(y)$ substantially smaller than $D-\ell$.

The next subsection shows that the optimising weights, asymptotically,
concentrate their mass in a constant-width boundary window below the cutoff
$\ell$.

\subsection{Concentration and truncation of optimised rank-shell weights}
\label{subsec:concentration-truncation-optimised-rank-shell-weights}

The performance bounds above determine the value and sampling interpretation of
the optimised effective-rank proxy. We now describe the shape of the optimising
vector $w^{(\ell)}$. In this subsection we assume $n\le m$ and set $D=n$; the
case $m<n$ follows by transposing matrices. The optimised weights are best
viewed from the moving cutoff shell. Define $u_r^{(\ell)}\coloneqq
    w_{\ell-r}^{(\ell)}$ for $0\le r\le\ell$, and set $u_r^{(\ell)}=0$ for $r>\ell$.
Thus $r=0$ is the cutoff shell $\ell$, and increasing $r$ moves inward. In these
reversed coordinates, the limiting profile near the cutoff is obtained by
zooming into the deficit of the normalised ladder Hamiltonian
\begin{equation}
    \bar L^{(\ell)}
    \coloneqq
    \frac{q-1}{q^{m+n}}(B^{(\ell)})^\dagger B^{(\ell)}.
\end{equation}
If $U_\ell$ reverses the truncated shell basis, then the proof shows, in the
operator-norm sense made precise in the
Appendix~\ref{subsec:proof-boundary-concentration-optimised-rank-shell-weights},
that
\begin{equation}
    q^{D-\ell}U_\ell(I-\bar L^{(\ell)})U_\ell^\dagger
    \to
    \mathcal K_s.
\end{equation}
For fixed rectangular offset $s=m-n$, the limiting cutoff profile is governed by
the compact Jacobi operator
\begin{equation}
    (\mathcal K_s u)_r
    =
    (1+q^{-s})q^{-r}u_r
    +
    q^{-(s+1)/2}q^{-(r-1)}u_{r-1}
    +
    q^{-(s+1)/2}q^{-r}u_{r+1},
    \qquad u_{-1}=0.
\end{equation}
Its coefficients decay geometrically as $r$ moves inward, so finite boundary
windows approximate $\mathcal K_s$ in operator norm.

\begin{lemma}[Boundary concentration of optimised rank-shell weights]
    \label{lem:boundary-concentration-optimised-rank-shell-weights}
    Assume $n\le m$ and set $D=n$. If $\ell\to\infty$, $D-\ell\to\infty$, and
    $m-n=s$ is fixed, then the reversed profiles $u^{(\ell)}$ of the optimised
    weights defined above converge in $\ell^2(\mathbb N_0)$ to the normalised
    positive top eigenvector of $\mathcal K_s$.

    Consequently, there is a constant $C_s>0$, depending only on $q$ and $s$,
    such that, for all sufficiently large $\ell$ and every integer $0\le
        J\le\ell$,
    \begin{equation}
        \sum_{i=0}^{\ell-J-1}\abs*{w_i^{(\ell)}}^2
        \le
        C_s q^{-2J}.
    \end{equation}
    Thus keeping the final $J=O(\log_q(1/\varepsilon))$ rank shells loses at
    most $\varepsilon$ mass for all sufficiently large $\ell$. If, along
    another sequence, $m-n\to\infty$ while $\ell\to\infty$ and
    $D-\ell\to\infty$, then
    $\abs*{w_\ell^{(\ell)}}^2\to1$.
\end{lemma}
\noindent The proof is in
Appendix~\ref{subsec:proof-boundary-concentration-optimised-rank-shell-weights}.

The profile estimate also gives a stability statement for the effective-rank
proxy when the optimised weights are truncated to a boundary window. For a fixed
integer $J\ge0$ and all sufficiently large $\ell$, let $P_J^{(\ell)}$ be the
coordinate projection onto the final $J+1$ shell coordinates
$\ell-J,\ldots,\ell$, and define
\begin{equation}
    \widehat w^{(\ell,J)}
    \coloneqq
    \frac{P_J^{(\ell)}w^{(\ell)}}{\|P_J^{(\ell)}w^{(\ell)}\|_2}.
\end{equation}

\begin{corollary}[Effective-rank proxy stability under boundary truncation]
    \label{cor:effective-rank-proxy-boundary-truncation}
    Assume the setting of
    Lemma~\ref{lem:boundary-concentration-optimised-rank-shell-weights} with
    $m-n=s$ fixed. There is a constant $C_s'>0$, depending only on $q$ and $s$,
    such that, for all sufficiently large $\ell$ and every integer $0\le
        J\le\ell$,
    \begin{equation}
        0
        \le
        R_\ell(\widehat w^{(\ell,J)})-R_\ell(w^{(\ell)})
        \le
        C_s' q^{-2J}.
    \end{equation}
    Hence truncating to the final $J=O(\log_q(1/\varepsilon))$ rank shells
    increases the optimised effective-rank proxy by at most $O(\varepsilon)$,
    uniformly in the boundary-window width.
\end{corollary}
\noindent The proof is in
Appendix~\ref{subsec:proof-effective-rank-proxy-boundary-truncation}.

The boundary-profile estimate differs from the corresponding analysis in the
Hamming case of Ref.~\cite{jordan25:DQI-original}, where a flat trial vector on
a $\sqrt{\ell}$-wide boundary window suffices for the eigenvalue asymptotics.
Here the $q$-geometric coefficients produce a constant-width boundary profile.

\section{Conclusion and outlook}
\label{sec:conclusion-outlook}

We have shown that translation association schemes provide a natural language
for separating the core DQI mechanism from the special geometry of Hamming
space. The common DQI routine for singleton-target optimisation problems
remains: preparing weighted superpositions of states with fixed weights (shell
Dicke states), encoding the optimisation target into Fourier phases, coherently
decoding a dual code on low shells, and applying a quantum Fourier transform. In
a translation scheme, this produces a quantum state with amplitudes determined
only by the residual distance. For $P$-polynomial schemes and reversible
decoders, DQI performance analysis reduces to a finite tridiagonal problem on
shells.

As an example beyond Hamming we study the rank-metric case with a rank-metric
nearest-codeword optimisation problem. Here, the shells are fixed-rank matrices
in $\F_q^{m\times n}$, shell Dicke states are uniform superpositions over such
matrices, and the relevant Fourier eigenvalues are the $q$-Krawtchouk
polynomials. If the dual code uniquely decodes all rank errors up to the cutoff
$\ell$, the expectation and variance of a normalised proxy score reduce to a
truncated Jacobi matrix. Optimising the shell weights is therefore an eigenvalue
problem for this finite tridiagonal matrix. In the large-system regime,
optimised initial weights lead to an effective-rank proxy close to
$\min(m,n)-\ell$, and the optimised weights concentrate near the outer
rank-$\ell$ shell. For an efficient DQI protocol we require efficient unique
dual-decoding with a reversible implementation and efficient rank-shell Dicke
state preparation. We propose Gabidulin codes for the former and show the
latter.

The resulting expected effective-rank proxy can be converted into a tail bound
for the actual residual rank of a sampled solution, as shown in
Appendix~\ref{sec:expected-proxy-score-rank-tail-bound}. However, the same
appendix also shows that for Gabidulin primal codes, and more generally for
standard $\F_{q^m}$-linear singleton nearest-codeword instances in the dual
unique-decoding regime, this tail bound is separated from the OPT scale by a
covering-radius obstruction. Thus an OPT-relative guarantee appears to require
either stronger coherent decoding, nonstandard code families, or a different
objective ensemble.

This points naturally toward objective ensembles closer to max-LINSAT. The
present rank-metric objective has a singleton target, while max-LINSAT uses a
product of local target sets. The sum-rank metric provides a bridge: it reduces
to ordinary rank metric for one matrix block and to Hamming distance for scalar
blocks~\cite{gorla2023sumrankmetriccodes}. This suggests a ``max sum-rank''
problem in which each local constraint contributes a rank penalty rather than a
binary violation, with maximum sum-rank distance codes, such as linearised
Reed--Solomon codes, as possible structured code
families~\cite{martinezpenas2018lrs,campsmoreno2022optimal}. An interesting open
problem is whether such a formulation avoids the singleton-target
covering-radius obstruction and can support a DQI-compatible regime analogous to
OPI.

Other directions include a rank-metric analogue of Hamiltonian DQI and
extensions beyond translation schemes, such as Johnson and Grassmann geometries.
The broader question is which finite geometries provide efficient shell
preparation, efficient decoding primitives, and optimisation problems for which
the induced distance-dependent amplitude bias is algorithmically meaningful.

\subsection*{Acknowledgements}
We thank Konstantinos Meichanetzidis and Frederic Sauvage for enlightening
discussions and ongoing collaborations.

\bibliography{main}

\clearpage

\appendix

\section{\texorpdfstring{Rank-metric $q$-Krawtchouk polynomials}{Rank-metric q-Krawtchouk polynomials}}
\label{sec:rank-metric-q-krawtchouk-polynomials}

We provide in this section additional details on the radial Fourier eigenvalues
$\lambda_r(\cdot)$ in the rank-metric scheme. These are the $q$-Krawtchouk
polynomials. First, recall that Lemma~\ref{lem:rank-score-first-fourier-mode} gives the
first $q$-Krawtchouk polynomial as
\begin{equation}
    \lambda_1(t) = \frac{q^{m+n-t}-q^m-q^n+1}{q-1}, \quad t\in [0,D]
\end{equation}
Additionally, from the definition of $\lambda_r(\cdot)$ in
\eqref{eq:distance-matrix-character-sum} we have
\begin{equation}
    \lambda_r(0)=\abs{\Gamma_r}
\end{equation}

We now provide an explicit representation of the $q$-Krawtchouk polynomials
$\lambda_r(t)$ in terms of $q$-Pochhammer and $q$-binomial coefficients (see
Eqs.~\eqref{eq:finite-q-pochhammer-symbol} and
\eqref{eq:gaussian-binomial-coefficient} for definition).  Let $D=\min(m,n)$ and
$M=\max(m,n)$. For $0\le r,t\le D$, an explicit expression for the
$q$-Krawtchouk polynomials $\lambda_r(t)$ is given by
\begin{equation}\label{eq:q-krawtchouk-polynomial}
    \lambda_r(t)
    =
    \sum_{j=0}^{r}
    (-1)^{\,r-j}\;
    q^{\,jM+\binom{r-j}{2}}\;
    \binom{D-j}{\,D-r\,}_q\;
    \binom{D-t}{\,j\,}_q
\end{equation}
This representation is standard, see e.g., Ref.~\cite{delsarte1978bilinear}. For
concreteness we represent numerically the $q$-Krawtchouk polynomials in
Fig.~\ref{fig:q-krawtchouk-polynomial-plot} for $q=2$, $m=5$, $n=6$.

\begin{figure}[h!]
    \centering
    \includegraphics[width=0.8\linewidth]{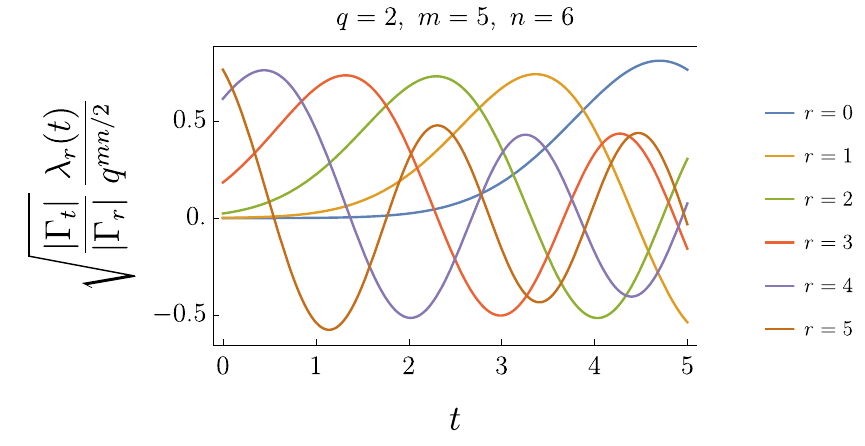}
    \caption{Plot of the rescaled $q$-Krawtchouk polynomial
        $\sqrt{\abs{\Gamma_t}/\abs{\Gamma_r}}\lambda_r(t) q^{-mn/2}$ as a
        function of $t\in [0,D]$ for all $r\in [0,D]$. The evaluation of the
        function has been extrapolated from the integer points to continuous
        points by using the connection between the $q$-Krawtchouk polynomial
        and the $q$-deformed hypergeometric function $\lambda_r(t) =
            (-1)^r\,q^{\binom{r}{2}} \binom{D}{r}_q\; {}_2\phi_1\!\left(
            \begin{matrix}
                q^{-r},\; q^{\,t-D} \\[2pt]
                q^{-D}
            \end{matrix}
            ;\, q,\; q^{\,M-t+1}
            \right)$ with $D=\min(m,n)$ and $M=\max(m,n)$ \cite{koekoekHypergeometricOrthogonalPolynomials2010}.}
    \label{fig:q-krawtchouk-polynomial-plot}
\end{figure}

Upon rescaling, the $q$-Krawtchouk polynomials are symmetric upon exchanging
their variable and index, i.e.,
\begin{equation}
    \label{eq:q-krawtchouk-spectral-symmetry}
    \abs{\Gamma_t}\lambda_r(t) = \abs{\Gamma_r}\lambda_t(r)
\end{equation}
This symmetry indicates that the problem is bi-spectral and that a
$q$-Krawtchouk transform is self-dual.

Moreover, the $q$-Krawtchouk polynomials
also enjoy a Christoffel--Darboux identity.

\begin{lemma}[Christoffel--Darboux identity]
    \label{lem:christoffel-darboux-identity}
    Fix a $P$-polynomial translation association scheme and let $\lambda_r(t)$
    denote the possibly complex radial Fourier eigenvalue of the distance
    operator $A_r$ on the frequency shell indexed by $t\in [0,D]$. Write
    $\lambda_r^*(u)\coloneqq\overline{\lambda_r(u)}$. Assume the three-term
    recurrence (in $r$) holds for every $t$ with real recurrence coefficients:
    \begin{equation}\label{eq:q-krawtchouk-three-term-recurrence}
        \lambda_1(t)\lambda_r(t)=b_{r-1}\lambda_{r-1}(t)+a_r\lambda_r(t)+c_{r+1}\lambda_{r+1}(t),
        \qquad r=0,\dots,D,
    \end{equation}
    with $b_{-1}=c_{D+1}=0$. We recall the orthogonality relation
    \eqref{eq:character-sum-orthogonality}
    \begin{equation}\label{eq:christoffel-darboux-orthogonality}
        \sum_{t=0}^D \abs{\Gamma_t} \lambda_i(t) \lambda_j^*(t) = \abs{\X} \abs{\Gamma_i}\delta_{ij}
    \end{equation}
    Then for $0 \leq L < D$ and all $t,u$ such that
    $\lambda_1(t)\neq \lambda_1^*(u)$,
    \begin{equation}\label{eq:christoffel-darboux-lambda}
        \sum_{r=0}^{L}\frac{\lambda_r(t)\lambda_r^*(u)}{\abs{\Gamma_r}}
        =
        \frac{c_{L+1}}{\abs{\Gamma_L}}\,
        \frac{\lambda_{L+1}(t)\lambda_L^*(u)-\lambda_L(t)\lambda_{L+1}^*(u)}
        {\lambda_1(t)-\lambda_1^*(u)}.
    \end{equation}
\end{lemma}

\begin{proof}
    Fix $t,u$ with $\lambda_1(t)\neq\lambda_1^*(u)$. The coefficients
    $a_r,b_r,c_r$ are real intersection numbers, so conjugating the recurrence
    \eqref{eq:q-krawtchouk-three-term-recurrence} at $u$ gives
    \begin{equation}
        \lambda_1^*(u)\lambda_r^*(u)
        =b_{r-1}\lambda_{r-1}^*(u)+a_r\lambda_r^*(u)+c_{r+1}\lambda_{r+1}^*(u).
    \end{equation}
    Write the recurrence at $t$, multiply it by $\lambda_r^*(u)$, multiply the
    conjugated recurrence at $u$ by $\lambda_r(t)$, and subtract:
    \begin{align*}
        (\lambda_1(t)-\lambda_1^*(u))\,\lambda_r(t)\lambda_r^*(u)
         & =
        \big(b_{r-1}\lambda_{r-1}(t)+a_r\lambda_r(t)+c_{r+1}\lambda_{r+1}(t)\big)\lambda_r^*(u)              \\
         & \quad-\lambda_r(t)\big(b_{r-1}\lambda_{r-1}^*(u)+a_r\lambda_r^*(u)+c_{r+1}\lambda_{r+1}^*(u)\big) \\
         & =
        b_{r-1}\big(\lambda_{r-1}(t)\lambda_r^*(u)-\lambda_r(t)\lambda_{r-1}^*(u)\big)
        +c_{r+1}\big(\lambda_{r+1}(t)\lambda_r^*(u)-\lambda_r(t)\lambda_{r+1}^*(u)\big),
    \end{align*}
    since the $a_r$ terms cancel. Divide by $\abs{\Gamma_r}$ and sum over
    $r=0,\dots,L$:
    \begin{equation}
        \label{eq:christoffel-darboux-telescoping-sum}
        \begin{split}
             & (\lambda_1(t)-\lambda_1^*(u))\sum_{r=0}^{L}\frac{\lambda_r(t)\lambda_r^*(u)}{\abs{\Gamma_r}} \\
             & =
            \sum_{r=0}^{L} \left[ \frac{b_{r-1}}{\abs{\Gamma_r}}\big(\lambda_{r-1}(t)\lambda_r^*(u)-\lambda_r(t)\lambda_{r-1}^*(u)\big)+\frac{c_{r+1}}{\abs{\Gamma_r}}\big(\lambda_{r+1}(t)\lambda_r^*(u)-\lambda_r(t)\lambda_{r+1}^*(u)\big) \right]
        \end{split}
    \end{equation}
    We can now use the ratio relation \eqref{eq:ratio-off-diagonal} in the form
    $b_{r-1}/\abs{\Gamma_r}=c_r/\abs{\Gamma_{r-1}}$. Thus the $b_{r-1}$ summand
    for index $r$ cancels the $c_r$ summand for index $r-1$. The lower boundary
    term vanishes because $b_{-1}=0$, so only the upper boundary term remains:
    \begin{equation}
        (\lambda_1(t)-\lambda_1^*(u))\sum_{r=0}^{L}\frac{\lambda_r(t)\lambda_r^*(u)}{\abs{\Gamma_r}}
        =
        \frac{c_{L+1}}{\abs{\Gamma_L}}\Big(\lambda_{L+1}(t)\lambda_L^*(u)-\lambda_L(t)\lambda_{L+1}^*(u)\Big).
    \end{equation}
\end{proof}

\begin{remark}
    In the rank-metric scheme the radial Fourier eigenvalues are real. Hence the
    complex conjugations in \eqref{eq:christoffel-darboux-lambda} disappear.
    Since $\lambda_1(t)=(q^{m+n-t}-q^m-q^n+1)/(q-1)$ is strictly decreasing in
    $t$, the condition $\lambda_1(t)\neq\lambda_1(u)$ is equivalent to $t\neq
        u$. The symmetry \eqref{eq:q-krawtchouk-spectral-symmetry} then translates
    the Christoffel--Darboux identity as
    \begin{equation}
        \label{eq:christoffel-darboux-lambda-symmetrised}
        \sum_{r=0}^{L}
        \abs{\Gamma_r}\,\lambda_t(r)\lambda_u(r)
        =
        c_{L+1}\abs{\Gamma_{L+1}}\,
        \frac{\lambda_t(L+1)\lambda_u(L)-\lambda_t(L)\lambda_u(L+1)}
        {\lambda_1(t)-\lambda_1(u)}.
    \end{equation}
\end{remark}
The Christoffel--Darboux identity
\eqref{eq:christoffel-darboux-lambda-symmetrised} can notably control the
overlap between states such as
\begin{equation}
    \ket{\tilde{\lambda}_r} \propto \sum_{x\in \X, \rank(x) \leq L}\lambda_r(\rank(x))\ket{x}
\end{equation}
which arise if the Fourier transform $\mathcal{F}$ is truncated to low-modes
leading to an approximate transform.\\

Finally, the $q$-Krawtchouk polynomials admit a finite summation identity
which gives an orthogonal-polynomial interpretation of the Fourier-paired
rank-noise profiles appearing in
Refs.~\cite{debris-alazardQuantumReductionFinding2024,blanvillainQuantumDecodingProblem2026}.

\begin{lemma}[Finite $q$-binomial transform identity for the $q$-Krawtchouk polynomials]
    \label{lem:zeta-q-krawtchouk-transform}
    Assume throughout the convention that $\binom{a}{b}_q=0$ whenever
    $b\notin\{0,\ldots,a\}$. For $0\le u\le D$, define
    \begin{equation}
        \zeta_u(r)
        \coloneqq
        \binom{D-r}{u-r}_q .
    \end{equation}
    One then has the following summation identity
    \begin{equation}\label{eq:zeta-q-krawtchouk-transform}
        \sum_{r=0}^{u}
        \zeta_u(r)\lambda_r(t)
        =
        q^{Mu}\binom{D-t}{u}_q
        =
        q^{Mu}\zeta_{D-u}(t).
    \end{equation}
\end{lemma}

\begin{proof}
    Substituting Eq.~\eqref{eq:q-krawtchouk-polynomial} and exchanging
    the $r$- and $j$-sums gives
    \begin{equation}
        \begin{aligned}
            \sum_{r=0}^{u}\zeta_u(r)\lambda_r(t)
             & =
            \sum_{j=0}^{u}
            q^{jM}\binom{D-t}{j}_q
            \sum_{r=j}^{u}
            (-1)^{r-j}q^{\binom{r-j}{2}}
            \binom{D-r}{u-r}_q
            \binom{D-j}{D-r}_q .
        \end{aligned}
    \end{equation}
    Using the Gaussian-binomial product identity
    \begin{equation}\label{eq:gaussian-binomial-product-identity}
        \binom{D-r}{u-r}_q
        \binom{D-j}{D-r}_q
        =
        \binom{D-j}{u-j}_q
        \binom{u-j}{r-j}_q ,
    \end{equation}
    and setting $h=r-j$, one obtains
    \begin{equation}
        \begin{aligned}
            \sum_{r=0}^{u}\zeta_u(r)\lambda_r(t)
             & =
            \sum_{j=0}^{u}
            q^{jM}\binom{D-t}{j}_q
            \binom{D-j}{u-j}_q
            \sum_{h=0}^{u-j}
            (-1)^h q^{\binom{h}{2}}
            \binom{u-j}{h}_q .
        \end{aligned}
    \end{equation}
    The remaining finite sum follows from the finite $q$-binomial theorem,
    namely
    \begin{equation}\label{eq:finite-q-binomial-theorem}
        (a;q)_L
        \coloneqq
        \prod_{\ell=0}^{L-1}(1-aq^\ell)
        =
        \sum_{h=0}^{L}
        (-1)^h q^{\binom{h}{2}}
        \binom{L}{h}_q a^h .
    \end{equation}
    Evaluating Eq.~\eqref{eq:finite-q-binomial-theorem} at $a=1$ gives
    \begin{equation}
        \sum_{h=0}^{u-j}
        (-1)^h q^{\binom{h}{2}}
        \binom{u-j}{h}_q
            =
            (1;q)_{u-j}
        =
        \delta_{j,u}.
    \end{equation}
    Consequently, only the term $j=u$ survives, and
    \begin{equation}
        \sum_{r=0}^{u}\zeta_u(r)\lambda_r(t)
        =
        q^{Mu}\binom{D-t}{u}_q .
    \end{equation}
    Finally, using Gaussian-binomial symmetry,
    \begin{equation}
        \binom{D-t}{u}_q
        =
        \binom{D-t}{D-u-t}_q
        =
        \zeta_{D-u}(t),
    \end{equation}
    proves Eq.~\eqref{eq:zeta-q-krawtchouk-transform}. This is a finite $q$-binomial transform identity; equivalently, its
    proof is a degenerate terminating $q$-Chu--Vandermonde evaluation.
\end{proof}

The identity in Eq.~\eqref{eq:zeta-q-krawtchouk-transform} gives an
orthogonal-polynomial formulation of the Fourier-paired rank-noise
profiles introduced in
Ref.~\cite{debris-alazardQuantumReductionFinding2024} and further
analysed in the quantum-decoding framework of
Ref.~\cite{blanvillainQuantumDecodingProblem2026}. It therefore
identifies a natural input family for soft-decoding extensions of
rank-metric DQI. We detail in the next corollary which quantum identities follow the previous lemma.

\begin{corollary}[Fourier-paired $\zeta$-states and explicit DQI amplitudes]
    \label{cor:fourier-paired-zeta-states}
    For $0\le u\le D$, define the normalised radial state
    \begin{equation}
        \ket{\Pi_u}
        \coloneqq
        \frac{1}{\sqrt{\mathcal Z_u}}
        \sum_{r=0}^{D}
        \sqrt{\abs{\Gamma_r}}\,
        \zeta_u(r)\ket{R_r},
        \qquad
        \mathcal Z_u
        \coloneqq
        \sum_{r=0}^{D}
        \abs{\Gamma_r}\zeta_u(r)^2.
    \end{equation}
    Then, using the conventions of \eqref{eq:fourier-transform-radial-state},
    \begin{equation}\label{eq:fourier-pi-complement}
        \mathcal F_{\X}\ket{\Pi_u}
        =
        \ket{\Pi_{D-u}},
    \end{equation}
    and
    \begin{equation}\label{eq:zeta-normalisation-duality}
        \mathcal Z_{D-u}
        =
        q^{M(D-2u)}\mathcal Z_u.
    \end{equation}
    Moreover, the normalisation admits the closed form
    \begin{equation}\label{eq:zeta-normalisation-hypergeometric-form}
        \mathcal Z_u
        =
        q^{Mu}
        \binom{D}{u}_q\,
        {}_2\phi_1
        \left(
        \begin{matrix}
            q^{-u},\,q^{u-D} \\
            q
        \end{matrix};
        q,\,
        q^{D-M+1}
        \right).
    \end{equation}
    Finally, if one sets
    \begin{equation}\label{eq:zeta-dqi-shell-weights}
        \ell=u,
        \qquad
        w_k^{(u)}
        \coloneqq
        \sqrt{\frac{\abs{\Gamma_k}}{\mathcal Z_u}}\,
        \zeta_u(k),
        \qquad
        0\le k\le u,
    \end{equation}
    in the DQI protocol at Step 2 (i.e., $\ket{\psi_2}=\ket{\Pi_u}$) then the explicit final state \eqref{eq:dqi-final-radial-state} of the protocol would read
    \begin{equation}\label{eq:zeta-dqi-output-state}
        \ket{\psi_6^{(u)}}
        =
        \frac{q^{Mu}}{\sqrt{\abs{\Omega}\mathcal Z_u}}
        \sum_{x\in\Omega}
        \zeta_{D-u}\!\left(\rank(\Delta_y(x))\right)\ket{x}.
    \end{equation}
    assuming decoding at Step 4 is possible.
\end{corollary}

\begin{proof}
    Since $\zeta_u(r)=0$ for $r>u$, the definition of $\mathcal Z_u$
    and the rank-shell cardinality formula give
    \begin{equation}\label{eq:zeta-normalisation-shell-expression}
        \begin{aligned}
            \mathcal Z_u
             & =
            \sum_{r=0}^{u}
            \abs{\Gamma_r}
            \binom{D-r}{u-r}_q^2
            \\
             & =
            \sum_{r=0}^{u}
            (-1)^r q^{\binom{r}{2}}
            (q;q)_r
            \binom{M}{r}_q
            \binom{D}{r}_q
            \binom{D-r}{u-r}_q^2 .
        \end{aligned}
    \end{equation}
    Using the Gaussian-binomial product identity
    \begin{equation}
        \binom{D}{r}_q
        \binom{D-r}{u-r}_q
        =
        \binom{D}{u}_q
        \binom{u}{r}_q ,
    \end{equation}
    twice, Eq.~\eqref{eq:zeta-normalisation-shell-expression} becomes
    \begin{equation}\label{eq:zeta-normalisation-shell-expression-reduced}
        \mathcal Z_u
        =
        \binom{D}{u}_q^2
        \sum_{r=0}^{u}
        (-1)^r q^{\binom{r}{2}}
        (q;q)_r
        \binom{M}{r}_q
        \frac{\binom{u}{r}_q^2}{\binom{D}{r}_q}.
    \end{equation}
    Now use
    \begin{equation}\label{eq:q-binomial-pochhammer-representation}
        \binom{a}{r}_q
        =
        (-1)^r
        q^{ar-\binom{r}{2}}
        \frac{(q^{-a};q)_r}{(q;q)_r}.
    \end{equation}
    Substituting Eq.~\eqref{eq:q-binomial-pochhammer-representation}
    into Eq.~\eqref{eq:zeta-normalisation-shell-expression-reduced} yields
    \begin{equation}\label{eq:zeta-normalisation-3phi1-form}
        \mathcal Z_u
        =
        \binom{D}{u}_q^2\,
        {}_3\phi_1
        \left(
        \begin{matrix}
            q^{-M},\,q^{-u},\,q^{-u} \\
            q^{-D}
        \end{matrix};
        q,\,
        q^{M+2u-D}
        \right).
    \end{equation}
    Here we use the standard convention
    \begin{equation}
        {}_3\phi_1
        \left(
        \begin{matrix}
            a_1,\,a_2,\,a_3 \\
            b_1
        \end{matrix};
        q,z
        \right)
        \coloneqq
        \sum_{r=0}^{\infty}
        \frac{
            (a_1;q)_r(a_2;q)_r(a_3;q)_r
        }{
            (b_1;q)_r(q;q)_r
        }
        (-1)^r q^{-\binom{r}{2}}z^r .
    \end{equation}
    The series in Eq.~\eqref{eq:zeta-normalisation-3phi1-form} terminates at $r=u$.
    Applying the terminating basic-hypergeometric Jackson’s transformation \cite{DLMF17}
    \begin{equation}\label{eq:terminating-3phi1-2phi1-transformation}
        \begin{aligned}
             & \binom{D}{u}_q\,
                              {}_3\phi_1
            \left(
            \begin{matrix}
                q^{-M},\,q^{-u},\,q^{-u} \\
                q^{-D}
            \end{matrix};
            q,\,
            q^{M+2u-D}
            \right)
            =
            q^{Mu}\,
                   {}_2\phi_1
            \left(
            \begin{matrix}
                q^{-u},\,q^{u-D} \\
                q
            \end{matrix};
            q,\,
            q^{D-M+1}
            \right)
        \end{aligned}
    \end{equation}
    gives Eq.~\eqref{eq:zeta-normalisation-hypergeometric-form}.

    Replacing $u$ by $D-u$ in
    Eq.~\eqref{eq:zeta-normalisation-hypergeometric-form} gives
    \begin{equation}
        \begin{aligned}
            \mathcal Z_{D-u}
             & =
            q^{M(D-u)}
            \binom{D}{D-u}_q\,
                             {}_2\phi_1
            \left(
            \begin{matrix}
                q^{u-D},\,q^{-u} \\
                q
            \end{matrix};
            q,\,
            q^{D-M+1}
            \right)
            \\
             & =
            q^{M(D-2u)}\mathcal Z_u,
        \end{aligned}
    \end{equation}
    since $\binom{D}{D-u}_q=\binom{D}{u}_q$ and
    ${}_2\phi_1$ is symmetric in its numerator parameters.

    Using Eq.~\eqref{eq:fourier-transform-radial-state},
    $\abs{\X}=q^{MD}$, and
    Eq.~\eqref{eq:zeta-q-krawtchouk-transform}, one obtains
    \begin{equation}
        \begin{aligned}
            \mathcal F_{\X}\ket{\Pi_u}
             & =
            \frac{1}{\sqrt{\abs{\X}\mathcal Z_u}}
            \sum_{j=0}^{D}
            \sqrt{\abs{\Gamma_j}}
            \left(
            \sum_{r=0}^{u}
            \zeta_u(r)\lambda_r(j)
            \right)
            \ket{R_j}
            \\
             & =
            q^{M(u-D/2)}
            \sqrt{\frac{\mathcal Z_{D-u}}{\mathcal Z_u}}\,
            \ket{\Pi_{D-u}}
            \\
             & =
            \ket{\Pi_{D-u}}.
        \end{aligned}
    \end{equation}
    Finally, substituting Eq.~\eqref{eq:zeta-dqi-shell-weights} into
    Eq.~\eqref{eq:dqi-final-radial-state} yields
    \begin{equation}
        \begin{aligned}
            \ket{\psi_6^{(u)}}
             & =
            \frac{1}{\sqrt{\abs{\Omega}\mathcal Z_u}}
            \sum_{x\in\Omega}
            \left(
            \sum_{k=0}^{u}
            \zeta_u(k)
            \lambda_k\!\left(\rank(\Delta_y(x))\right)
            \right)
            \ket{x}
            \\
             & =
            \frac{q^{Mu}}{\sqrt{\abs{\Omega}\mathcal Z_u}}
            \sum_{x\in\Omega}
            \zeta_{D-u}\!\left(\rank(\Delta_y(x))\right)\ket{x},
        \end{aligned}
    \end{equation}
    where we used Eq.~\eqref{eq:zeta-q-krawtchouk-transform} to go from the first to the second line.
\end{proof}

\section{MacWilliams transform for rank-weight distributions}
\label{sec:macwilliams-transform-rank-weight-distributions}
For any $x\in\X$ define the rank weight distribution around $x$ as
\begin{equation}\label{eq:rank-weight-distribution}
    W_r(x)\coloneqq\abs{\Set{y\in \Code \given \rank(y-x)=r}},\qquad r \in [0,D]
\end{equation}

We then have the generalised MacWilliams identity.
\begin{lemma}[Generalised MacWilliams identity in rank
        metric]
    \label{lem:rank-metric-macwilliams-identity}
    For every $x\in\X$ and every $r\in[0,D]$,
    \begin{equation}\label{eq:rank-metric-macwilliams-identity}
        \sum_{\ell=0}^{D} W_\ell(x)\,\lambda_r(\ell)
        \;=\;
        \abs{\Code}\sum_{e\in\Gamma_r}
        \adchar_x(-e)
        \,\indicator_{\Code^\perp}(e) = \abs{\Code}\sum_{e\in\Gamma_r \cap \Code^\perp}
        \adchar_x(-e) .
    \end{equation}
\end{lemma}
\begin{proof}
    By definition of $W_\ell(x)$ in Eq.~\eqref{eq:rank-weight-distribution} and
    $\lambda_r(\ell)$ in Eq.~\eqref{eq:distance-matrix-character-sum},

    \begin{equation}\label{eq:macwilliams-proof-expansion}
        \sum_{\ell=0}^{D} W_\ell(x)\,\lambda_r(\ell)
        =\sum_{y\in \Code} \lambda_r(\rank(y-x))
        =\sum_{y\in \Code}\ \sum_{e \in \Gamma_r }
        \adchar_{y-x}(e)
        =\sum_{e\in \Gamma_r}
        \adchar_x(-e)
        \sum_{y\in \Code}
        \adchar_y(e),
    \end{equation}
    where we swap sums and factor the phase for the last equality. The inner sum
    is a character sum over the additive group $\Code$:
    \begin{equation}\label{eq:macwilliams-inner-character-sum}
        \sum_{y\in \Code}
        \adchar_y(e)
        =
        \begin{cases}
            \abs{\Code}, & e\in \Code^\perp,    \\
            0,           & e\notin \Code^\perp,
        \end{cases}
    \end{equation}
    by orthogonality of characters. Inserting
    Eq.~\eqref{eq:macwilliams-inner-character-sum} into
    Eq.~\eqref{eq:macwilliams-proof-expansion} gives the last equation in
    Eq.~\eqref{eq:rank-metric-macwilliams-identity}.
\end{proof}

\section{From expected proxy score to a rank tail bound}
\label{sec:expected-proxy-score-rank-tail-bound}

The main text optimises the expectation of the normalised proxy score
$O_q(\Delta_y(x))$, or equivalently of $q^{-\rank(\Delta_y(x))}$. Since this
proxy score is monotone in the residual rank, it is natural to ask whether the
proxy-score bias gives a statement about actual sampled ranks. For the optimised
rank-metric DQI weights, the variance formula in
Lemma~\ref{lem:rank-metric-proxy-score-observable-moments} gives a simple second-moment conversion.
It yields a tail bound for the rank of a sample from the final DQI state. The
bound is not, by itself, an approximation guarantee relative to
$\mathrm{OPT}(y)$; such a guarantee would require additional information about
the rank distribution in the coset $\Code-y$ or about a structured instance
family.

\begin{proposition}[Proxy-to-tail conversion for optimised rank weights]
    \label{prop:proxy-to-tail-conversion}
    Let $D=\min(m,n)$, let $0\le \ell<D$, and assume the dual-distance condition
    $2\ell+2<d^\perp$. Let $w^{(\ell)}$ be a normalised eigenvector associated
    with the largest eigenvalue of $\tilde A^{(\ell)}$. Let $R$ be the residual
    rank of the candidate obtained by measuring the final DQI state prepared
    with weights $w^{(\ell)}$; equivalently, if $\tilde{x}\in\Omega$ is this
    random candidate, then
    \begin{equation}
        R\coloneqq \rank(\Delta_y(\tilde{x})).
    \end{equation}
    Define
    \begin{equation}
        R_\ell^*
        \coloneqq
        R_\ell(w^{(\ell)})
        =
        -\log_q(1-\mu_{\min}^{(\ell)}).
    \end{equation}
    Then $R_\ell^*\le D-\ell$ and, for every $a>0$,
    \begin{equation}
        \Pr\!\left[
            R\le R_\ell^*+a
            \right]
        \ge
        \frac{(1-q^{-a})^2}{1+q}.
    \end{equation}
\end{proposition}

In particular,
\begin{equation}
    \Pr\!\left[
        R\le D-\ell+a
        \right]
    \ge
    \frac{(1-q^{-a})^2}{1+q}.
\end{equation}

\begin{proof}
    Set $Y\coloneqq q^{-R}$. By Lemma~\ref{lem:rank-score-first-fourier-mode} and
    Eq.~\eqref{eq:dqi-observable-expectation-ladder}, the expectation of $Y$ in
    the final DQI state is
    \begin{equation}
        \mathbb E[Y]
        =
        1-\frac{q-1}{q^{m+n}}\,
        \|B^{(\ell)}w^{(\ell)}\|^2
        =
        1-\mu_{\min}^{(\ell)}
        =
        q^{-R_\ell^*}.
    \end{equation}
    Lemma~\ref{lem:truncated-rank-jacobi-spectral-bounds} gives
    $\mu_{\min}^{(\ell)}\le 1-q^{\ell-D}$, and hence
    $R_\ell^*\le D-\ell$.

    We next bound the relative variance of $Y$. Since $Y$ and $O_q$ differ by an
    affine rescaling, and since
    \begin{equation}
        \Var_{z\sim\mathrm{Unif}}\!\left[q^{-\rank(z)}\right]
        =
        (q-1)q^{-(m+n)}(1-q^{-m})(1-q^{-n}),
    \end{equation}
    while
    \begin{equation}
        \abs{\Gamma_1}
        =
        \frac{(q^m-1)(q^n-1)}{q-1},
    \end{equation}
    Lemma~\ref{lem:rank-metric-proxy-score-observable-moments} implies
    \begin{equation}
        \Var(Y)
        =
        (q-1)^2q^{-2(m+n)}
        \abs*{w_\ell^{(\ell)}}^2 b_\ell c_{\ell+1}.
    \end{equation}
    Here the term $w^\dagger(\tilde A^{(\ell)})^2w-(w^\dagger\tilde
        A^{(\ell)}w)^2$ vanishes because $w^{(\ell)}$ is an eigenvector of $\tilde
        A^{(\ell)}$. Using the rank-metric intersection numbers
    \begin{equation}
        b_\ell
        =
        \frac{(q^m-q^\ell)(q^n-q^\ell)}{q-1},
        \qquad
        c_{\ell+1}
        =
        \frac{(q^{\ell+1}-1)q^\ell}{q-1},
    \end{equation}
    we obtain the bound
    \begin{equation}
        b_\ell c_{\ell+1}
        \le
        \frac{q^{m+n+2\ell+1}}{(q-1)^2}.
    \end{equation}
    Since $\abs*{w_\ell^{(\ell)}}\le1$ and $\mathbb E[Y]^2=q^{-2R_\ell^*}\ge
        q^{-2(D-\ell)}$, it follows that
    \begin{equation}
        \frac{\Var(Y)}{\mathbb E[Y]^2}
        \le
        q^{1+2D-m-n}
        =
        q^{1-\abs{m-n}}
        \le q.
    \end{equation}
    Thus $\mathbb E[Y^2]\le(1+q)\mathbb E[Y]^2$.

    Paley--Zygmund's inequality applied to the nonnegative random variable $Y$
    gives, for every $a>0$,
    \begin{equation}
        \Pr\!\left[
            Y\ge q^{-a}\mathbb E[Y]
            \right]
        \ge
        \frac{(1-q^{-a})^2\mathbb E[Y]^2}{\mathbb E[Y^2]}
        \ge
        \frac{(1-q^{-a})^2}{1+q}.
    \end{equation}
    Since $\mathbb E[Y]=q^{-R_\ell^*}$, the event
    $Y\ge q^{-a}\mathbb E[Y]$ is exactly $R\le R_\ell^*+a$. The weaker bound
    with $D-\ell+a$ follows from $R_\ell^*\le D-\ell$.
\end{proof}

\begin{remark}
    Proposition~\ref{prop:proxy-to-tail-conversion} converts the optimised
    expected proxy score into an explicit tail statement about the actual
    residual rank of a sampled candidate. For fixed $q$, this gives constant
    success probability. Repetition and classical evaluation of
    $\rank(\Delta_y(x))$ can amplify this success probability and keep the best
    sample. However, the threshold $D-\ell+a$ is an absolute rank threshold
    determined by the shell cutoff and the ambient dimensions. It does not imply
    that the sampled rank is close to $\mathrm{OPT}(y)$ on every instance. An
    OPT-relative approximation guarantee would require further assumptions or a
    sharper analysis of the rank distribution in the relevant coset.
\end{remark}

\begin{proposition}[Gabidulin obstruction to additive OPT approximation]
    \label{prop:gabidulin-additive-opt-approximation-obstruction}
    Let $\Code\le\F_q^{m\times n}$ be an $[n,k]$ Gabidulin code over $\F_{q^m}$,
    with $n\le m$ and $k<n$. Run rank-metric DQI with $\Code$ as the primal
    code, optimised shell weights $w^{(\ell)}$, and cutoff $0\le\ell<n$. Assume
    the expectation condition $2\ell+1<d^\perp$, where $d^\perp=k+1$ is the rank
    distance of $\Code^\perp$. Let $\tilde{x}\in\Omega$ denote the measured
    classical output of the final DQI state and set
    \begin{equation}
        R\coloneqq \rank(\Delta_y(\tilde{x})).
    \end{equation}
    Then there is a constant $C_q=2(1+q^{-1/2})$, such that for every $A\ge0$,
    \begin{equation}
        \Pr\!\left[
            R\le \mathrm{OPT}(y)+A
            \right]
        \le
        C_q\,q^{A+\ell-k}.
    \end{equation}
    In particular, in the unique-decoding regime $\ell<k/2$, this probability is
    at most $C_q q^{A-k/2}$ up to an inessential constant factor. Hence, for any
    sequence with $k\to\infty$ and fixed additive error $A$, the current
    Gabidulin-based DQI construction cannot sample a solution with rank at most
    $\mathrm{OPT}(y)+A$ with constant probability.
\end{proposition}

\begin{proof}
    We use two facts about Gabidulin codes. First, the dual of an $[n,k]$
    Gabidulin code is an $[n,n-k]$ Gabidulin code, so $d^\perp=k+1$. Second, the
    rank-metric covering radius of an $[n,k]$ Gabidulin code is
    $n-k$~\cite{byrne2017covering}. Therefore, for every target $y$,
    \begin{equation}
        \mathrm{OPT}(y)
        =
        \min_{c\in\Code}\rank(c-y)
        \le n-k .
    \end{equation}

    Let $Y\coloneqq q^{-R}$. As in the proof of
    Proposition~\ref{prop:proxy-to-tail-conversion}, the optimised weights give
    \begin{equation}
        \mathbb E[Y]=1-\mu_{\min}^{(\ell)}.
    \end{equation}
    We now use Lemma~\ref{lem:truncated-rank-jacobi-spectral-bounds} to upper
    bound this expectation. Since $n\le m$, set $s=q^{\ell-n}$. The lower bound
    on $\mu_{\min}^{(\ell)}$ reads
    \begin{equation}
        \mu_{\min}^{(\ell)}
        \ge
        \left(
        \sqrt{(1-q^{\ell-m})(1-s)}
        -
        \sqrt{(q^\ell-1)q^{\ell-1-m-n}}
        \right)^2 .
    \end{equation}
    The first square root is at least $1-s$, while the second is at most
    $q^{-1/2}s$. Hence
    \begin{equation}
        \mu_{\min}^{(\ell)}
        \ge
        (1-(1+q^{-1/2})s)^2,
    \end{equation}
    and therefore
    \begin{equation}
        \mathbb E[Y]
        =
        1-\mu_{\min}^{(\ell)}
        \le
        2(1+q^{-1/2})q^{\ell-n}
        =
        C_q q^{\ell-n}.
    \end{equation}

    If $R\le\mathrm{OPT}(y)+A$, then $R\le n-k+A$, and hence $Y\ge
        q^{-(n-k+A)}$. Markov's inequality applied to the nonnegative random
    variable $Y$ gives
    \begin{equation}
        \Pr\!\left[
            R\le\mathrm{OPT}(y)+A
            \right]
        \le
        q^{n-k+A}\mathbb E[Y]
        \le
        C_q q^{A+\ell-k}.
    \end{equation}
    The final assertion follows from the unique-decoding constraint $\ell<k/2$.
\end{proof}

\begin{remark}
    Proposition~\ref{prop:gabidulin-additive-opt-approximation-obstruction} explains why the
    proxy-to-tail conversion above should not be interpreted as an approximation
    theorem for Gabidulin nearest-codeword instances. Gabidulin codes are very
    good covering codes in the rank metric: every target is within rank $n-k$ of
    the code. In contrast, the DQI cutoff is limited by unique decoding of the
    dual Gabidulin code to roughly $k/2$ shells, and the optimised
    effective-rank proxy is centred around rank $n-\ell$, which is at least
    $n-k/2$ in this regime. Thus the algorithmic obstruction is not the
    proxy-to-tail conversion, but the mismatch between the primal covering
    radius and the dual unique-decoding radius. Closing the remaining
    approximation problem would require a different structured ensemble whose
    typical optimum lies near $n-\ell$, or a stronger coherent decoding
    primitive that permits substantially larger $\ell$.
\end{remark}

The same obstruction is not merely an artefact of the Gabidulin construction.
For an $\F_{q^m}$-linear code of length $n$ and dimension $k$, the Singleton
bound applied to the dual gives $d^\perp\le k+1$. Thus unique decoding of the
dual up to rank $\ell$ forces roughly $k\ge 2\ell$. On the other hand, any
systematic $\F_{q^m}$-linear code has covering radius at most $n-k$, since one
can match an arbitrary received word on the $k$ information positions. Hence the
singleton nearest-codeword optimum is bounded by approximately $n-2\ell$ in the
regime where the present DQI protocol is valid, while the optimised DQI tail
bound is centred near $n-\ell$. This suggests that an OPT-relative theorem for
the singleton rank-metric objective would require changing one of the structural
ingredients: a coherent decoding primitive beyond the unique-decoding radius, a
nonstandard code family with unusually large covering radius and efficiently
decodable dual, or a different optimisation problem, such as a product-target or
sum-rank variant, whose optimum scale is not fixed by the covering radius of a
single code.

\section{Deferred proofs}
\label{sec:deferred-proofs}

\subsection{Proof of Lemma~\ref{lem:recursion-coefficient-identities} (Recursion coefficient identities)}
\label{subsec:proof-recursion-coefficient-identities}

\begin{proof}
    Fix $i\in[0,D]$ and $x\in\Gamma_i$. In a $P$-polynomial translation scheme,
    the $A_1$-neighbours of $x$ lie only in the shells $\Gamma_{i-1}$,
    $\Gamma_i$, and $\Gamma_{i+1}$, with out-of-range shells omitted. By the
    definitions of the intersection numbers, the numbers of such neighbours are
    respectively $c_i$, $a_i$, and $b_i$, using $c_0=b_D=0$ at the endpoints.
    These three classes partition the $A_1$-neighbourhood of $x$, whose size is
    $\abs{\Gamma_1}$, and therefore
    \begin{equation}
        a_i+b_i+c_i=\abs{\Gamma_1}.
    \end{equation}

    Now fix $0\le i<D$ and consider the set of edges between consecutive shells
    \begin{equation}
        E_{i,i+1}\coloneqq
        \Set*{(x,z)\in\Gamma_i\times\Gamma_{i+1}\given (A_1)_{xz}=1}.
    \end{equation}
    We count $\abs{E_{i,i+1}}$ in two ways. First, fix $x \in \Gamma_i$. By
    definition, $b_i$ counts its $A_1$-neighbours in $\Gamma_{i+1}$. Therefore
    \begin{equation}
        \abs{E_{i,i+1}} = b_i\abs{\Gamma_i}.
    \end{equation}
    Next, fix $z \in \Gamma_{i+1}$. For each $z$ there are exactly $c_{i+1}$
    $A_1$-neighbours in $\Gamma_i$. Therefore
    \begin{equation}
        \abs{E_{i,i+1}} = c_{i+1}\abs{\Gamma_{i+1}}.
    \end{equation}
    Equating the two counts gives
    \begin{equation}
        b_i\abs{\Gamma_i} = c_{i+1}\abs{\Gamma_{i+1}}
    \end{equation}
    Since $c_{i+1}>0$ in the $P$-polynomial ordering, division by
    $c_{i+1}\abs{\Gamma_i}$ gives
    \begin{equation}
        \frac{\abs{\Gamma_{i+1}}}{\abs{\Gamma_i}}=\frac{b_i}{c_{i+1}}.
    \end{equation}
\end{proof}

\subsection{Proof of Lemma~\ref{lem:distance-matrix-fourier-eigenvectors} (Eigenvectors and eigenvalues of distance matrices)}
\label{subsec:proof-distance-matrix-fourier-eigenvectors}

\begin{proof}
    We prove that the Fourier transform diagonalises the distance matrices of
    translation schemes. First we use the definition of the Fourier transform,
    Eq.~\eqref{eq:fourier-basis-state}, the action of distance matrices,
    Eq.~\eqref{eq:action-distance-matrix}, and translation invariance to obtain
    \begin{equation}
        A_i \mathcal{F}\ket{x}
        =
        \frac{1}{\sqrt{\abs*{\X}}}\sum_{y\in \X}
        \adchar_x(y)
        \sum_{e\in\Gamma_i}\ket{y-e}
        =
        \frac{1}{\sqrt{\abs*{\X}}}\sum_{e\in\Gamma_i}\sum_{y^\prime\in\X}
        \adchar_x(y^\prime+e)
        \ket{y^\prime}.
    \end{equation}
    Pull out $\adchar_x(e)$ and rearrange:
    \begin{equation}
        A_i\mathcal{F}\ket{x}
        =
        \left(\sum_{e\in\Gamma_i}
        \adchar_x(e)
        \right)\cdot
        \frac{1}{\sqrt{\abs*{\X}}}\sum_{y^\prime\in\X}
        \adchar_x(y^\prime)
        \ket{y^\prime}
        =
        \lambda_i(x)\mathcal{F}\ket{x}.
    \end{equation}
\end{proof}

\subsection{Proof of Lemma~\ref{lem:radial-fourier-eigenvalue-conditions} (Sufficient conditions for radial Fourier eigenvalues)}
\label{subsec:proof-radial-fourier-eigenvalue-conditions}

\begin{proof}
    Let $x\in \X$ and $g\in G_0$. By the definition of the character sums,
    \begin{equation}
        \lambda_i(g(x)) = \sum_{e\in \Gamma_i} \adchar_{g(x)}(e) = \sum_{e\in \Gamma_i} \adchar_x(g^\top(e)),
    \end{equation}
    where $g^\top$ is the adjoint of $g$ defined by $\bilinear{g(x)}{y} =
        \bilinear{x}{g^\top(y)}$ with respect to the bilinear form used to define
    the characters. By assumption, $g^\top$ preserves each shell, and hence
    permutes $\Gamma_i$. Substituting $e'=g^\top(e)$, the above sum is
    \begin{equation}
        \sum_{e\in \Gamma_i} \adchar_x(g^\top(e)) = \sum_{e'\in \Gamma_i} \adchar_x(e') = \lambda_i(x)
    \end{equation}
    so $\lambda_i(g(x)) = \lambda_i(x)$. Hence, the function $\lambda_i$ is
    constant on the orbits of $G_0$.

    By the transitivity assumption, for any two points $x,y$ in the same shell,
    there exists an $h\in G_0$ such that $h(x)=y$. Hence,
    \begin{equation}
        \lambda_i(y) = \lambda_i(h(x)) = \lambda_i(x).
    \end{equation}
    Thus $\lambda_i$ is constant on each shell, so the Fourier eigenvalues are
    radial.
\end{proof}

\subsection{Proof of Lemma~\ref{lem:shell-orthogonality-dual-distance} (Shell orthogonality below the dual code distance)}
\label{subsec:proof-shell-orthogonality-dual-distance}

\begin{proof}
    We prove the stated orthogonality relation by expanding the character sums
    $\lambda_i$. Since $\overline{\chi(a)}=\chi(a)^*=\chi(-a)$, for every
    $c,e,e^\prime$ the product of characters in
    Eq.~\eqref{eq:code-averaged-character-sum} combines as follows:
    \begin{equation}
        \begin{split}
            S_{i,j}(y)
             & =
            \sum_{c\in \Code}\sum_{e\in \Gamma_i}\sum_{e^\prime\in \Gamma_j}
            \adchar_{c-y}(e)\adchar^*_{c-y}(e^\prime) \\
             & =
            \sum_{e\in \Gamma_i}\sum_{e^\prime\in \Gamma_j}
            \adchar_{-y}(e-e^\prime)
            \sum_{c\in \Code}
            \adchar_c(e-e^\prime).
        \end{split}
    \end{equation}

    \smallskip \noindent Using the subgroup average over $\Code$
    \begin{equation}
        \sum_{c\in \Code}
        \adchar_c(e-e^\prime)
        =
        \begin{cases}
            \abs{\Code}, & e-e^\prime\in \Code^\perp,    \\
            0,           & e-e^\prime\notin \Code^\perp.
        \end{cases}
    \end{equation}
    we obtain
    \begin{equation}
        S_{i,j}(y)
        =
        \abs{\Code}
        \sum_{e\in \Gamma_i}\sum_{e^\prime\in \Gamma_j}
        \indicator_{\Code^\perp}(e-e^\prime)\,
        \adchar_{-y}(e-e^\prime).
    \end{equation}

    Assume $i+j<d^\perp$. By the triangle inequality and translation invariance
    of $d$, we have
    \begin{equation}
        d(e-e^\prime, 0)
        \le d(e, 0) + d(-e^\prime, 0)
        = d(e, 0)+d(e^\prime, 0)
        = i+j
        < d^\perp.
    \end{equation}
    Hence, if $e-e^\prime\in\Code^\perp$, then $e-e^\prime$ is a dual codeword
    of distance strictly smaller than $d^\perp$. By definition of $d^\perp$,
    this forces $e-e^\prime=0$, and therefore $e=e^\prime$. In particular,
    $e\in\Gamma_i$ and $e^\prime\in\Gamma_j$ then imply $i=j$. Therefore
    \begin{equation}
        S_{i,j}(y)
        =
        \abs{\Code}
        \sum_{e\in \Gamma_i}\sum_{e^\prime\in \Gamma_j}
        \delta_{e, e^\prime}
        =
        \abs{\Code}\,\abs{\Gamma_i}\,\delta_{i,j}.
    \end{equation}
\end{proof}

\subsection{Deferred rank-metric code facts}
\label{subsec:deferred-rank-metric-code-facts}

We recall the projection argument behind the rank-metric Singleton bound used in
Sec.~\ref{subsec:rank-metric-codes}. Let $\Code\le\F_q^{m\times n}$ have minimum
rank distance $d_\mathrm{rk}$, and first assume $n\le m$. Project onto the first
$n-d_\mathrm{rk}+1$ columns:
\begin{equation}
    \pi\colon \F_q^{m\times n}\to\F_q^{m\times(n-d_\mathrm{rk}+1)}.
\end{equation}
If $c\ne c^\prime\in\Code$ and $\pi(c)=\pi(c^\prime)$, then $c-c^\prime$ is
supported on the final $d_\mathrm{rk}-1$ columns, so
\begin{equation}
    \rank(c-c^\prime)\le d_\mathrm{rk}-1,
\end{equation}
contradicting the definition of $d_\mathrm{rk}$. Thus $\pi|_{\Code}$ is
injective and
\begin{equation}
    \dim_{\F_q}\Code\le m(n-d_\mathrm{rk}+1).
\end{equation}
The case $m<n$ is symmetric, giving
\begin{equation}
    \dim_{\F_q}\Code\le \max(m,n)(D-d_\mathrm{rk}+1).
\end{equation}

We also recall why an $[n,k]$ Gabidulin code has distance $n-k+1$. Identify
$\F_q^{m\times n}$ with $(\F_{q^m})^n$, so matrix rank is the $\F_q$-dimension
of the span of the coordinates in $\F_{q^m}$, and let
$U=\operatorname{span}_{\F_q}\Set{\alpha_1,\ldots,\alpha_n}$, where the
evaluation points are $\F_q$-linearly independent. A nonzero linearised
polynomial
\begin{equation}
    f(z)=\sum_{i=0}^r a_i z^{q^i}, \qquad a_r\ne 0,
\end{equation}
has an $\F_q$-linear kernel of dimension at most $r$. Hence, if $f=f_{\mathbf
            a}$ has $q$-degree $<k$, then
\begin{equation}
    \dim_{\F_q}\ker(f)\le k-1.
\end{equation}
The rank of the Gabidulin codeword is
\begin{equation}
    \rank(f(\alpha_1),\ldots,f(\alpha_n))
    = \dim_{\F_q}f(U)
    = n-\dim_{\F_q}(\ker(f)\cap U)
    \ge n-k+1.
\end{equation}
Equality is attained by taking the $q$-annihilator polynomial of
$\operatorname{span}_{\F_q}\Set{\alpha_1,\ldots,\alpha_{k-1}}$ when $k>1$, and
by taking $f(z)=a_0z$ when $k=1$. Therefore $d_{\mathrm{rk}}=n-k+1$. Since
$\dim_{\F_q}\Code=mk$, the Singleton bound is met with equality:
\begin{equation}
    m(n-d_{\mathrm{rk}}+1)=m(n-(n-k+1)+1)=mk.
\end{equation}
The remaining standard facts used in Lemma~\ref{lem:gabidulin-code-properties},
that duals of Gabidulin codes are Gabidulin and that unique decoding is
polynomial time up to $\floor{(d-1)/2}$ rank errors, are due to
Gabidulin~\cite{gabidulin1985theory}.

\subsection{Proof of Lemma~\ref{lem:rank-metric-shell-sizes-intersection-numbers} (Rank-metric shell sizes and intersection numbers)}
\label{subsec:proof-rank-metric-shell-sizes-intersection-numbers}
\begin{proof}[Proof of part~\ref{item:rank-metric-p-polynomial}]
    We will show that $p_{1i}^k=0$ unless $k \in \Set{i-1,i,i+1}$. Take
    $x\in\Gamma_i$ and a rank-one matrix $y \in \Gamma_1$. Then
    \begin{equation}
        \rank(x+y) \leq \rank(x) + \rank(y) = \rank(x) +1
    \end{equation}
    and
    \begin{equation}
        \rank(x) = \rank((x+y)-y) \leq \rank(x+y) +1.
    \end{equation}
    Hence
    \begin{equation}
        \rank(x)-1 \leq \rank(x+y)  \leq \rank(x) +1
    \end{equation}
    and multiplication by $A_1$ sends the $i$th distance matrix only to the
    neighbouring shells $i-1,i,i+1$. Since the rank-metric scheme is transitive
    on each shell, the corresponding counts depend only on $i$. Therefore
    $A_1A_i$ is a linear combination of $A_{i-1}$, $A_i$, and $A_{i+1}$, so the
    scheme is $P$-polynomial.
\end{proof}

\begin{proof}[Proof of part~\ref{item:rank-metric-shell-size}]
    First, choose the column space. A rank-$i$ matrix $x \in \X$ has an
    $i$-dimensional column space $\mathrm{col}(x)\le \F_q^m$. The number of
    choices is:
    \begin{equation}
        \binom{m}{i}_q.
    \end{equation}
    Similarly, the row space of $x$ is an $i$-dimensional subspace
    $\mathrm{row}(x)\le \mathbb{F}_q^n$. The number of such choices is:
    \begin{equation}
        \binom{n}{i}_q.
    \end{equation}

    Finally, for a fixed pair $(u,v)$ of $i$-dimensional subspaces with $u \le
        \F_q^m$ and $v \le \F_q^n$, we need to count the matrices $x$ with
    $\mathrm{col}(x)=u$ and $\mathrm{row}(x)=v$. A matrix of this form is
    equivalent to an invertible linear map $ v \to u $ since restricting the
    domain and codomain to these $i$-spaces gives an $i\times i$ matrix of full
    rank. The number of such maps is
    \begin{equation}
        \abs{\mathrm{GL}_i(\F_q)} =\prod_{t=0}^{i-1} (q^i-q^t)
        .
    \end{equation}
    Finally, the number of rank-$i$ matrices in $\F_q^{m \times n}$ is
    \begin{align*}
        \abs{\Gamma_i} & =\binom{m}{i}_q\binom{n}{i}_q\prod_{t=0}^{i-1}{(q^i-q^t)}                                                                                               \\
                       & = \left(\prod_{t=0}^{i-1}\frac{q^m-q^t}{q^i-q^t} \right) \left(\prod_{t=0}^{i-1}\frac{q^n-q^t}{q^i-q^t}\right)\left(\prod_{t=0}^{i-1} (q^i-q^t) \right) \\
                       & = \frac{
                               q^{2\sum_{t=0}^{i-1}t}
                               \prod_{t=0}^{i-1}(q^{m-t}-1)(q^{n-t}-1)}{q^{\sum_{t=0}^{i-1}t}
                               \prod_{t=0}^{i-1}(q^{i-t}-1)}                                                                                     \\
                       & = q^{i(i-1)/2}
        \frac{
            \prod_{t=0}^{i-1}(q^{m-t}-1)(q^{n-t}-1)
        }{
            \prod_{t=0}^{i-1}(q^{i-t}-1)
        }                                                                                                                               \\
                       & = q^{\frac{i(i-1)}{2}}
        \left(
        \prod_{t=0}^{i-1}\frac{q^{m-t}-1}
                         {q^{i-t}-1}
        \right)
        \left(
        \prod_{t=0}^{i-1}\frac{q^{n-t}-1}
                         {q^{i-t}-1}
        \right)
        \prod_{t=0}^{i-1}(q^{i-t}-1)                                                                                                                                             \\
                       & = q^{\frac{i(i-1)}{2}}
        \binom{m}{i}_q
        \binom{n}{i}_q
        \prod_{s=1}^{i}(q^s-1)                                                                                                                                                   \\
                       & = (-1)^i\binom{m}{i}_q
        \binom{n}{i}_q
        (q;q)_i
        q^{\frac{i(i-1)}{2}}.
    \end{align*}
\end{proof}

\begin{proof}[Proof of part~\ref{item:rank-metric-intersection-parameters}]
    For any $z \in \Gamma_i$, there exists a basis in which
    \begin{equation}
        g z h=\left(\begin{array}{cc}
            I_i & 0 \\
            0   & 0
        \end{array}\right) \eqqcolon z_0 \quad \text{for some } g \in \mathrm{GL}_m(\F_q), h \in \mathrm{GL}_n(\F_q).
    \end{equation}
    Notice that for every $e \in \Gamma_1$,
    \begin{equation}
        \rank(z+e) = \rank(g(z+e)h) = \rank(gzh + geh).
    \end{equation}
    Moreover, since the map $e \mapsto geh$ gives a bijection on $\Gamma_1$, we
    have
    \begin{equation}
        p^i_{1k}
        = \abs*{\Set{e \in \Gamma_1 \given \rank(z+e) = k}}
        = \abs*{\Set{e \in \Gamma_1 \given \rank(z_0+e) = k}} .
    \end{equation}

    Consequently, the intersection counts depend only on $i$ and we can assume
    \begin{equation}
        z = z_0 =\left(\begin{array}{cc}
            I_i & 0 \\
            0   & 0
        \end{array}\right).
    \end{equation}
    Every rank-one matrix can be written as
    \begin{equation}
        e=u v^{\top} \quad \text{with } u \in \F_q^m \setminus\Set{0},  v \in \F_q^n \setminus\Set{0} .
    \end{equation}

    Two pairs represent the same matrix precisely when $ (u, v) =\left(\alpha u,
        \alpha^{-1} v\right) $ for some $ \alpha \in \mathbb{F}_q ^\times$. This
    means we only need to count pairs, and divide by $q-1$ to compensate for
    overcounting. Now write
    \begin{equation}
        u=
        \begin{pmatrix}
            u_1 \\
            u_2
        \end{pmatrix},
        \qquad
        v=
        \begin{pmatrix}
            v_1 \\
            v_2
        \end{pmatrix},
    \end{equation}
    with
    \begin{equation}
        u_1 \in \F_q^{i}, \qquad
        u_2 \in \F_q^{m-i}, \qquad
        v_1 \in \F_q^{i}, \qquad
        v_2 \in \F_q^{n-i}.
    \end{equation}

    \begin{lemma}
        \label{lem:rank-increase-by-one-condition}
        Given a rank-one matrix $e = uv^\top \in \Gamma_1$ then $\rank(z + e) =
            i+1 $ if and only if $u_2 \neq 0 $ and $v_2  \neq 0$.
    \end{lemma}
    \begin{proof}
        If $u_2 = 0$, then $u$ belongs to the column space of $z_0$. Since the columns of
        $e$ are $(uv^\top)_j = v_j u$, each column is a scaling of $u$ and therefore
        every column of $e$ is also in the column space of $z_0$. Consequently,
        $\mathrm{col}(z_0 + e) \subseteq \mathrm{col}(z_0)$ and therefore
        \begin{equation}
            \rank(z_0 +e) \leq i.
        \end{equation}
        Similarly, if $v_2=0$ the row space of $z_0+e$ is contained in that of
        $z_0$.

        Conversely, suppose $u_2,v_2 \neq 0 $. Then choose $a,b$ such that
        $(u_2)_a \neq 0 $ and $(v_2)_b \neq 0$. The $(i+1) \times (i + 1)$ minor
        of $z_0 + e$ formed from the first $i$ rows and columns together with
        row $i+a$ and column $i+b$ has determinant $(u_2)_a(v_2)_b \neq 0$.
        Therefore $\rank(z_0 +e) \geq i+1$. Since the rank change is bounded by
        $\rank(e) =1$, we have
        \begin{equation}
            \rank(z_0 + e) = i+1
        \end{equation}
        in this case.
    \end{proof}

    It remains to count the number of rank-one matrices $e$ satisfying the
    conditions of Lemma \ref{lem:rank-increase-by-one-condition}. First, note
    that $\abs*{\Set{u_2 \in \F_q^{m-i} \given u_2 \neq 0}}= q^{m-i} -1$.
    Similarly, $\abs*{\Set{v_2 \in \F_q^{n-i} \given v_2 \neq 0}}= q^{n-i} -1$.
    After accounting for the free choices for $u_1 \in \F_q^i$ and $v_1 \in
        \F_q^i$, the number of distinct matrices $e \in \Gamma_1$ such that
    $\rank(z_0 + e) = i+1$ is
    \begin{equation}
        \label{eq:rank-b-coefficient-proof}
        b_i \coloneqq \frac{(q^m-q^i)(q^n-q^i)}{q-1}.
    \end{equation}

    Now, the rank can only decrease if $u_2 = v_2 = 0$, in which case
    \begin{equation}
        z_0 + e =
        \left(\begin{array}{cc}
            I_i + u_1v_1^\top & 0 \\
            0                 & 0
        \end{array}\right).
    \end{equation}
    For vectors $ u_1,v_1$ we have $\mathrm{det}(I_i + u_1v_1^\top) =  1 +
        v_1^\top u_1$, and so the matrix is only singular when $v_1^\top u_1 = -1$.
    In this case the rank must be $i-1$, and so
    \begin{equation}
        \rank(I_i+ u_1v_1^\top) = i-1.
    \end{equation}
    Again, we count such matrices. There are $q^i-1$ choices for $u_1 \neq 0$.
    Upon fixing $u_1$, we have a single linear constraint $v_1^\top u_1=-1$
    leaving $q^{i-1}$ choices for $v_1$. Therefore the number of distinct
    matrices is
    \begin{equation}
        \label{eq:rank-c-coefficient-proof}
        c_i \coloneqq \frac{(q^i-1)q^{i-1}}{q-1}.
    \end{equation}
    Finally, we can now solve for the remaining recursion coefficient. Recall
    from Lemma~\ref{lem:rank-metric-shell-sizes-intersection-numbers}, the total number of rank-one
    matrices is
    \begin{equation}
        \abs{\Gamma_1} = \frac{(q^m-1)(q^n-1)}{q-1}.
    \end{equation}
    Now
    \begin{equation}
        \label{eq:rank-a-coefficient-proof}
        a_i = \abs{\Gamma_1} - b_i -c_i = \frac{(q^m-1)(q^n-1)}{q-1}- b_i -c_i.
    \end{equation}
\end{proof}

\subsection{Proof of Lemma~\ref{lem:rank-shell-dicke-preparation} (Rank-shell
    Dicke state preparation)}
\label{subsec:proof-rank-shell-dicke-preparation}

The construction uses the standard parametrisation of a rank-$r$ matrix
by its column space, its row space, and an invertible coordinate matrix.
For $0\le r\le s$, we denote by
\begin{equation}
    \mathrm{Gr}(r,s)
    \coloneqq
    \left\{
    U\le \F_q^s:
    \dim(U)=r
    \right\}
\end{equation}
the finite Grassmannian of $r$-dimensional subspaces of $\F_q^s$.
Its cardinality is the Gaussian binomial coefficient
\begin{equation}
    \abs{\mathrm{Gr}(r,s)}
    =
    \binom{s}{r}_q.
\end{equation}
For each $U\in\mathrm{Gr}(r,s)$, let $R_U\in\F_q^{r\times s}$ be
the unique reduced-row-echelon matrix whose rows span $U$, and set
$B_U\coloneqq R_U^\top\in\F_q^{s\times r}$. Thus the columns of
$B_U$ form a canonical ordered basis of $U$. We write
\begin{equation}
    \ket{\mathrm{Gr}(r,s)}
    \coloneqq
    \frac{1}{\sqrt{\binom{s}{r}_q}}
    \sum_{U\in\mathrm{Gr}(r,s)}
    \ket{B_U}
\end{equation}
for the corresponding uniform superposition over canonical basis matrices.
Similarly, let
\begin{equation}
    \ket{\mathrm{GL}_r}
    \coloneqq
    \frac{1}{\sqrt{\abs{\mathrm{GL}_r(\F_q)}}}
    \sum_{M\in\mathrm{GL}_r(\F_q)} \ket{M},
\end{equation}
with the convention that $\mathrm{GL}_0(\F_q)$ is the singleton containing the
empty matrix.

Now let $e\in\F_q^{m\times n}$ have rank $r$. Its column space and
row space determine two points
\begin{equation}
    U=\operatorname{Col}(e)\in\mathrm{Gr}(r,m),
    \qquad
    V=\operatorname{Row}(e)\in\mathrm{Gr}(r,n).
\end{equation}
We use the same convention for the row space $V$: if $R_V$ is the
reduced-row-echelon matrix whose rows span $V$, then $B_V=R_V^\top\in
    \F_q^{n\times r}$, so $B_V^\top$ is the canonical row-basis matrix for
$V$. Since $B_U$ and $B_V$ have full column rank, there exists a unique
matrix $M\in\mathrm{GL}_r(\F_q)$ such that
\begin{equation}\label{eq:rank-r-grassmannian-factorisation}
    e=B_U M B_V^\top .
\end{equation}
Conversely, every triple
\begin{equation}
    (U,V,M)
    \in
    \mathrm{Gr}(r,m)\times
    \mathrm{Gr}(r,n)\times
    \mathrm{GL}_r(\F_q)
\end{equation}
defines through Eq.~\eqref{eq:rank-r-grassmannian-factorisation} a
matrix of rank exactly $r$, with column space $U$ and row space $V$.
Hence
\begin{equation}\label{eq:rank-r-parametrisation-bijection}
    \mathrm{Gr}(r,m)\times
    \mathrm{Gr}(r,n)\times
    \mathrm{GL}_r(\F_q)
    \longrightarrow
    \Gamma_r,
    \qquad
    (U,V,M)\longmapsto B_U M B_V^\top,
\end{equation}
is a bijection. Consequently,
\begin{equation}
    \abs{\Gamma_r}
    =
    \binom{m}{r}_q\binom{n}{r}_q\abs{\mathrm{GL}_r(\F_q)}.
\end{equation}
The strategy for preparing $\ket{R_r}$ is therefore to prepare
$\ket{\mathrm{Gr}(r,m)}\ket{\mathrm{Gr}(r,n)}\ket{\mathrm{GL}_r}$, compute
$e=B_U M B_V^\top$, and reversibly erase the parametrising registers.

In the following sublemmas, a circuit is called \emph{clean} if all auxiliary
registers are returned to their initial $\ket{0}$ state.
Figure~\ref{fig:rank-shell-preparation-circuit} shows the three reversible
compute-uncompute patterns used in the construction.

Zhandry's quantum lightning construction gives an important precedent for using
low-rank matrix superpositions in quantum algorithms
\cite{zhandryQuantumLightningNever2021}. That work prepares constrained
superpositions over upper-triangular representations of low-rank symmetric
matrices. Our
construction prepares uniform superpositions over all rectangular matrices of a
fixed rank.

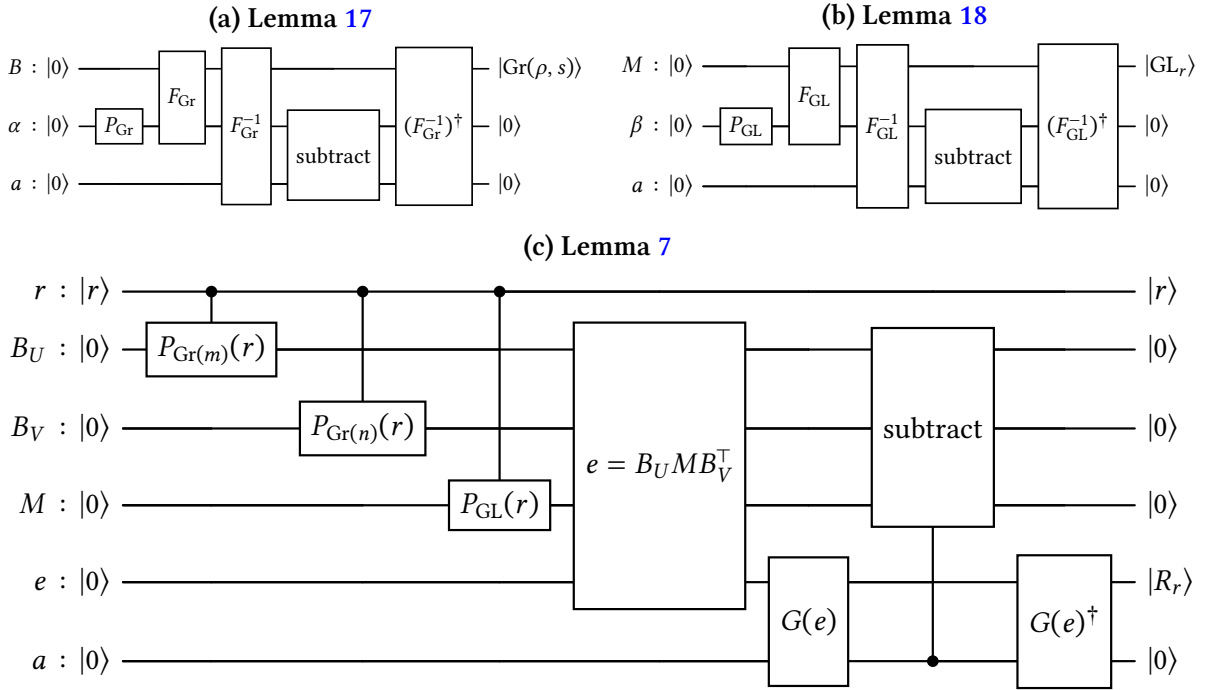
\begin{figure}[H]
    \centering
    \begin{minipage}{0.49\linewidth}
        \centering
        \textbf{(a) Lemma~\ref{lem:clean-grassmannian-preparation}}
        \vspace{0.3ex}

        \resizebox{\linewidth}{!}{%
            \begin{quantikz}[row sep=0.32cm,column sep=0.28cm]
                \lstick{$B:\ket{0}$}
                & \qw
                & \gate[2]{F_{\mathrm{Gr}}}
                & \gate[3]{F_{\mathrm{Gr}}^{-1}}
                & \qw
                & \gate[3]{(F_{\mathrm{Gr}}^{-1})^\dagger}
                & \rstick{$\ket{\mathrm{Gr}(\rho,s)}$} \\
                \lstick{$\alpha:\ket{0}$}
                & \gate{P_{\mathrm{Gr}}}
                & \qw
                & \qw
                & \gate[2]{\mathrm{subtract}}
                & \qw
                & \rstick{$\ket{0}$} \\
                \lstick{$a:\ket{0}$}
                & \qw
                & \qw
                & \qw
                & \qw
                & \qw
                & \rstick{$\ket{0}$}
            \end{quantikz}}
    \end{minipage}
    \hfill
    \begin{minipage}{0.49\linewidth}
        \centering
        \textbf{(b) Lemma~\ref{lem:clean-general-linear-group-preparation}}
        \vspace{0.3ex}

        \resizebox{\linewidth}{!}{%
            \begin{quantikz}[row sep=0.32cm,column sep=0.28cm]
                \lstick{$M:\ket{0}$}
                & \qw
                & \gate[2]{F_{\mathrm{GL}}}
                & \gate[3]{F_{\mathrm{GL}}^{-1}}
                & \qw
                & \gate[3]{(F_{\mathrm{GL}}^{-1})^\dagger}
                & \rstick{$\ket{\mathrm{GL}_r}$} \\
                \lstick{$\beta:\ket{0}$}
                & \gate{P_{\mathrm{GL}}}
                & \qw
                & \qw
                & \gate[2]{\mathrm{subtract}}
                & \qw
                & \rstick{$\ket{0}$} \\
                \lstick{$a:\ket{0}$}
                & \qw
                & \qw
                & \qw
                & \qw
                & \qw
                & \rstick{$\ket{0}$}
            \end{quantikz}}
    \end{minipage}

    \vspace{1.2ex}
    \textbf{(c) Lemma~\ref{lem:rank-shell-dicke-preparation}}
    \vspace{0.3ex}

    \resizebox{\linewidth}{!}{%
        \begin{quantikz}[row sep=0.32cm,column sep=0.30cm]
            \lstick{$r:\ket{r}$}
            & \ctrl{1}
            & \ctrl{2}
            & \ctrl{3}
            & \qw
            & \qw
            & \qw
            & \qw
            & \rstick{$\ket{r}$} \\
            \lstick{$B_U:\ket{0}$}
            & \gate{P_{\mathrm{Gr}(m)}(r)}
            & \qw
            & \qw
            & \gate[4]{e=B_UMB_V^\top}
            & \qw
            & \gate[3]{\mathrm{subtract}}
            & \qw
            & \rstick{$\ket{0}$} \\
            \lstick{$B_V:\ket{0}$}
            & \qw
            & \gate{P_{\mathrm{Gr}(n)}(r)}
            & \qw
            & \qw
            & \qw
            & \qw
            & \qw
            & \rstick{$\ket{0}$} \\
            \lstick{$M:\ket{0}$}
            & \qw
            & \qw
            & \gate{P_{\mathrm{GL}}(r)}
            & \qw
            & \qw
            & \qw
            & \qw
            & \rstick{$\ket{0}$} \\
            \lstick{$e:\ket{0}$}
            & \qw
            & \qw
            & \qw
            & \qw
            & \gate[2]{G(e)}
            & \qw
            & \gate[2]{G(e)^\dagger}
            & \rstick{$\ket{R_r}$} \\
            \lstick{$a:\ket{0}$}
            & \qw
            & \qw
            & \qw
            & \qw
            & \qw
            & \ctrl{-4}
            & \qw
            & \rstick{$\ket{0}$}
        \end{quantikz}}
    \caption{Schematic state-preparation circuits used in
        Lemma~\ref{lem:rank-shell-dicke-preparation}.
        In panels (a) and (b), $P_{\mathrm{Gr}}$ and $P_{\mathrm{GL}}$ prepare
        uniform superpositions over enumeration addresses. The maps
        $F_{\mathrm{Gr}}$ and $F_{\mathrm{GL}}$ compute the canonical object
        from the address, while $F_{\mathrm{Gr}}^{-1}$ and
        $F_{\mathrm{GL}}^{-1}$ recompute the unique address from the output so
        it can be subtracted and uncomputed. In panel (c), the black controls
        are shown only on the state-preparation blocks whose target distribution
        is indexed by the input rank; the rank register is never a target. The
        multiplication, subtraction, and $G(e)$ blocks act on fixed padded
        registers. Here $G(e)$ denotes the reversible canonicalisation that
        recovers $(B_U,B_V,M)$ from $e$ into $a$, and the vertical source line
        into the subtract block shows that this recovered triple is subtracted
        from the work registers. All auxiliary registers return
        to $\ket{0}$.}
    \label{fig:rank-shell-preparation-circuit}
\end{figure}

\begin{lemma}[Clean Grassmannian state preparation]
    \label{lem:clean-grassmannian-preparation}
    Let $0\le \rho\le s$. With the canonical basis convention above, in the
    model of reversible $\F_q$-arithmetic, controlled single-register rotations,
    and clean uniform $\F_q$-register preparations, there is a circuit
    \begin{equation}
        \ket{0}\ket{0}
        \longmapsto
        \ket{\mathrm{Gr}(\rho,s)}\ket{0}
    \end{equation}
    using $\mathscr{O}(s\rho^2)$ field operations and $\mathscr{O}(s\rho)$
    state-preparation primitives.
\end{lemma}

\begin{proof}
    We use the standard recursive decomposition of the finite Grassmannian. The
    reduced-row-echelon representation and the associated ranking and unranking
    viewpoint are standard in enumerative coding for Grassmannian
    spaces~\cite{silberstein11:grassmannian-enumerative}. Write
    $\F_q^s=\F_q^{s-1}\oplus \F_q e_s$ and let $\pi\colon\F_q^s\to\F_q^{s-1}$ be
    the projection. A subspace $U\in\mathrm{Gr}(\rho,s)$ is of exactly one of
    the following two types. If $e_s\in U$, then
    \begin{equation}
        U=W\oplus \F_q e_s,
        \qquad
        W\in\mathrm{Gr}(\rho-1,s-1).
    \end{equation}
    If $e_s\notin U$, then $\pi|_U$ is injective, its image is some
    $W\in\mathrm{Gr}(\rho,s-1)$, and
    \begin{equation}
        U=\Set{(w,\varphi(w))\given w\in W}
    \end{equation}
    for a unique linear functional $\varphi\in W^*$. Once the canonical basis
    $B_W$ is fixed, $\varphi$ is specified by its values on that basis, giving
    $q^\rho$ choices. Hence we recover the $q$-deformed Pascal identity
    \begin{equation}
        \label{eq:grassmannian-recursive-count}
        \binom{s}{\rho}_q
        =
        \binom{s-1}{\rho-1}_q
        +
        q^\rho\binom{s-1}{\rho}_q .
    \end{equation}

    Define a set of enumeration addresses recursively from this decomposition.
    Such an address records the branch choices and the graph-branch field
    values. The branch $e_s\in U$ stores one address for
    $\mathrm{Gr}(\rho-1,s-1)$, while the graph branch stores one address for
    $\mathrm{Gr}(\rho,s-1)$ and a vector in $\F_q^\rho$ giving the values of
    $\varphi$. The base cases $\rho=0$ and $\rho=s$ are singletons. The
    decomposition above gives a bijection between enumeration addresses and
    $\mathrm{Gr}(\rho,s)$.

    A controlled rotation with squared amplitudes proportional to the two
    summands in \eqref{eq:grassmannian-recursive-count}, followed by the
    corresponding recursive preparation, therefore prepares the uniform
    superposition over addresses. In the graph branch, the values of $\varphi$
    are prepared as a uniform state over $\F_q^\rho$. This is the usual
    prefix-counting state-preparation method, applied to the recursive count
    tree~\cite{grover02:state-preparation}.

    From an address, compute the corresponding canonical basis matrix $B_U$
    into an output register. Conversely, from $B_U$ one recovers the unique
    address by testing whether $e_s\in U$ and then applying the same
    decomposition recursively. Thus the address is a
    reversible function of $B_U$. Compute this inverse address into an auxiliary
    register, subtract it from the original address registers, and reverse the
    auxiliary computation. This is Bennett's compute-copy-uncompute pattern for
    reversible simulation~\cite{bennett89:reversible-computation}, and all work
    registers return to $\ket{0}$.

    Each recursive level only updates an echelon basis of width at most $\rho$.
    Over all $s$ levels this costs $\mathscr{O}(s\rho^2)$ field operations and
    $\mathscr{O}(s\rho)$ state-preparation primitives.
\end{proof}

\begin{lemma}[Clean general-linear-group state preparation]
    \label{lem:clean-general-linear-group-preparation}
    For every $r\ge0$, there is a clean circuit
    \begin{equation}
        \ket{0}\ket{0}
        \longmapsto
        \ket{\mathrm{GL}_r}\ket{0}
    \end{equation}
    using $\mathscr{O}(r^3)$ field operations and
    $\mathscr{O}(r^2)$ state-preparation primitives.
\end{lemma}

\begin{proof}
    Let
    \begin{equation}
        \mathcal U_r
        \coloneqq
        \prod_{j=1}^r
        \left(\F_q^{j-1}\times
        \left(\F_q^{r-j+1}\setminus\Set{0}\right)\right).
    \end{equation}
    Then
    \begin{equation}
        \abs{\mathcal U_r}
        =
        \prod_{j=1}^r q^{j-1}(q^{r-j+1}-1)
        =
        \prod_{j=0}^{r-1}(q^r-q^j)
        =
        \abs{\mathrm{GL}_r(\F_q)}.
    \end{equation}
    We define a bijection $F\colon\mathcal U_r\to\mathrm{GL}_r(\F_q)$. Starting
    with $T_0=I_r$, process $u_j\in \F_q^r\setminus\Span(e_1,\ldots,e_{j-1})$
    sequentially. Set $v_j=T_{j-1}^{-1}u_j$ and make $v_j$ the $j$th column of
    $F(u_1,\ldots,u_r)$. Then let $E_j(u_j)$ be the canonical product of
    elementary row operations that preserves $\Span(e_1,\ldots,e_{j-1})$ and
    sends $u_j$ to a vector in $\Span(e_1,\ldots,e_j)$ with nonzero $e_j$
    component. Set $T_j=E_j(u_j)T_{j-1}$.

    The resulting columns are linearly independent. Conversely, given an
    invertible matrix with columns $v_1,\ldots,v_r$, recover the address by
    starting from $T_0=I_r$, setting $u_j=T_{j-1}v_j$, and applying the same
    update rule for $T_j$. Since $v_j$ is not in the span of the previous
    columns, $u_j$ lies in $\F_q^r\setminus\Span(e_1,\ldots,e_{j-1})$. Hence $F$
    is a bijection.

    Prepare a uniform superposition over $\mathcal U_r$. The free prefix
    coordinates are uniform $\F_q$ registers. The nonzero suffix vector in
    dimension $d$ is prepared by branching on whether the first coordinate is
    nonzero, with branch counts $(q-1)q^{d-1}$ and $q^{d-1}-1$, and recursing in
    the zero branch. This uses $\mathscr{O}(r^2)$ state-preparation primitives
    over all columns.

    Compute $F(u_1,\ldots,u_r)$ into an output matrix register. Since $F$ and
    its inverse are both computable in $\mathscr{O}(r^3)$ field operations, the
    address registers can be erased by computing the inverse address from the
    output matrix, subtracting it from the original address, and uncomputing the
    auxiliary inverse computation. This again uses Bennett's reversible
    compute-copy-uncompute pattern. The output is the desired clean uniform
    state over $\mathrm{GL}_r(\F_q)$.
\end{proof}

\begin{proof}[Proof of Lemma~\ref{lem:rank-shell-dicke-preparation}]
    We first describe an exact circuit in the model where reversible
    $\F_q$-arithmetic, clean uniform $\F_q$-register preparations, and
    controlled rotations with prescribed algebraic amplitudes are elementary
    operations. The last sentence of the lemma follows by standard finite-field
    arithmetic and rotation synthesis.

    We now prepare the three registers
    \begin{equation}
        \ket{\mathrm{Gr}(r,m)}\,
        \ket{\mathrm{Gr}(r,n)}\,
        \ket{\mathrm{GL}_r},
    \end{equation}
    using Lemmas~\ref{lem:clean-grassmannian-preparation} and
    \ref{lem:clean-general-linear-group-preparation}. These preparations are clean, so all their
    internal work registers have returned to $\ket{0}$. We then compute an
    output matrix register
    \begin{equation}
        e=B_U M B_V^\top .
    \end{equation}
    The map
    \begin{equation}
        (U,V,M)\longmapsto B_U M B_V^\top
    \end{equation}
    is a bijection from
    $\mathrm{Gr}(r,m)\times\mathrm{Gr}(r,n)\times\mathrm{GL}_r(\F_q)$ onto
    $\Gamma_r$. Indeed, the column space and row space of $e$ are respectively
    $U$ and $V$, and after the canonical bases $B_U,B_V$ are fixed, the
    coordinate matrix $M$ is uniquely determined. Conversely, every rank-$r$
    matrix has a column space $U$, a row space $V$, and a unique invertible
    coordinate matrix in these two bases. Hence
    \begin{equation}
        \abs{\Gamma_r}
        =
        \binom{m}{r}_q\binom{n}{r}_q\abs{\mathrm{GL}_r(\F_q)}
    \end{equation}
    and, after computing $e$, the joint state is a uniform superposition over
    triples mapping bijectively to $\Gamma_r$.

    It remains to erase the work registers. From $e$ we compute, reversibly, the
    canonical bases of its column and row spaces by Gaussian elimination. Let
    these be $B_U$ and $B_V$. The same elimination also gives left inverses
    $L_U\in\F_q^{r\times m}$ and $L_V\in\F_q^{r\times n}$ satisfying $L_U
        B_U=I_r$ and $L_V B_V=I_r$. Then
    \begin{equation}
        M=L_U e L_V^\top .
    \end{equation}
    Thus the triple $(B_U,B_V,M)$ is a reversible function of $e$. Compute this
    triple into an auxiliary register, subtract it from the existing work
    registers to zero them, and reverse the auxiliary computation. The output
    register is then left in the state
    \begin{equation}
        \frac{1}{\sqrt{\abs{\Gamma_r}}}\sum_{e\in\Gamma_r}\ket{e}
        =
        \ket{R_r}.
    \end{equation}

    The multiplication $B_U M B_V^\top$ and the inverse computation from $e$
    require $\mathscr{O}(mnr)$ field operations by Gaussian elimination stopped
    after $r$ pivots and by rectangular matrix multiplication with an
    $r$-dimensional inner index. The auxiliary Grassmannian and $\mathrm{GL}_r$
    preparations cost $\mathscr{O}((m+n)r^2+r^3)$ field operations and
    $\mathscr{O}((m+n)r+r^2)$ state-preparation primitives. Since $r\le
        D=\min(m,n)$, these terms are bounded by $\mathscr{O}(mnr)$. For the
    controlled map, pad all work registers to the size needed for $r=\ell$. The
    rank-indexed preparation subroutines are implemented coherently by guarding
    their recursive branches with reversible comparisons against the input value
    of $r$. Each such comparison computes a predicate of $r$ into a temporary
    flag, uses that flag to control the relevant branch, and then uncomputes the
    flag. Thus the rank register is used only as a control, and the controlled
    unitary has the block-diagonal form $\sum_{r=0}^{\ell}\ket{r}\bra{r}\otimes
        U_r$ on the rank register and padded work space. In particular, $\sum_r
        c_r\ket{r}\ket{0}$ is mapped to $\sum_r c_r\ket{r}\ket{R_r}$ with no
    disturbance of the rank register. The multiplication, canonicalisation,
    subtraction, and uncomputation can be run as fixed padded reversible
    arithmetic. Taking the worst case over $0\le r\le\ell$ gives the controlled
    complexity $\mathscr{O}(mn\ell)$.

    Finally, there are $\mathscr{O}(mn\ell)$ state-preparation primitives.
    Synthesising each rotation to accuracy $\mathscr{O}(\varepsilon/(mn\ell))$
    gives total error at most $\varepsilon$ by the standard hybrid bound.
    Standard reversible finite-field arithmetic over $\F_q$ has $\polylog(q)$
    elementary-gate overhead in the usual binary
    encodings~\cite{beauregard03:galois-arithmetic,
        amento13:binary-field-inversion}. This yields
    $\tilde{\mathscr{O}}(mn\ell\,\polylog(q/\varepsilon))$ elementary gates.
\end{proof}

\subsection{Proof of Lemma~\ref{lem:rank-score-first-fourier-mode}
    (Rank score as a first Fourier mode)}
\label{subsec:proof-rank-score-first-fourier-mode}

\begin{proof}
    Let $z\in\F_q^{m\times n}$, $a\in\F_q^m$, and $b\in\F_q^n$, and consider the
    scalar bilinear form $a^\top z\,b\in\F_q$. Define the function
    \begin{equation}\label{eq:q-minus-rank-score}
        g(z)
        \coloneqq
        \mathbb{E}_{a\in\F_q^m,\ b\in\F_q^n}\Big[
            \adchar_z(ab^\top)
            \Big]
        \;=\;
        \frac{1}{q^{m+n}}\sum_{a\in\F_q^m}\sum_{b\in\F_q^n}\adchar_z(ab^\top).
    \end{equation}
    Let us consider the sum over the vector $a$ for a fixed vector $b$
    \begin{equation}
        \sum_{a\in\F_q^m}\sum_{b\in\F_q^n}\adchar_z(ab^\top)
        =
        \sum_{b\in\F_q^n}\Big(\sum_{a\in\F_q^m}\chi\big(a^\top(zb)\big)\Big).
    \end{equation}
    The inner sum equals $q^m$ if $zb=0$ and $0$ otherwise. Hence
    \begin{equation}
        \sum_{a\in\F_q^m}\sum_{b\in\F_q^n}\chi(a^\top z b)
        =
        q^m\,\abs{\ker(z)}
        =
        q^m\,q^{\,n-\rank(z)}
        =
        q^{m+n-\rank(z)}.
    \end{equation}
    where we have used $\abs{\ker(z)}=q^{\dim(\ker z)}=q^{\,n-\rank(z)}$.
    Finally, dividing by $q^{m+n}$, we obtain the first half of the proof
    \begin{equation}
        \label{eq:rank-matrix-score-function}
        q^{-\rank(z)} = g(z)
    \end{equation}

    The second part of the proof revolves around counting in
    \eqref{eq:q-minus-rank-score} how many matrices $ab^\top$ of rank $0$ and
    $1$ emerge from the summation over $(a,b)$.

    If $a=0$ or $b=0$, then $ab^\top=0$, so the summand equals $1$. The number
    of such pairs is
    \begin{equation}
        q^n+q^m-1.
    \end{equation}

    If $a\neq 0$ and $b\neq 0$, then $ab^\top$ has rank $1$. Every nonzero
    rank-one matrix $r=ab^\top$ can be obtained in exactly $q-1$ ways, since
    \begin{equation}
        \label{eq:rank-one-matrix-gauge}
        (\alpha a)(\alpha^{-1}b)^\top=ab^\top
        \qquad (\alpha\in\F_q ^\times).
    \end{equation}
    Hence
    \begin{equation}
        \sum_{a\in\F_q^m}\sum_{b\in\F_q^n}
        \adchar_z(ab^\top)
        =
        (q^m+q^n-1)
        +
        (q-1)\sum_{r \in \Gamma_1}
        \adchar_z(r).
    \end{equation}
    We finally conclude from the definition of $\lambda_1(t)$ in
    Eq.~\eqref{eq:distance-matrix-character-sum} and obtain
    \eqref{eq:q-minus-rank-character-score}.
\end{proof}

\subsection{Proof of Lemma~\ref{lem:rank-score-uniform-moments}
    (Uniform expectation and variance of the rank score)}
\label{subsec:proof-rank-score-uniform-moments}

\begin{proof}
    Let $z$ be uniform in $\F_q^{m\times n}$. From
    \eqref{eq:rank-matrix-score-function} we have the identity
    \begin{equation}\label{eq:q-minus-rank-character-expansion}
        q^{-\rank(z)}
        =\mathbb{E}_{a\in\F_q^m,\ b\in\F_q^n}\Big[
            \adchar_z(ab^\top)
            \Big].
    \end{equation}
    Taking expectation over random $z$ and swapping expectations,
    \begin{equation}
        \mathbb{E}_z\big[q^{-\rank(z)}\big]
        =
        \mathbb{E}_{a,b}\Big[\mathbb{E}_z [
                    \adchar_z(ab^\top)
                ]\Big].
    \end{equation}
    For fixed $(a,b)$, the scalar $a^\top z b=\sum_{i,j} a_i z_{ij} b_j$ is a
    linear form in the entries of $z$. If $a=0$ or $b=0$, then $a^\top z b=0$
    for all $z$, so $\mathbb{E}_z[\adchar_z(ab^\top)]=1$. If $a\neq 0$ and
    $b\neq 0$, then the linear form is nontrivial, and by additive-character
    orthogonality its average over uniform $z$ is $0$. Hence
    \begin{equation}
        \mathbb{E}_z[
                \adchar_z(ab^\top)
            ]=\indicator_{{0}}(ab^\top).
    \end{equation}
    Therefore
    \begin{equation}
        \label{eq:q-minus-rank-first-moment-proof}
        \begin{split}
            \mathbb{E}_z\big[q^{-\rank(z)}\big] & = \frac{\abs*{\Set{(a,b) \given ab^\top=0}}}{q^{n+m}} \\
                                                & =\Pr_{a,b}[ab^\top=0]                                 \\
                                                & =\Pr[a=0]+\Pr[b=0]-\Pr[a=0 \cap b=0]
            =q^{-m}+q^{-n}-q^{-(m+n)}.
        \end{split}
    \end{equation}

    As for the variance, we first compute $\mathbb{E}_z[q^{-2\rank(z)}]$. Using
    \eqref{eq:rank-matrix-score-function}, we get
    \begin{equation}
        q^{-2\rank(z)}
        =
        \mathbb{E}_{a,b}\mathbb{E}_{a',b'}
        \Big[
            \adchar_z(ab^\top+a'b'^{\!\top})
            \Big].
    \end{equation}
    Taking expectation over random $z$ and using additive-character
    orthogonality,
    \begin{equation}
        \mathbb{E}_z[q^{-2\rank(z)}]
        =
        \Pr_{a,b,a',b'}\!\big[ab^\top+a'b'^{\!\top}=0\big].
    \end{equation}

    Set $r\coloneqq ab^\top$ and $r'\coloneqq a'b'^{\!\top}$. Then $r,r'$ are
    either $0$ or rank-one, and the condition is $r'=-r$.

    Now $r=0$ if and only if $a=0$ or $b=0$, hence $\abs*{\Set{(a,b) \given
                ab^\top=0}}=q^m+q^n-1$. Therefore the number of quadruples with $r=r'=0$ is
    $(q^m+q^n-1)^2$.

    As explained in \eqref{eq:rank-one-matrix-gauge}, every nonzero rank-one
    matrix $r$ has exactly $q-1$ factorisations and thus the number of nonzero
    rank-one matrices is $\frac{(q^m-1)(q^n-1)}{q-1}$. For each such $r$, there
    are $q-1$ choices of $(a,b)$ with $ab^\top=r$, and $q-1$ choices of
    $(a',b')$ with $a'b'^{\!\top}=-r$. Thus the number of quadruples with $r\neq
        0$ and $r'=-r$ is

    \begin{equation}
        \frac{(q^m-1)(q^n-1)}{q-1}\,(q-1)^2
        =
        (q-1)(q^m-1)(q^n-1).
    \end{equation}

    Combining the two contributions,
    \begin{equation}
        \abs*{\Set{(a,b,a',b') \given ab^\top+a'b'^{\!\top}=0}}
        =
        (q^m+q^n-1)^2+(q-1)(q^m-1)(q^n-1).
    \end{equation}
    and dividing by $q^{2m+2n}$ gives
    $\Pr_{a,b,a',b'}\!\big[ab^\top+a'b'^{\!\top}=0\big]$ and thus the second
    moment
    \begin{equation}
        \label{eq:q-minus-rank-second-moment-proof}
        \mathbb{E}_z[q^{-2\rank(z)}]
        =
        \bigl(q^{-m}+q^{-n}-q^{-(m+n)}\bigr)^2
        +
        (q-1)\,q^{-(m+n)}(1-q^{-m})(1-q^{-n}).
    \end{equation}
    Combining the second moment with the average from
    Eq.~\eqref{eq:q-minus-rank-first-moment-proof}, the variance finally reads
    \begin{equation}
        \mathrm{Var}_z\!\big[q^{-\rank(z)}\big]
        =
        (q-1)\,q^{-(m+n)}(1-q^{-m})(1-q^{-n}).
    \end{equation}
\end{proof}

\subsection{Proof of Lemma~\ref{lem:rank-metric-proxy-score-observable-moments}
    (Expectation and variance of the rank-metric proxy-score observable)}
\label{subsec:proof-rank-metric-proxy-score-observable-moments}

\begin{proof}
    Let
    \begin{equation}
        \mathcal{L}\coloneqq
        \sum_{x\in\Omega}
        \lambda_1(\rank(\Delta_y(x)))\ket{x}\bra{x}
    \end{equation}
    and
    \begin{equation}
        \gamma\coloneqq
        \frac{1}{\sqrt{\abs{\Gamma_1}}} =
        \sqrt{\frac{q-1}{(q^m-1)(q^n-1)}}.
    \end{equation}
    Then $\mathcal{O}_q=\gamma \mathcal{L}$. Since $\Phi$ identifies $\Omega$
    bijectively with $\Code$, we have $\abs{\Omega}=\abs{\Code}$. Moreover, $2
        \ell < d^\perp$, and so Lemma~\ref{lem:shell-orthogonality-dual-distance}
    implies
    \begin{equation}
        \braket{\psi_6}{\psi_6}  = \sum_{k=0}^\ell \abs{w_k}^2 =1.
    \end{equation}

    We first compute the expectation of $\mathcal{L}$. Inserting the rank-metric
    specialisation of $\ket{\psi_6}$ from Eq.~\eqref{eq:dqi-final-state}, we
    obtain
    \begin{equation}\label{eq:dqi-observable-expectation-proof}
        \begin{split}
            \bra{\psi_6}\mathcal{L}\ket{\psi_6}
             & =\bra{\psi_6} \sum_{x\in \Omega} \lambda_1(\rank(\Delta_y(x)))\ket{x}\bra{x} \ket{\psi_6}                                                                                                                                              \\
             & = \sum_{k,k^\prime=0}^\ell \frac{w_k w^*_{k^\prime}}{\abs{\Omega}\sqrt{\abs{\Gamma_k}\abs{\Gamma_{k^\prime}}}} \sum_{x\in \Omega} \lambda_1(\rank(\Delta_y(x)))\lambda_k(\rank(\Delta_y(x)))\lambda^*_{k^\prime}(\rank(\Delta_y(x)))   \\
             & = \sum_{k,k^\prime=0}^\ell \frac{w_k w^*_{k^\prime}}{\abs{\Omega}\sqrt{\abs{\Gamma_k}\abs{\Gamma_{k^\prime}}}} \sum_{c\in \Code} \lambda_1(\rank(c-y))\lambda_k(\rank(c-y))\lambda^*_{k^\prime}(\rank(c-y))                            \\
             & = \sum_{k,k^\prime=0}^\ell \frac{w_k w^*_{k^\prime}}{\abs{\Omega}\sqrt{\abs{\Gamma_k}\abs{\Gamma_{k^\prime}}}}                                                                                                                         \\
             & \times \sum_{c\in \Code} \left(b_{k-1}\lambda_{k-1}(\rank(c-y)))+a_k \lambda_k(\rank(c-y)) + c_{k+1}\lambda_{k+1}(\rank(c-y))\right)\lambda^*_{k^\prime}(\rank(c-y))                                                                   \\
             & = \frac{\abs{\Code}}{\abs{\Omega}} \sum_{k,k^\prime=0}^\ell w_k w^*_{k^\prime}\sqrt{\frac{\abs{\Gamma_{k^\prime}}}{\abs{\Gamma_{k}}}}\left(b_{k-1} \delta_{k-1,k^\prime}+a_k \delta_{k,k^\prime} + c_{k+1}\delta_{k+1,k^\prime}\right) \\
             & = w^\dagger \tilde{A}^{(\ell)} w .
        \end{split}
    \end{equation}
    Here we used $\Code=\Im(\Phi)=\Phi(\Omega)$, the three-term character
    recurrence \eqref{eq:character-sum-three-term-recurrence}, the low-shell
    orthogonality lemma~\ref{lem:shell-orthogonality-dual-distance}, and finally
    $\abs{\Code}=\abs{\Omega}$. Multiplying by $\gamma$ gives
    Eq.~\eqref{eq:dqi-observable-expectation-jacobi}.

    It remains to compute the second moment of $\mathcal{L}$:
    \begin{equation}\label{eq:dqi-observable-square-expectation}
        \begin{split}
            \mathbb{E}_{\psi_6}[\mathcal{L}^2]
             & \coloneqq
            \bra{\psi_6}\mathcal{L}^2\ket{\psi_6}
            =
            \bra{\psi_6} \sum_{x\in \Omega} \lambda_1(\rank(\Delta_y(x)))^2\ket{x}\bra{x} \ket{\psi_6}                        \\
             & = \sum_{k,k^\prime=0}^\ell \frac{w_k w^*_{k^\prime}}{\abs{\Omega}\sqrt{\abs{\Gamma_k}\abs{\Gamma_{k^\prime}}}}
            \sum_{x\in \Omega} \lambda_1(\rank(\Delta_y(x)))^2
            \lambda_k(\rank(\Delta_y(x)))\lambda^*_{k^\prime}(\rank(\Delta_y(x)))                                             \\
             & = \sum_{k,k^\prime=0}^\ell \frac{w_k w^*_{k^\prime}}{\abs{\Omega}\sqrt{\abs{\Gamma_k}\abs{\Gamma_{k^\prime}}}}
            \sum_{c\in \Code} \lambda_1(\rank(c-y))^2
            \lambda_k(\rank(c-y))\lambda^*_{k^\prime}(\rank(c-y))                                                             \\
             & =
            \sum_{k,k^\prime=0}^\ell \frac{w_k w^*_{k^\prime}}{\abs{\Omega}\sqrt{\abs{\Gamma_k}\abs{\Gamma_{k^\prime}}}}
            \sum_{c\in \Code}     \Bigl(
            b_{k-1}\lambda_{k-1}(\rank(c-y))
            +a_k\lambda_k(\rank(c-y))
            +c_{k+1}\lambda_{k+1}(\rank(c-y))
            \Bigr)                                                                                                            \\ &\quad\times \lambda_1(\rank(c-y))\lambda^*_{k^\prime}(\rank(c-y))\\
             & =
            \sum_{k,k^\prime=0}^\ell \frac{w_k w^*_{k^\prime}}{\abs{\Omega}\sqrt{\abs{\Gamma_k}\abs{\Gamma_{k^\prime}}}}
            \sum_{c\in \Code}
            \Bigl[
                b_{k-1}b_{k-2}\lambda_{k-2}
                +b_{k-1}(a_{k-1}+a_k)\lambda_{k-1} +(b_{k-1}c_k+a_k^2+c_{k+1}b_k)\lambda_k\\
                &\quad
                +c_{k+1}(a_k+a_{k+1})\lambda_{k+1}
                +c_{k+1}c_{k+2}\lambda_{k+2}
                \Bigr](\rank(c-y))
            \lambda^*_{k^\prime}(\rank(c-y))                                                                                  \\
             &
            =\frac{\abs{\Code}}{\abs{\Omega}}
            \sum_{k,k^\prime=0}^\ell
            w_k w_{k^\prime}^*
            \sqrt{\frac{\abs{\Gamma_{k^\prime}}}{\abs{\Gamma_k}}}
            \Bigl[
                b_{k-1}b_{k-2}\delta_{k-2,k^\prime}
                +b_{k-1}(a_{k-1}+a_k)\delta_{k-1,k^\prime}\\
                &\quad
                +(b_{k-1}c_k+a_k^2+c_{k+1}b_k)\delta_{k,k^\prime}
                +c_{k+1}(a_k+a_{k+1})\delta_{k+1,k^\prime}+c_{k+1}c_{k+2}\delta_{k+2,k^\prime}
                \Bigr].
        \end{split}
    \end{equation}

    Compared to the derivation \eqref{eq:dqi-observable-expectation-proof}, we
    have applied the three-term recurrence
    \eqref{eq:character-sum-three-term-recurrence} twice, explaining why terms
    $\Set{k-2,k-1,k,k+1,k+2}$ do arise. The orthogonality of the code sum
    lemma~\ref{lem:shell-orthogonality-dual-distance} can be invoked if
    $k+k^\prime+2 \leq 2\ell +2 < d^\perp$ which strengthens the constraint of
    the average.

    The pentadiagonal matrix arising in the last line of
    \eqref{eq:dqi-observable-square-expectation} is almost the square of
    $\tilde{A}^{(\ell)}$ with a subtlety in the bottom right coefficient
    $(\ell+1,\ell+1)$ and we find
    \begin{equation}
        \label{eq:dqi-observable-square-tridiagonal}
        \mathbb{E}_{\psi_6}[\mathcal{L}^2]
        = w^\dagger \left(\tilde{A}^{(\ell)}\right)^2 w +\abs{w_\ell}^2 b_{\ell}c_{\ell+1}
    \end{equation}
    Combining Eq.~\eqref{eq:dqi-observable-square-tridiagonal} with
    Eq.~\eqref{eq:dqi-observable-expectation-proof} and $\mathcal{O}_q=\gamma
        \mathcal{L}$, the variance of $\mathcal{O}_q$ reads
    \begin{equation}
        \operatorname{Var}_{\psi_6}[\mathcal{O}_q]
        =
        \frac{1}{\abs{\Gamma_1}}
        \left(
        w^\dagger \left(\tilde{A}^{(\ell)}\right)^2 w
        - \left(w^\dagger \tilde{A}^{(\ell)} w\right)^2
        + \abs{w_\ell}^2 b_{\ell}c_{\ell+1}
        \right).
    \end{equation}

\end{proof}

\subsection{Proof of Lemma~\ref{lem:truncated-rank-jacobi-spectral-bounds}
    (Spectral bounds for the truncated rank Jacobi operator)}
\label{subsec:proof-truncated-rank-jacobi-spectral-bounds}

\begin{proof}
    We first use interlacing to obtain the upper bound. The matrix $\tilde
        A^{(\ell)}$ is the $(\ell+1)\times(\ell+1)$ principal submatrix of $\tilde
        A$. The eigenvalues of $\tilde A$ are the decreasing sequence $\lambda_1(t)$
    for $t\in[0,D]$, with
    \begin{equation}
        \lambda_1(t) = \frac{q^{m+n-t}-q^m-q^n+1}{q-1}.
    \end{equation}
    By Cauchy's interlacing theorem and
    Eq.~\eqref{eq:jacobi-ladder-factorisation},
    \begin{equation}
        \lambda_1(D-\ell)
        \leq
        \abs{\Gamma_1}
        -
        \frac{q^{m+n}}{q-1} \mu^{(\ell)}_{\min}
        \leq
        \lambda_1(0).
    \end{equation}
    Since $\abs{\Gamma_1}=\lambda_1(0)$, this is equivalent to
    \begin{equation}
        0 \leq \mu^{(\ell)}_{\min} \leq 1-q^{\ell-D}.
    \end{equation}
    Interlacing also gives $\mu_{\min}^{(\ell)}\ge\mu_{\min}^{(\ell+1)}$, so
    increasing $\ell$ can only decrease the optimised effective-rank proxy.

    For the lower bound, recall the exact coefficients defining the tridiagonal
    matrix $\tilde{A}$ from Eqs.~\eqref{eq:rank-b-coefficient},
    \eqref{eq:rank-c-coefficient}:
    \begin{equation}
        b_i
        =
        \frac{(q^m-q^i)(q^n-q^i)}{q-1},
        \qquad
        c_i
        =
        \frac{(q^i-1)q^{i-1}}{q-1}.
    \end{equation}
    If $\ell=0$, then $B^{(0)}=(\sqrt{b_0})$ and hence
    \begin{equation}
        \mu_{\min}^{(0)}
        =
        \frac{q-1}{q^{m+n}}b_0
        =
        (1-q^{-m})(1-q^{-n}),
    \end{equation}
    which is exactly the claimed lower bound. We therefore assume $\ell\ge1$.
    Denote
    \begin{equation}
        \alpha
        \coloneqq
        \min_{0\le i\le \ell}\sqrt{b_i},
        \qquad
        \beta
        \coloneqq
        \max_{1\le i\le \ell}\sqrt{c_i}.
    \end{equation}
    Writing $B^{(\ell)}$ as its diagonal part minus its strictly upper
    bidiagonal part, the diagonal part has smallest singular value at least
    $\alpha$, and the strictly upper bidiagonal part has operator norm at most
    $\beta$. Hence, by the reverse triangle inequality, for every $w$,
    \begin{equation}
        \|B^{(\ell)}w\|_2
        \ge
        (\alpha-\beta)\|w\|_2.
    \end{equation}
    Since $b_i$ is decreasing and $c_i$ is increasing in $i$, we have
    $\alpha=\sqrt{b_\ell}$ and $\beta=\sqrt{c_\ell}$. Moreover,
    $\ell<D=\min(m,n)$ implies
    \begin{equation}
        b_\ell-c_\ell
        =
        \frac{(q^m-q^\ell)(q^n-q^\ell)-(q^\ell-1)q^{\ell-1}}{q-1}.
    \end{equation}
    Since $m,n\ge\ell+1$ and $q\ge2$, the numerator is at least
    \begin{equation}
        q^{2\ell}(q-1)^2-(q^\ell-1)q^{\ell-1}
        =
        q^{\ell-1}\bigl(q^\ell(q(q-1)^2-1)+1\bigr)
        >0,
    \end{equation}
    so $\alpha>\beta$. Therefore
    \begin{equation}
        \mu_{\min}^{(\ell)}
        \ge
        \frac{q-1}{q^{m+n}}(\alpha-\beta)^2.
    \end{equation}
    Thus
    \begin{equation}
        \mu^{(\ell)}_{\min}\ge
        \left(
        \sqrt{(1-q^{\ell-m})(1-q^{\ell-n})}
        -
        \sqrt{(q^\ell-1)q^{\,\ell-1-m-n}}
        \right)^2.
    \end{equation}
    It remains only to verify $\alpha-\beta\ge0$. Using the formulas for
    $b_\ell$ and $c_\ell$,
    \begin{equation}
        \frac{q-1}{q^{m+n}}(b_\ell-c_\ell)
        >
        1-q^{\ell-m}-q^{\ell-n}.
    \end{equation}
    Since $\ell<D=\min(m,n)$, we have $q^{\ell-n}+q^{\ell-m}\le 2/q\le1$.
    Therefore $\alpha-\beta\ge0$.
\end{proof}

\subsection{Proof of Lemma~\ref{lem:boundary-concentration-optimised-rank-shell-weights}
    (Boundary concentration of optimised rank-shell weights)}
\label{subsec:proof-boundary-concentration-optimised-rank-shell-weights}

\begin{proof}
    It is enough to treat the case $n\le m$; the case $m<n$ follows by
    transposing matrices. Set $D=n$, $s=m-n$, and $\varepsilon_\ell\coloneqq
        q^{\ell-D}$. We work with the normalised ladder Hamiltonian
    \begin{equation}
        \bar L^{(\ell)}
        \coloneqq
        \frac{q-1}{q^{m+n}}(B^{(\ell)})^\dagger B^{(\ell)}.
    \end{equation}
    Its coefficients are
    \begin{equation}
        \bar b_i\coloneqq \frac{q-1}{q^{m+n}}b_i
        =
        (1-q^{i-m})(1-q^{i-n}),
        \qquad
        \bar c_i\coloneqq \frac{q-1}{q^{m+n}}c_i
        =
        (q^i-1)q^{i-1-m-n}.
    \end{equation}
    Hence $\bar L^{(\ell)}$ has diagonal entries $\bar b_i+\bar c_i$ and
    off-diagonal entries $-\sqrt{\bar b_i\bar c_{i+1}}$.

    Let $e_r$, $r\ge0$, be the standard basis of $\ell^2(\mathbb N_0)$, and let
    $P_\ell$ be the projection onto
    $\operatorname{span}\Set{e_0,\ldots,e_\ell}$. On this finite-dimensional
    subspace define the reversal unitary $U_\ell e_i=e_{\ell-i}$. The next
    operator is the zoomed-in edge model: $I-\bar L^{(\ell)}$ measures the
    deficit from the limiting value $1$, the reversal moves the cutoff shell to
    $r=0$, and the factor $\varepsilon_\ell^{-1}$ keeps the boundary entries
    finite. We consider the boundary-scaled operator
    \begin{equation}
        T_\ell
        \coloneqq
        \varepsilon_\ell^{-1}
        P_\ell(I-U_\ell\bar L^{(\ell)}U_\ell^\dagger)P_\ell,
    \end{equation}
    extended by zero on $(I-P_\ell)\ell^2(\mathbb N_0)$. If $w^{(\ell)}$ is a
    normalised eigenvector of $\bar L^{(\ell)}$ with eigenvalue
    $\mu_{\min}^{(\ell)}$, then $u^{(\ell)}=U_\ell w^{(\ell)}$ is a normalised
    eigenvector of $T_\ell$ with eigenvalue
    $\varepsilon_\ell^{-1}(1-\mu_{\min}^{(\ell)})$.

    In the reversed coordinate $r=\ell-i$, the nonzero entries of $T_\ell$ are
    \begin{equation}
        (T_\ell)_{rr}
        =
        \varepsilon_\ell^{-1}
        (1-\bar b_{\ell-r}-\bar c_{\ell-r})
    \end{equation}
    and, for $1\le r\le \ell$,
    \begin{equation}
        (T_\ell)_{r,r-1}
        =
        \varepsilon_\ell^{-1}
        \sqrt{\bar b_{\ell-r}\bar c_{\ell-r+1}},
        \qquad
        (T_\ell)_{r-1,r}=(T_\ell)_{r,r-1}.
    \end{equation}
    For each fixed $r$ and fixed $s=m-n$, take any sequence with
    $\ell\to\infty$ and $D-\ell\to\infty$. Then
    $\varepsilon_\ell=q^{\ell-D}\to0$, and the exact coefficient formulas give
    \begin{equation}
        \varepsilon_\ell^{-1}
        (1-\bar b_{\ell-r}-\bar c_{\ell-r})
        =
        (1+q^{-s})q^{-r}+o(1),
    \end{equation}
    and
    \begin{equation}
        \varepsilon_\ell^{-1}
        \sqrt{\bar b_{\ell-r}\bar c_{\ell-r+1}}
        =
        q^{-(s+1)/2}q^{-(r-1)} + o(1).
    \end{equation}
    By symmetry, the coefficient multiplying $u_{r+1}$ in row $r$ converges to
    $q^{-(s+1)/2}q^{-r}$. Moreover, after extending $T_\ell$ by zero for
    $r>\ell$, the same formulas are dominated by summable row and column bounds
    of order $q^{-r}$, and the tail for $r>\ell$ is $O(q^{-\ell})$. Schur's test
    therefore gives norm convergence: one first fixes a finite boundary window,
    where the entries converge uniformly, and then lets the window size grow,
    using the uniform tail bound. Thus
    \begin{equation}
        T_\ell\to \mathcal K_s
        \qquad
        \text{on } \ell^2(\mathbb N_0),
    \end{equation}
    where $\mathcal K_s$ is the boundary operator defined in
    Sec.~\ref{subsec:concentration-truncation-optimised-rank-shell-weights}.

    The row sums of $\mathcal K_s$ are $O(q^{-r})$, so $\mathcal K_s$ is the
    norm limit of its finite truncations and hence compact. It is self-adjoint,
    preserves the positive cone, and the positive off-diagonal entries connect
    every neighbouring pair of basis vectors, making it irreducible. Jentzsch's
    infinite-dimensional Perron-Frobenius theorem implies that its top
    eigenvalue $\kappa_s$ is simple and has a normalised strictly positive
    eigenvector $u^{(s)}$. For each $\ell$, the eigenvalue
    $\varepsilon_\ell^{-1}(1-\mu_{\min}^{(\ell)})$ is the largest eigenvalue of
    $T_\ell$, because $\mu_{\min}^{(\ell)}$ is the smallest eigenvalue of $\bar
        L^{(\ell)}$. Since $\kappa_s$ is isolated, standard perturbation theory for
    self-adjoint compact operators gives
    \begin{equation}
        \varepsilon_\ell^{-1}(1-\mu_{\min}^{(\ell)})\to\kappa_s,
        \qquad
        u^{(\ell)}\to u^{(s)}
        \quad\text{in }\ell^2(\mathbb N_0).
    \end{equation}
    This proves both the boundary-profile convergence and
    \begin{equation}
        \mu_{\min}^{(\ell)}
        =
        1-\kappa_s q^{\ell-D}+o(q^{\ell-D}).
    \end{equation}

    It remains to quantify how much mass lies away from the boundary, uniformly
    in the boundary-window width. Let $P_J$ project onto the first $J+1$ basis
    vectors and write $Q_J=I-P_J$. The entry estimates above give a constant
    $C_s'>0$ such that, for all sufficiently large $\ell$ and all $0\le
        J\le\ell$,
    \begin{equation}
        \|Q_JT_\ell Q_J\|+\|Q_JT_\ell P_J\|
        \le
        C_s' q^{-J}.
    \end{equation}
    Indeed, in reversed coordinates the entries of $T_\ell$ are tridiagonal and
    bounded by $C_s' q^{-r}$ in row and column $r$, uniformly in $\ell$.
    Applying $Q_J$ to $T_\ell u^{(\ell)}=\kappa_\ell u^{(\ell)}$ gives
    \begin{equation}
        (\kappa_\ell I-Q_JT_\ell Q_J)Q_Ju^{(\ell)}
        =
        Q_JT_\ell P_Ju^{(\ell)}.
    \end{equation}
    Since $\kappa_\ell\to\kappa_s>0$, for all sufficiently large $\ell$ we have
    $\kappa_\ell\ge\kappa_s/2$. For all $J$ above a constant depending only on
    $q$ and $s$, the inverse on the left has norm at most $4/\kappa_s$.
    Enlarging the constant to cover the remaining finite values of $J$, we
    obtain, uniformly for $0\le J\le\ell$,
    \begin{equation}
        \|Q_Ju^{(\ell)}\|_2
        \le
        C_s q^{-J}.
    \end{equation}
    Therefore
    \begin{equation}
        \sum_{i=0}^{\ell-J-1}\abs*{w_i^{(\ell)}}^2
        =
        \sum_{r>J}\abs*{u_r^{(\ell)}}^2
        =
        \|Q_Ju^{(\ell)}\|_2^2
        \le
        C_s q^{-2J},
    \end{equation}
    after increasing $C_s$ if necessary.

    Finally suppose that $s=m-n\to\infty$. The same entry estimates show
    directly that $T_\ell$ converges in norm to the diagonal compact operator
    \begin{equation}
        \mathcal K_\infty e_r=q^{-r}e_r,
    \end{equation}
    whenever $s\to\infty$ and $\varepsilon_\ell\to0$. Indeed, the off-diagonal
    entries are $O(q^{-(s+1)/2}q^{-r})$, the diagonal correction from $q^{-r}$
    is $O(q^{-s}q^{-r})$, and the row and column tails remain uniformly
    $O(q^{-r})$. The top eigenvalue of $\mathcal K_\infty$ is $1$, its top
    eigenvector is $e_0$, and the spectral gap is $1-q^{-1}$. Perturbation
    theory therefore gives $u^{(\ell)}\to e_0$. Equivalently,
    $\abs*{w_\ell^{(\ell)}}^2\to1$.
\end{proof}

\subsection{Proof of Corollary~\ref{cor:effective-rank-proxy-boundary-truncation}
    (Effective-rank proxy stability under boundary truncation)}
\label{subsec:proof-effective-rank-proxy-boundary-truncation}

\begin{proof}
    We use the notation from the proof of
    Lemma~\ref{lem:boundary-concentration-optimised-rank-shell-weights}. Let
    $P_J$ project onto $\operatorname{span}\Set{e_0,\ldots,e_J}$ in the reversed
    boundary coordinate, and set
    \begin{equation}
        \widehat u^{(\ell,J)}
        \coloneqq
        \frac{P_Ju^{(\ell)}}{\|P_Ju^{(\ell)}\|_2}.
    \end{equation}
    This is the reversed-coordinate version of $\widehat w^{(\ell,J)}$. The
    denominator is nonzero for all sufficiently large $\ell$, since
    $u^{(\ell)}\to u^{(s)}$ and the limiting vector has strictly positive first
    coordinate. For a normalised radial vector $v$, write
    \begin{equation}
        \mu(v)
        \coloneqq
        \frac{q-1}{q^{m+n}}\|B^{(\ell)}v\|_2^2.
    \end{equation}
    Since $T_\ell=\varepsilon_\ell^{-1} P_\ell(I-U_\ell\bar
        L^{(\ell)}U_\ell^\dagger)P_\ell$, we have
    \begin{equation}
        1-\mu(v)
        =
        \varepsilon_\ell
        \bra{U_\ell v}T_\ell\ket{U_\ell v}.
    \end{equation}
    Therefore, with
    \begin{equation}
        \kappa_\ell
        \coloneqq
        \bra{u^{(\ell)}}T_\ell\ket{u^{(\ell)}},
        \qquad
        \tau_{\ell,J}
        \coloneqq
        \bra{\widehat u^{(\ell,J)}}T_\ell\ket{\widehat u^{(\ell,J)}},
    \end{equation}
    the effective-rank proxy difference is
    \begin{equation}
        R_\ell(\widehat w^{(\ell,J)})-R_\ell(w^{(\ell)})
        =
        \log_q\!\left(\frac{\kappa_\ell}{\tau_{\ell,J}}\right).
    \end{equation}

    It remains to bound the loss in the Rayleigh quotient uniformly in $J$. Set
    $Q_J=I-P_J$. The coefficient estimates used in the proof of
    Lemma~\ref{lem:boundary-concentration-optimised-rank-shell-weights} give a
    constant $C_s>0$ such that, for all sufficiently large $\ell$ and all $0\le
        J\le\ell$,
    \begin{equation}
        \|Q_JT_\ell Q_J\|+\|Q_JT_\ell P_J\|
        \le
        C_s q^{-J}.
    \end{equation}
    Indeed, in reversed coordinates the entries of $T_\ell$ are tridiagonal and
    bounded by $C_s q^{-r}$ in row and column $r$, uniformly in $\ell$.

    Since $\kappa_\ell\to\kappa_s>0$, we may assume $\kappa_\ell\ge\kappa_s/2$.
    Applying $Q_J$ to $T_\ell u^{(\ell)}=\kappa_\ell u^{(\ell)}$ gives
    \begin{equation}
        (\kappa_\ell I-Q_JT_\ell Q_J)Q_Ju^{(\ell)}
        =
        Q_JT_\ell P_Ju^{(\ell)}.
    \end{equation}
    For all $J$ above a constant depending only on $q$ and $s$, the inverse on
    the left has norm at most $4/\kappa_s$. For the remaining finitely many $J$,
    the trivial bound $\|Q_Ju^{(\ell)}\|_2\le1$ is absorbed into the constant.
    Thus, uniformly for $0\le J\le\ell$,
    \begin{equation}
        \|Q_Ju^{(\ell)}\|_2
        \le
        C_s q^{-J}.
    \end{equation}

    Write $p=P_Ju^{(\ell)}$ and $q_J=Q_Ju^{(\ell)}$. Taking the inner product of
    $T_\ell u^{(\ell)}=\kappa_\ell u^{(\ell)}$ with $p$ gives
    \begin{equation}
        \bra{p}T_\ell\ket{p}
        +
        \bra{p}T_\ell\ket{q_J}
        =
        \kappa_\ell\|p\|_2^2.
    \end{equation}
    Therefore
    \begin{equation}
        \kappa_\ell-\tau_{\ell,J}
        =
        \frac{\bra{p}T_\ell\ket{q_J}}{\|p\|_2^2}.
    \end{equation}
    The denominator is uniformly bounded away from zero, since
    $\|p\|_2\ge\abs*{u_0^{(\ell)}}$ and $u^{(\ell)}\to u^{(s)}$ with
    $u_0^{(s)}>0$. Because $u^{(\ell)}$ is the top eigenvector of $T_\ell$, we
    also have $\tau_{\ell,J}\le\kappa_\ell$. Hence
    \begin{equation}
        0
        \le
        \kappa_\ell-\tau_{\ell,J}
        \le
        C_s q^{-2J}
    \end{equation}
    uniformly for $0\le J\le\ell$.

    Finally, $\tau_{\ell,J}$ is uniformly bounded below. The entries of $T_\ell$
    are nonnegative: the diagonal entries are scaled intersection numbers and
    the off-diagonal entries are positive square roots. Moreover, $u^{(\ell)}$,
    and hence $\widehat u^{(\ell,J)}$, is nonnegative,
    $(T_\ell)_{00}\to1+q^{-s}$, and the normalised vector $\widehat
        u^{(\ell,J)}$ has first coordinate at least $u_0^{(\ell)}$. Thus
    $\tau_{\ell,J}\ge c_s>0$ for all sufficiently large $\ell$ and all $0\le
        J\le\ell$. The logarithm is therefore Lipschitz on the relevant interval,
    and
    \begin{equation}
        \log_q\!\left(\frac{\kappa_\ell}{\tau_{\ell,J}}\right)
        \le
        C_s' q^{-2J}.
    \end{equation}
    This proves the claim.
\end{proof}

\end{document}